\documentclass[a4paper]{article}
\usepackage[left=3cm,right=3cm]{geometry}

\usepackage{booktabs}   
\usepackage{subcaption} 

\usepackage{amsmath}
\usepackage{amsthm}
\usepackage{amsfonts}
\usepackage{amssymb}
\usepackage{cmll}
\usepackage{bm}

\usepackage{tikz}
\usetikzlibrary{matrix}
\usetikzlibrary{cd}
\usepackage{pgfplots}
\usepackage{stmaryrd}
\usepackage{ebproof}
\usepackage{url}

\usepackage{varioref}
\usepackage{hyperref}
\usepackage{cleveref}

\newcommand\CMLLPAR{
\usepackage{cmll}
\newcommand\IPar{\mathord{\parr}}
}

\CMLLPAR
\input{mainnot}

\title{From Differential Linear Logic to Coherent Differentiation}
\author{Thomas Ehrhard\\
Université Paris Cité, CNRS, Inria, IRIF, F-75013, Paris, France}

\begin{document}

\maketitle

\begin{abstract}
In this survey, we present in a unified way the categorical and
syntactical settings of coherent differentiation introduced recently,
which shows that the basic ideas of differential linear logic and of
the differential lambda-calculus are compatible with determinism.
Indeed, due to the Leibniz rule of the differential calculus,
differential linear logic and the differential lambda-calculus feature
an operation of addition of proofs or terms operationally interpreted
as a strong form of nondeterminism.
The main idea of coherent differentiation is that these sums can be
controlled and kept in the realm of determinism by means of a notion
of summability, upon enforcing summability restrictions on the
derivatives which can be written in the models and in the syntax.


\end{abstract}

\section*{Introduction}


During his development of the denotational semantics of \SystF{} in
the cartesian closed category of qualitative domains, and more
specifically of coherence spaces, Girard observed in~\cite{Girard86}
that a specific class of morphisms arises naturally among the general
morphisms of this model (the stable maps).
These particular stable functions are characterized by an additional
preservation property (they commute with compatible unions).

Girard not only recognized the relevance of these morphisms to the
description of the stable semantics, he also understood that they
should play a more fundamental role than the general stable functions
themselves.
He called them \emph{linear maps} because they collectively behave
very much like linear maps in algebra, forming a symmetric monoidal
category which is even \Staraut{} (a general categorical notion
introduced in~\cite{Barr79} describing categories of linear morphisms
where all objects are reflexive, that is, canonically isomorphic to
their bidual).

He understood that this observation, \emph{a priori} relative to the
stable semantics, was the denotational shadow of a fundamental and
hitherto hidden structure of \Intlog{} itself: Linear Logic (\LL),
see~\cite{Girard87}.
This essential discovery had a major impact in Logic and Computer
Science, notably on the study and design of programming languages.

One aspect of linearity which is not directly addressed in \LL{}
---~although it is contemplated in the concluding section
of~\cite{Girard87}~--- is its central role in the differential
calculus where differentiation consists in extracting from a morphism
its ``best linear approximation''.
The purpose of Differential \LL{} (\DILL{} for short) is to take this
role of linearity into account.
This logical system was introduced by the author and
Laurent~Regnier %
in~\cite{EhrhardRegnier02,EhrhardRegnier06d} and is summarized
in~\cite{Ehrhard18}.
It extends \LL{} without adding new connectives, but by adding a set
of new deduction rules that we can classify as follows.
\begin{itemize}
\item There are 3 rules relative to the exponential modality \(\oc\_\),
  dual to the standard rules of weakening, contraction and
  dereliction.
\item And there are two rules expressing that any finite family of
  proofs of the same formula has a sum which is again a proof of that
  formula.
  This includes the case of an empty family, meaning that any formula
  is provable by a \(0\) proof.
\end{itemize}
These latter rules mean that the proofs in \(\DILL\) are essentially
partial (the \(0\) proof is very similar to the term \(\Omega\) in the
theory of Böhm trees, which is a completely undefined term).
More importantly, the unrestricted ability of adding proofs of the
same formula means that \DILL{} features a fundamental non-determinism
in the sense that it is possible to add the two normal proofs of
\(\Plus\One\One\) which corresponds to the type of booleans, that is,
intuitively, the two booleans \(\True\) and \(\False\).
The most natural operational understanding of this ``boolean''
\(\True+\False:\Plus\One\One\) consists in considering it as a program
which can nondeterministically reduce to \(\True\) or \(\False\).
We provide in \Cref{sec:difflam} a more precise description of the \(\lambda\)-calculus account of this extension of \LL{}, which features a typing rule allowing to add any two terms having the same type.


\subsection*{Determinism and differentiation}
In \cite{EhrhardRegnier06a,EhrhardRegnier06b} we developed a Taylor
expansion of \(\lambda\)-terms which is based on the differential
\(\lambda\)-calculus.
This expansion consists in translating a \(\lambda\)-term \(M\) into a
(generally infinite) sum \(M^\ast\) of ``resource terms'', that is, of
differential \(\lambda\)-terms whose only use of the standard
\(\lambda\)-calculus operation of application is application to \(0\).
In other words, all the applications \(\App PQ\) are replaced
hereditarily by differential terms:
\begin{align*}
  (\App PQ)^\ast=
  \sum_{n=0}^\infty\frac 1{\Factor n}
  \App{\Diffsymb^nP^\ast\cdot
  (\overbrace{Q^\ast,\dots,Q^\ast}^n)}0\,.
\end{align*}
In this expression, if \(M:\Timpl AB\) and \(N_1,\dots,N_n:A\) are
differential \(\lambda\)-term, then
\[\Diffsymb^nM\cdot (N_1,\dots,N_n):\Timpl AB\] is the \(n\)th
derivative of \(M\) applied to linear argument \(\List N1n\), see
\Cref{sec:difflam}.

Analyzing the execution of standard \(\lambda\)-terms in the Krivine
Machine we showed that, if a standard \(\lambda\)-term \(M\) is
\(\beta\)-equivalent to a variable \(x\) then there is exactly one
term \(s\) in the Taylor expansion \(M^\ast\) of \(M\) which does not
reduce to \(0\), and this term reduces to \(\Factor n\,x\) if we take
multiplicities into account.
This resource term \(s\) is a trace of the execution of \(M\) in the
machine, or more precisely, \(s\) provides a precise account of the
quantitative use of the various subterms of \(M\) during the
execution.

In other words, in the infinite sum \(M^\ast\), only one term is
non-zero, although for knowing which one, one needs to reduce \(M\) to
its normal form.
This means that this infinite sum is only apparently nondeterministic.
But for proving this property we strongly use the fact that \(M\) is a
\emph{standard} \(\lambda\)-term, that is, contains no differential
construct of shape \(\Diffsymb P\cdot Q\).

So beyond this first encouraging observation, the question remained
open of whether differentiation (in the sense of \(\DILL\)) can be
made compatible with determinism.
Towards a positive answer to this question, the second crucial
observation was that, in the model of \LL{} based on probabilistic
coherence spaces (PCS), nonlinear morphisms are functions which are
analytic in the sense that they are defined by powerseries and hence
should be differentiable, in spite of the fact that, in the
corresponding category \(\PCOH\), it is not always possible to add two
morphisms of the same type.
For instance, the sum of two subprobability distributions on the
natural numbers is not always a subprobability distribution (the
global mass can become \(>1\)).
%
The simplest interesting object of this category is the closed unit
interval \(\Intercc 01\) that we denote as \(\Sone\), and a morphism
\(\Sone\to\Sone\) is a function \(f:\Intercc 01\to\Intercc 01\) such
that there is a (necessarily uniquely defined) sequence
\((t_n)_{n\in\Nat}\) of nonnegative real numbers such that
\(\forall x\in\Intercc01\ f(x)=\sum_{n=0}^\infty
t_nx^n\in\Intercc01\), a condition which simply means
\(\sum_{n=0}^\infty t_n\leq 1\).
Such a function has a derivative, defined on \(\Interco 01\) by %
\(f'(x)=\sum_{n=0}^\infty (n+1)t_{n+1}x^n\) and this derivative cannot
always be extended to \(\Intercc01\), think of \(f(x)=1-\sqrt{1-x}\)
whose derivative is \(f'(x)=\frac 1{2\sqrt{1-x}}\) for \(x<1\).
At this point it is useful to remember that this basic notion of
derivative that we have learned at school hides an ingredient that is
essential when one considers more general situations: a linear
argument.
Indeed, the derivative of a (sufficiently regular) function
\(h:E\to F\) where \(E\) and \(F\) are vector spaces is a function
\(h':E\to(\Limpl EF)\) where \(\Limpl EF\) is the vector space of
linear functions from \(E\) to \(F\).
In the case where \(E=\Real\) it is natural to identify \(\Limpl EF\)
with \(F\) and then \(h':\Real\to F\) as we are used to.
This identification amounts to saying that the derivative
\(h'(x)\in\Limpl\Real F\) is always given \(1\in\Real\) as linear
argument which is perfectly fine when one deals with vector spaces or
similar structures.
But the interval \(\Intercc 01\) is not a vector space, and in
particular when \(x,u\in\Intercc01\), it is not always true that
\(x+u\in\Intercc01\).
If we set \(S=\{(x,u)\in\Intercc01\St x+u\in\Intercc01\}\) and
consider the function \(\Diffsymb f:S\to\Realp\) given by %
\(\Diffsymb f(x,u)=f'(x)u\) then it is easy to check that %
\(\forall (x,u)\in S\ \Diffsymb f(x,u)\leq f(x+u)-f(x)\leq 1\) %
so that \(\Diffsymb f\) is actually an analytic morphism %
\(S\to\Intercc 01\).
We can even say that the map %
\(\mathsf T f:S\to S\) given by %
\(\mathsf T f(x,u)=(f(x),f'(x)u)\) is analytic.
This simple observation, together with the fact that the
correspondence %
\(f\mapsto\mathsf Tf\) is functorial is the starting point of a new
setting for differentiation in the \(\lambda\)-calculus and in \LL{}
called \emph{Coherent Differentiation} (CD).

\subsection*{Content}

The paper starts with a historical description of some ideas
leading to CD, that we divide in 3 phases:
\begin{itemize}
\item a first unpublished attempt by the author at defining the
  derivative of a stable function on coherence spaces in the 1980's;
\item the introduction of the differential \(\lambda\)-calculus and of
  \(\DILL\) in the 2000's
\item and last the discovery of CD in 2021, which gives a clear status
  to the first attempt and makes it work completely.
\end{itemize}

From \Cref{sec:summability}, we describe the categorical and syntactic
setting of CD.
In \Cref{sec:summability}, we introduce the basic structure of
summability, which allows to consider categories where hom-sets are
commutative partial monoids and more precisely axiomatizes
functorially an operation which maps an object \(X\) to the object
\(\Sfun X\) of pairs \((x_0,x_1)\) such that the sum \(x_0+x_1\)
exists.
This operation is presented as a functor \(\Sfun:\cL\to\cL\) equipped
with three natural transformations where, intuitively, \(\cL\) is a
``linear'' category.
Elaborating on this infrastructure, we introduce in
\Cref{sec:diff-structure} the basic idea of CD which is to represent
differentiation as a distributive law between the \(\oc\)-comonad of a
(weak) structure of model of \LL{} on \(\cL\) and the functor
\(\Sfun\) that we equip with a monad structure canonically induced by
the summability structure.
In \Cref{sec:closed-differential} we consider the case where the
category \(\cL\) is symmetric monoidal closed and hence induces a
model of the \(\lambda\)-calculus, explaining how the differential
structure interacts with the closed structure, and with fixpoint
operators when available (\Cref{sec:fixpoints}).

In these developments, we consider a particularly important situation,
called the \emph{elementary} situation where the functor \(\Sfun\) can
be described on objects by \(\Sfun X=\Limplp{\With\Sone\Sone}{X}\) and
similarly on morphisms.
In that case the differential structure boils down to a
\(\oc\)-coalgebra structure on \(\With\Sone\Sone\).
It turns out that all the concrete models of CD known so far are
elementary.

Last in \Cref{sec:syntax} we outline a syntax incorporating in a
functional language the categorical structures developed in the
previous section.
This functional language is an extension of Scott-Milner-Plotkin's
PCF~\cite{Plotkin77}.
Here are the main features of this extension.
\begin{itemize}
\item The only ground type (integers) is equipped with a \texttt{let}
  construct allowing to use call-by-value on integers, which is
  crucial when the language is extended with probabilistic choice
  (this refinement of PCF was introduced
  in~\cite{EhrhardPaganiTasson14});
\item there is a type constructor corresponding to the functor
  \(\Sfun\) alluded to above, and associated term constructors
  corresponding to the main categorical ingredients of the categorical
  axiomatization of CD;
\item the operational semantics is described by means of an abstract
  machine.
\end{itemize}
We give the main results about this operational semantics which have
been proven in~\cite{Ehrhard23b}: soundness and adequacy, and we
explain how the denotational semantics shows that this operational
semantics is deterministic.


\section{Notations and terminology}
A finite multiset of elements of a set \(A\) is a function %
\(m:A\to\Nat\) such that \(\Supp m=\{a\in A\St m(a)\not=0\}\) is
finite.
We use \(\Msetempty\) for the empty multiset such that
\(\Supp\Msetempty=\emptyset\) and standard algebraic notations
\(m_1+m_2\) and \(\sum_{i\in I}m_i\) (for \(I\) finite) for the
pointwise addition of multisets.

Given \(m=\Mset{\List a1k}\in\Mfin A\) and \(i\in I\) we define %
\(\Indactms im=\Mset{(i,a_1),\dots,(i,a_k)}\in\Mfin{I\times A}\).

If \(M=\Mset{\List m1k}\in\Mfin{\Mfin A}\), we set
\(\sum M=\sum_{i=1}^km_k\in\Mfin A\).

Let \(\cC\) be a category.
A family of morphisms \((h_i\in\cC(X,Y_i))_{i\in I}\) is %
\emph{jointly monic} if for any \(f,f'\in\cC(Z,X)\), if %
\((h_i\Compl f=h_i\Compl f')_{i\in I}\) then \(f=f'\).
And \((h_i\in\cC(X_i,Y))_{i\in I}\) is \emph{jointly epic} if for any
\(f,f'\in\cC(Y,Z)\), if %
\((f\Compl h_i=f'\Compl h_i)_{i\in I}\) then \(f=f'\).

\section{Differentiation in \LL} %
\label{sec:diff-in-LL}

A coherence space (CS) is a pair \(E=(\Web E,\mathord{\Coh E})\) where
\(\Web E\) is a set, the \emph{web}, and \(\Coh E\) is a binary,
reflexive and symmetric relation on \(\Web E\), the \emph{coherence
  relation}.
A clique of \(E\) is a subset \(x\) of \(\Web E\) such that
\(\forall a,a'\in x\ a\Coh Ea'\).
Given CS \(E\) and \(F\), one defines a CS \(\Limpl EF\) by
\(\Web{\Limpl EF}\) and \((a,b)\Coh{\Limpl EF}(a',b')\) if
\(a\Coh Ea'\Implies(b\Coh Fb'\text{ and }b=b'\Implies a=a')\), and the
category \(\COH\) has CS as objects, and \(\COH(E,F)=\Cl{\Limpl EF}\),
identity morphisms being the diagonal relations and composition being
the standard composition of relations.

The category \(\COH\) is a model of \LL{} and the exponential
considered first by Girard \(\Ocg E\) defined as follows:
\(\Web{\Ocg E}=\{\Set{\List a1n} \St n\in\Nat\text{ and }\{\List
a1n\}\in\Cl E\}\).
The main feature of this exponential is that its Kleisli category
\(\Klcat\COH{\Ocg}\) is isomorphic to the category of coherence spaces
and stable functions by the following correspondence:
with any \(t\in\Klcat\COH\Ocg(E,F)=\Cl{\Limpl{\Ocg E}F}\) one
associates the stable function %
\(\Fun t:\Cl E\to\Cl F\) defined by %
\(\Fun t(x)=\Set{b\in\Web F\St\exists x_0\subseteq x\ (x_0,b)\in
  t}\) %
and the mapping \(t\mapsto\Fun t\) is a bijection between %
\(\Klcat\COH\Ocg(E,F)\) and the set of stable functions
\(\Cl E\to\Cl F\).

Such a \(t\in\Klcat\COH\Ocg(E,F)\) is linear if all its elements
\((x_0,b)\) are such that \(x_0\) is a singleton, so that these linear
stable maps are just the same thing as elements of \(\COH(E,F)\).
We use the linear algebraic notation \(\Matappa tx\) to denote the
application of such a linear \(t\) to \(x\in\Cl E\), that is
\(\Matappa tx=\Set{b\in\Web F\St \exists a\in x\ (a,b)\in t}\).

\subsection{A first attempt: the derivative of a stable function} %
\label{sec:coh-space-stable-diff}

In front of these definitions, and in view of the role of linearity in
Analysis and Geometry, a natural question appeared to the author: %
is it possible to turn such a general stable morphism
\(t\in\Klcat\COH\Ocg(E,F)\) into a linear map by an operation similar
to differentiation?

This question should be made a bit more precise and is actually twofold: %
given \(x\in\Cl E\),
\begin{itemize}
\item can we define a coherence space \(E_x\) of all possible
  ``extensions'' of \(x\), that is, such that for all
  \(u\in\Cl{E_x} \), the set \(x\cup u\in\Cl E\)
\item and is there a linear \(t'(x)\in\Cl{\Limpl{E_x}{F}}\) such that,
  for all \(u\in\Cl{E_x}\), \(\Fun t(x)\cup\Matappa{t'(x)}u\) exists, is
  a subset of \(\Fun t(x\cup u)\) and is the ``best'' possible
  approximation of that set by means of a linear map?
\end{itemize}
A natural tentative answer to the first question is to take %
\(\Web{E_x}=\Set{a'\in\Web E\St \forall a\in x\ a\Coh Ea'}\) (with
\(a_1\Coh{E_x}a_2\) if \(a_1\Coh{E}a_2\)) and then
\(t'(x)=\Set{(a',b)\in\Web{E_x}\times\Web F \St \exists x_0\subseteq x\
  (x_0\cup\Set{a'},b)\in t}\).

This works quite well if we take \(x=\emptyset\), in that case
\(E_x=E\) and \(t'(\emptyset)\in\COH(E,F)\) and is characterized by
\begin{align*}
  \Matappa{t'(\emptyset)}u=\Union_{a\in u}\Fun t(\Set a)
\end{align*}
and notice also that we have a morphism %
\(\Coder_E=\Set{(a,\Set a)\St a\in\Web E}\in\COH(E,\Ocg E)\) %
such that \(t'(\emptyset)=t\Compl\Coder_E\).

Imagine now that \(x=\Set{a,a'}\) for some \(a\Scoh Ea'\), and that
moreover \(t=\Set{(x,b)}\).
Notice that \(a,a'\in\Web{E_x}\).
Then our definition of \(t'(x)\) yields \((a,b),(a',b)\in t'(x)\) and
hence \(t'(x)\notin\Cl{\Limpl{E_x}F}\).
We stopped our investigation of this idea at this point in 1986-1987,
but we could have tried to push this line of ideas a little bit
further as we explain now.

We can interpret this failure as meaning that our definition of
\(E_x\) is not satisfactory, we can try
\(\Web{E_x}=\Set{a'\in\Web E\St \forall a\in x\ a\Scoh Ea'}\) (with
\(a_1\Coh{E_x}a_2\) if \(a_1\Coh{E}a_2\)), that is, if
\(u\in\Cl{E_x}\), not only \(x\cup u\in\Cl E\) but also
\(x\cap u=\emptyset\).

With this definition of \(E_x\), let \(((a_i,b_i)\in t'(x))_{i=1,2}\)
and assume that \(a_1\Coh{E_x}a_2\), that is \(a_1\Coh{E}a_2\).
This means that there are \((x_i\subseteq x)_{i=1,2}\) with
\(((x_i\cup\Set{a_i},b_i)\in t)_{i=1,2}\), which implies
\(b_1\Coh Fb_2\) because \(x_1\cup x_2\cup\Set{a_1,a_2}\in\Cl E\) by
definition of \(E_x\).
If moreover \(b_1=b_2\), we know that
\(x_1\cup\Set{a_1}=x_2\cup\Set{a_2}\) which implies \(a_1=a_2\)
because \(a_1\notin x_2\) and \(a_2\notin x_1\) (remember that
\(x_1,x_2\subseteq x\)).
So we do have \(t'(x)\in\COH(E_x,F)\).
We can be even more precise: let \((a',b')\in t'(x)\) and let
\(b\in\Fun t(x)\), so that there is \(x_0\subseteq x\) such that
\((x_0,b)\in t\).
Let \(x'_0\subseteq x\) be such that \((x'_0\cup\Set{a'},b')\in t\),
then we have \(x_0\Scoh{\Ocg E}x'_0\cup\Set{a'}\) because we know that
\(a'\notin x_0\).
Therefore \(b\Scoh Fb'\) and we have shown that
\(b'\in\Web{F_{\Fun t (x)}}\) from which it follows that %
\(t'(x)\in\COH(E_x,F_{\Fun t(x)})\).

Let us adopt the following convention introduced by Girard:
given \((x_i\in\Cl E)_{i\in I}\), we use the notation
\(\sum_{i\in I}x_i\) to denote \(\Union_{i\in I}x_i\) and to express
at the same time that the \(x_i\)'s are pairwise disjoint.
Then our definition of \(t'(x)\) satisfies
\begin{align*}
  \Matappa{t'(x)}u=\sum_{a\in u}(\Fun t(x+\Set a)\setminus\Fun t(x))\,.
\end{align*}
Indeed, let first \(b\in\Matappa{t'(x)}u\), so let \(a\in u\) be such
that %
\((a,b)\in t'(x)\).
There is \(x_0\subseteq x\) such that \((x_0+\Set a,b)\in t\) from
which it follows that \(b\in\Fun t(x+\Set a)\).
Next since \(x_0+\Set a\) is minimal such that
\(b\in\Fun t(x_0+\Set a)\) and \(a\notin x\), we cannot have
\(b\in\Fun t(x)\) from which the \(\subseteq\) inclusion follows.
Let now \(b\in\sum_{a\in u}(\Fun t(x+\Set a)\setminus\Fun t(x))\), so
let \(a\in u\) be such that
\(b\in\Fun t(x+\Set a)\setminus\Fun t(x)\).
Let \(x'_0\subseteq x+\Set a\) be such that \((x'_0,b)\in t\).
Since \(b\notin\Fun t(x)\), we cannot have \(x'_0\subseteq x\) and
hence we must have \(x'_0=x_0+\Set a\) for some \(x_0\subseteq x\).
Since \(a\in u\in\Cl{E_x}\), we have \((a,b)\in t'(x)\) from which the
\(\supseteq\) inclusion follows.

This shows in particular that
\(\Fun t(x)+\Matappa{t'(x)}{u}\subseteq\Fun t(x+u)\).
Let now \(h\in\COH(E_x,F_{\Fun t(x)})\) be such that %
\(\forall u\in\Cl{E_x}\ \Fun t(x)+\Matappa hu\subseteq\Fun t(x+u)\),
that is \(\Matappa hu\subseteq\Fun t(x+u)\setminus\Fun t(x)\).
In particular, for \(u=\Set a\) with \(a\in\Web{E_x}\) we get
\(\Matappa h{\Set a}\subseteq\Matappa{t'(x)}{\Set a}\) which means that
\(h\subseteq t'(x)\).
In that precise sense \(t'(x)\) is the best linear under-approximation
of the map \(u\mapsto \Fun t(x+u)\setminus\Fun t(x)\) so can be
reasonably be called the derivative of \(t\) at \(x\).

%
Let \(s\in\Klcat\COH\Ocg(E,F)\) and \(t\in\Klcat\COH\Ocg(F,G)\) so
that %
\(t\Comp s\in\Klcat\COH\Ocg(E,G)\) is the composition in the Kleisli
category and can be described as follows:
\begin{multline*}
  t\Comp s
  =\{(x_1\cup\cdots\cup x_n,c)\in\Web{\Ocg E}\times\Web G
  \St \exists \List b1n\in\Web F\\
  ((x_i,b_i)\in s)_{i=1}^n \text{ and }
  (\Set{\List b1n},c)\in t\}
\end{multline*}
and is fully characterized by %
\(\Fun{t\Comp s}(x)=\Fun t(\Fun s(x))\).

Remember that the coherence space \(\With\Sone\Sone\) has %
\(\{(1,\Sonelem),(2,\Sonelem)\}\) as web, with %
\({(1,\Sonelem)}\Scoh{\With\Sone\Sone}{(2,\Sonelem)}\).
Let \(s\in\Klcat\COH\Ocg(\Sone,\With\Sone\Sone)\) and %
\(t\in\Klcat\COH\Ocg(\With\Sone\Sone,\Sone)\) be given by
\begin{align*}
  s&=\Set{(\Sonelem,(1,\Sonelem)),(\Sonelem,(2,\Sonelem))}\\
  t&=\Set{(\Set{(1,\Sonelem),(2,\Sonelem)},\Sonelem)}\,.
\end{align*}
Then we have
\(t\Comp s=\Set{(\Set\Sonelem,\Sonelem)}\).
Therefore
\(\Matappa{(t\Comp s)'(\emptyset)}{\Set\Sonelem}=\Set{\Sonelem}\)
and on the other hand we have %
\begin{align*}
  \Matappa{t'(\Fun s(\emptyset))}{(\Matappa{s'(\emptyset)}{\Set\Sonelem})}
  &=\Matappa{t'(\emptyset)}{(\Matappa{s'(\emptyset)}{\Set\Sonelem})}\\
  &=\Matappa{t'(\emptyset)}{\Set{(1,\Sonelem),(2,\Sonelem)}}\\
  &=\Fun t\Set{(1,\Sonelem)})
    \cup\Fun t\Set{(2,\Sonelem)})\\
  &=\emptyset\,.
\end{align*}

This means that the chain rule
\begin{align*}
  (t\Comp s)'(x)=t'(s(x))\Compl s'(x)
\end{align*}
does not hold for this differentiation of stable function, we only
have a weak version thereof
\((t\Comp s)'(x)\supseteq t'(s(x))\Compl s'(x)\).

The reason for the failure of the chain rule is clear: the morphism
\(t\) is nonlinear (it needs to be fed with
\(\Set{(1,\Sonelem),(2,\Sonelem)}\) to output the atomic result
\(\Sonelem\), that is, it uses its parameter at least twice) so in the
computation of \(\Fun{t\Comp s}(\Set\Sonelem)\), the atomic data
\(\Sonelem\) is actually used at least twice, but this nonlinearity
doesn't appear in \(t\Comp s=\Set{(\Set\Sonelem,\Sonelem)}\) which
turns out to be a linear morphism in \(\COH\): the two used copies of
\(\Sonelem\) have been ``merged''.


Another difficulty in this concrete approach to differentiation of
stable maps is that it is not clear how to express the regularity (and
hopefully, the stability) of the derivative \(t'(x)\) with respect to
\(x\) (the author was not aware of the ---~at that time recently
introduced~--- tangent categories of~\cite{Rosicky84}), and therefore,
it was unclear how to define higher derivatives in this kind of
setting.

\subsection{Coming back to differentiation in \LL}
\label{sec:difflam}
In the early 2000, motivated by the phase space parameterized \LL{}
models of~\cite{BucciarelliEhrhard99}, the author explored categorical
models of \LL{} where formulas are interpreted as vector spaces and
morphisms as linear maps (in the usual sense of Linear Algebra).
In such categories, the only known resource modalities yield infinite
dimensional vector spaces as soon as they are applied to a non \(0\)
space, and so the vector spaces under consideration must be equipped
with a topology compatible with their algebraic structure.
Two such models were developed by the author: Köthe sequence
spaces~\cite{Ehrhard00c} and finiteness spaces~\cite{Ehrhard00b}.

The objects of these models are topological vector spaces (tvs) which
admit a simple description based on the existence of a ``web'' (in the
sense of coherence spaces) which, in this algebraic context, can be
understood as a Schauder basis, that is, not exactly a basis in the
usual algebraic sense (Hamel basis), but a natural topological
adaptation thereof, in which all the elements of the vector space can
be written uniquely as infinite linear combinations of base vectors
(these infinite sums being defined as limits the sense of the topology
the vector space is endowed with).
One important feature of these models is that these webs are not used
in the definition of morphisms, which are just linear and continuous
maps.

The obtained categories \(\cL\) are models of \(\LL\) such that the
Kleisli category \(\Klcat\cL{\oc}\) can be described as a category of
analytic (or more precisely entire) functions.
Moreover, the homsets of these categories have a natural tvs structure
(because \(\cL\) is an SMCC) and in particular have an operation of
addition: they are additive categories.

The fact that the morphisms of such CCC \(\Klcat\cL\oc\) are analytic
and therefore infinitely differentiable led to the idea of extending
the typed%
\footnote{%
  Because the fixpoint operators of the untyped \(\lambda\)-calculus
  are hardly compatible with differentiation, and, in the concrete
  models at hand, the morphisms of \(\Klcat\cL\oc(X,X)\) do not have
  fixpoints in general.} %
\(\lambda\)-calculus with differential operations.
The basic differential typing rule of this calculus is
\begin{center}
  \begin{prooftree}
    \hypo{\Tseq\Gamma M{\Timpl AB}}
    \hypo{\Tseq\Gamma NA}
    \infer2{\Tseq\Gamma{\Diffapp MN}{\Timpl AB}}
  \end{prooftree}
\end{center}
the intuition being that \(A\) and \(B\) denote some kind of tvs \(E\)
(of the SMCC \(\cL\)) and \(F\), and \(M\) an entire function
\(f:E\to F\), more precisely \(f\in\Klcat\cL\oc(E,F)\) %
(possibly depending on additional parameters listed in \(\Gamma\)).
Such a function can be differentiated into another entire function %
\(f'\in\Klcat\cL\oc(E,\Limpl EF)\) such that, for any \(x\in E\), the
map \(u\mapsto f(x)+\Matappa{f'(x)}u\) is the best affine
approximation of the map \(u\mapsto f(x+u)\) (in the sense of the
topology our tvs are endowed with).
Then, if \(N\) denotes \(u\in E\), the term %
\(\Diffapp MN\) denotes the analytic function \(E\to F\) which maps
\(x\) to \(f'(x)\cdot u\).
This slightly unusual writing of differentials makes it easy to
iterate derivatives: given \((\Tseq\Gamma{N_i}A)_{i=1,2}\), we
have %
\(\Tseq\Gamma{\Diffapp M{(\Diffapp M{N_1})}{N_2}}{\Timpl AB}\), the
second derivative of \(M\), a bilinear morphism applied to its two
linear arguments.

This differential application induces a new redex, in the case where %
\(M=\Abst xAP\) with \(\Tseq{\Gamma,x:A}PB\), similar to a
\(\beta\)-redex.
The corresponding reduction is
\begin{align*}
  \Diffapp{(\Abst xAP)}{N}\Rel\Red\Diffp MxN
\end{align*}
where the term \(\Diffp MxN\) is the \emph{differential substitution}
of \(N\) for \(x\) in \(M\), defined by induction on \(M\), which is
typed as follows
\begin{align*}
  \Tseq{\Gamma,x:A}{\Diffp MxN}{B}
\end{align*}
which shows that, in the term \(\Diffp MxN\), the variable \(x\) can
still be free.
This is due to the fact that in that term only one linear copy is
substituted with \(N\).
This linear substitution operation performs a non-trivial operation on
\(M\), creating linear occurrences of \(x\) to be substituted by \(N\)
on request.
The most important case in the definition of \(\Diffp MxN\) is when
\(M\) is an (ordinary) application \(M=\App PQ\) with %
\(\Tseq{\Gamma,x:A}P{\Timpl CB}\) and \(\Tseq{\Gamma,x:A}Q{C}\).
We inductive definition of this linear substitution stipulates that
\begin{align*}
  \Diffp{\App PQ}xN
  =\App{\Diffp PxN}Q+\App{\Diffapp P{(\Diffp QxN)}}N
\end{align*}
and this definition involves the \(+\) operation on terms, subject to
the following typing rule
\begin{equation}
  \label{eq:sum-terms-lambdadiff}
  \begin{prooftree}
    \hypo{\Tseq\Gamma{M_1}A}
    \hypo{\Tseq\Gamma{M_2}A}
    \infer2{\Tseq\Gamma{M_1+M_2}A}
  \end{prooftree}
\end{equation}
which has an obvious denotational interpretation in our tvs models and
should, operationally, be understood as a nondeterministic
superposition.
The case of a variable is also interesting, we set
\begin{align*}
  \Diffp yxN=
  \begin{cases}
    N & \text{if }y=x\\
    0 & \text{otherwise}
  \end{cases}
\end{align*}
where we see a \(0\) which is the neutral element of the \(+\) above.
The meaning of this \(0\) is that if \(y\not=x\) then ``\(y\) does
not depend on \(x\)'' and so we are taking the derivative of a
constant function; a more operational understanding is that \(y\) has
no linear occurrence of \(x\) and hence the linear substitution fails.
Addition is allowed by \Cref{eq:sum-terms-lambdadiff} because in
term \(M\) the variable \(x\) may have several potential linear
occurrences.
Consider for instance \(M=\App x{\App xy}\) typed as follows:
\begin{align*}
  \Tseq{\Gamma,y:A,x:\Timpl AA}{\App x{\App xy}}{A}
\end{align*}
and let \(N\) be a term such that \(\Tseq\Gamma N{\Timpl AA}\), we
have
\begin{align*}
  \Diffp{\App x{\App xy}}xN
  =\App N{\App xy}+\App{\Diffapp x{\App Ny}}{\App xy}
\end{align*}
where some intuitively clear equations on terms (such as
\(\Diffapp P0=0\)) have been used implicitly.
Let us write the term \(M=\App{x_1}{\App{x_2}y}\), using \(x_1,x_2\)
to distinguish the two occurrences of \(x\) in \(M\).
Only the occurrence \(x_1\) is linear%
\footnote{In our definition of
  \(\frac{\partial M}{\partial x}\cdot N\), we are using implicitly
  the fact that our \(\lambda\)-calculus is equipped with a CBN
  operational semantics, translated in \LL{} by the standard Girard
  translation.
  If our calculus were CBV, translated in \LL{} through the ``boring''
  Girard translation, the situation would be different.}, %
the occurrence \(x_2\) is not because the occurrence \(x_1\) might
take as value a nonlinear function, using \(x_2\) in a nonlinear way.
This is why for substituting \(N\) linearly for \(x_2\) we need first
to make the function \(x_1\) use its argument linearly (or more
precisely extract a linear copy of its argument); this is exactly the
purpose of the \(\Diffapp x\_\) in the second term of the sum.

It is worth observing that these differential reduction rules produce
non trivial sums even if the term we start from does not contain such
sums.
Consider for instance, in a typed calculus with a base type of
booleans, the following term
\begin{align*}
  \Tseq{x:\Bool}{M=\If{x}{\If{x}{\False}{\True}}{\If{x}{\False}{\True}}}{\Bool}
\end{align*}
then the definition of differential substitution leads to
\begin{align*}
  \Diffapp{(\Diffapp{(\Abst x\Bool M)}\True)}\False
  \Red\True+\False\,.
\end{align*}
This results from the fact that, using iterated differential
application, we manage to give the variable \(x\) two incompatible
values \(\True\) and \(\False\) and the term \(M\) is written in such
a way that, when \(x\) changes value during the computation (an
impossible scenario in a deterministic setting), \(M\) issues
\(\True\) or \(\False\), depending on the scheduling of this change of
value.
The order of the differential substitution being essentially
irrelevant (this corresponds to the Schwarz rule of Calculus: the
second derivative is a \emph{symmetric} bilinear function), the
computation is necessarily nondeterministic and leads to this
nontrivial sum \(\True+\False\).

This kind of example strongly suggested that extending the
\(\lambda\)-calculus with differential constructs necessarily leads to
essentially non-deterministic systems.

Later on, the author developed a differential extension of \LL{},
fully compatible with theses semantic and syntactic ideas.
The beauty of this differential \LL{} is that the new differential
logical structure does not require new connectives, but introduces new
deduction rules relative to the resource modality \(\oc\) of \LL, dual
to the standard rules of dereliction, weakening and contraction.
We already mentioned codereliction \(\Coder_E\) in the setting of
coherence spaces in \Cref{sec:coh-space-stable-diff}; coweakening and
cocontraction are similar (the first also exists in coherence spaces,
the second not).
The associated new cut elimination rules preserve this new symmetry.
The recent~\cite{KerjeanLemay23} even extends this symmetry
to promotion, the fundamentally infinitary rule of \LL{}, which
becomes then a bimonad (as explained in that paper there is a further
price to pay for this extension).
Similar ideas were already considered in~\cite{Gimenez09}.

\section{Coherence and determinism}
At the most fundamental level, so was the situation concerning the
differential \(\lambda\)-calculus and \LL{} in May~2021.
Of course many results have been obtained and many notions have been
introduced concerning these systems, their applications and their
semantics since they have been introduced in the early 2000's, it
would not be possible to mention all of them here.
We can stress in particular many important advances on the categorical
semantics of differentiation ---~and notably the use of
2-categories~--- and applications of the syntactic Taylor expansion
associated with Differential \LL{} and the differential
\(\lambda\)-calculus.
Nevertheless, as far as we know, none of these developments questions the
assumption that it is always possible to add terms, proofs or morphisms
of the same type.
And as explained above there are very good reasons for such an
assumption.

However, the first observations summarized in
\Cref{sec:coh-space-stable-diff}, suggesting the possibility of giving
a meaning to derivatives in categorical models of \LL{} where addition
is a partial operation on morphisms, were strongly reinforced (more
than 30 years later!) by the study of Probabilistic Coherence Spaces
(PCS) developed in~\cite{Ehrhard22}, based on the fact that the
Kleisli morphisms can clearly be understood as analytic functions in a
very standard sense.
In that paper it is shown that the endeavor of
\Cref{sec:coh-space-stable-diff} can be carried out successfully in
PCS, in the sense that the chain rule, which failed in \(\COH\) as we
showed, perfectly holds in this setting.

As already observed, beyond the failure of the chain rule, one major
puzzling question in \Cref{sec:coh-space-stable-diff} was: how can we
express that \(s'(x)\in\COH(E_x,F_{\Fun s(x)})\) depends stably on
\(x\)?
This kind of question already arises in Differential Geometry where
the derivative of a map from a real manifold \(X\) to
another one \(Y\) at a point \(x\in X\) is a linear map
\(f'(x):\Tanspace xX\to\Tanspace{f(x)}Y\) where \(\Tanspace xX\) is
the tangent space to \(X\) at \(x\), which is a vector space.
To express that this derivative has, for instance, a derivative at
each point, one introduces a new manifold \(\Tanbundle X\) (the
tangent bundle of \(X\)) whose elements are the pairs \((x,u)\) such
that \(u\in\Tanspace xX\) and one turns this operation into a functor,
mapping \(f:X\to Y\) to the function \(\Tanbundle X\to\Tanbundle Y\)
defined by \(\Tanbundle f(x,u)=(f(x),f'(x)\cdot u)\) and then one can
speak of the regularity (for instance, the differentiability) of this
compound map \(\Tanbundle f\).
This standard construction has been categorically axiomatized
in~\cite{Rosicky84}, leading to a notion of \emph{tangent category}.
The functoriality of the operation \(\Tanbundle{}\) expresses exactly
the chain rule.

The main idea of Coherent Differentiation is very close to that of
tangent categories, with one major difference which required to the
author some time to be fully understood.
In the tangent bundle construct, the manifold \(X\) and the tangent
space at \(x\in X\) are typically of very different natures: the
tangent spaces are usually all isomorphic to \(\Realto d\) where \(d\)
is the dimension of the manifold ---~and so they are trivial geometric
objects~---, whereas the manifold itself is a complicated geometrical
object (defined typically by systems of equations, gluing, quotient
etc).
In our setting which arises from the \LL{} analysis of denotational
semantics, the manifold is replaced by a ``domain'' \(E\), a coherence
space for instance, and, given \(x\in\Cl E\), an element of
\(\Tanspace xE=E_x\) should be an \(y\in\Cl E\) such that \(x+y\) makes
sense, that is \(x\cap y=\emptyset\) and \(x\cup y\in\Cl E\), for the
quasi-example developed in \Cref{sec:coh-space-stable-diff}.
Since the \LL{} analysis of denotational semantics is based on a
fundamental analogy between domains and vector spaces, this means that
here the ``tangent bundle'' functor already applies non-trivially to
objects of the linear category and implements ---~functorially as we
shall see~--- a notion of partial summability.
In sharp contrast, in the tangent bundle case, when \(X\) is a vector
space, the associated tangent bundle is trivial:
\(\Tanbundle X=X\times X\) equipped with the first projection.

\begin{remark}
  This also means that some room is left for developing a notion of
  ``manifold'' for coherent differentiation, or of a notion of
  coherent tangent categories where the ``tangent spaces'' would only
  be partially additive.
  Such a generalization requires motivations coming from concrete
  computational situations or from coherent differential situations
  arising in geometry; as far as we know such situations are still to
  be discovered.
\end{remark}

The first ingredient of coherent differentiation is therefore an
axiomatization of categories where morphisms are only partially
summable.
It is perfectly meaningful, although not really necessary%
\footnote{See~\cite{EhrhardWalch23} where the theory is developed
  without this assumption.} %
to assume that such a category \(\cL\) is a ``linear category'', that
is an SMC category with possibly additional properties and structures
(cartesian products, resource modality etc).
The present paper makes this kind of assumption about \(\cL\).

We could assume that \(\cL\) is enriched in some kind of ``partial
commutative monoids'', but this would not be really sufficient,
because we also need to associate with any object \(X\) of \(\cL\) an
object \(\Sfun X\) whose elements are, intuitively, the pairs
\((x_0,x_1)\in X^2\) such that \(x_0+x_1\) is well defined.
Therefore our partial summability structure is axiomatized as a
functor \(\Sfun:\cL\to\cL\) equipped with three natural
transformations \(\Sproj0,\Sproj1,\Ssum:\Sfun X\to X\) which
intuitively map such a summable pair \((x_0,x_1)\) to \(x_0\), \(x_1\)
and \(x_0+x_1\) respectively.
Then saying that two morphisms \(f_0,f_1\in\cL(X,Y)\) are summable
simply means that there is a morphism \(h\in\cL(X,\Sfun Y)\) such that
\(\Sproj i\Compl h=f_i\) for \(i=0,1\) and since it is important for
us that \(h\), the ``witness of summability'' of \(f_0\) and \(f_1\),
be unique, we assume \(\Sproj0\) and \(\Sproj1\) to be jointly monic.
Thanks to this uniqueness, we can set \(f_0+f_1=\Ssum\Compl h\).
Suitable axioms on this structure allow to show that \(\cL\) is
enriched in partial commutative monoids%
\footnote{For a rather restrictive class of partially commutative
  monoids, some variations are probably possible on this aspect of
  the theory.
  The crucial point here is the way associativity is axiomatized in a
  partial setting: several options are available.}.

\section{Summability structures in a linear setting}
\label{sec:summability}

\subsection{Partial monoids}
We first describe the kind of partial commutative monoids that our
axiomatization of summability induces.

\begin{definition} %
  \label{def:part-comm-monoids}
  A \emph{partial commutative monoid} is a triple %
  \((M,0,\mathord +)\) where \(M\) is a set, \(0\in M\) and %
  \(\mathord+:M^2\to M\) is a partial function such that
  \begin{itemize}
  \item \(0+a\) is defined for all \(a\in M\) and \(0+a=a\);
  \item if \(a+b\) is defined then \(b+a\) is defined and \(a+b=b+a\);
  \item if \(a+b\) and \((a+b)+c\) are defined then \(b+c\) and
    \(a+(b+c)\) are defined, and \((a+b)+c=a+(b+c)\).
  \end{itemize}
\end{definition}

\begin{remark}
  This notion of partial commutative monoid is stronger than it might
  seem at first sight and involves some kind of ``positivity''.
  For instance the set \(M=\{0,1\}\subseteq\Relint\) with addition
  \(a+b\) defined as in \(\Relint\) if \(a+b\in M\) and
  undefined otherwise, is a partial commutative monoid.
  But if we apply a similar definition to \(M=\{-1,0,1\}\), the
  obtained structure does not satisfy the associativity condition of
  partial commutative monoids (take \(a=-1\) and \(b=c=1\)).

  This notion of partial monoid is perfectly adapted to the kind of
  denotational situations we are abstracting on ---~which are
  essentially positive~---, but other notions of partial monoid have
  been introduced and might lead to interesting notions of summability
  structure adapted to more algebraic or geometric situations; such an
  approach will be presented by Aymeric~Walch in a forthcoming paper.
\end{remark}

\begin{definition} %
  \label{def:part-com-mon-general-sums}
  Let \((M,0,\mathord +)\) be a partial commutative monoid.
  Then we define by induction on \(n\) what it means for a sequence
  \(\Vect a\in M^n\) to be summable, and the value \(\sum_{i=1}^na_i\)
  of its sum:
  \begin{itemize}
  \item if \(n=0\), the empty sequence is summable and has \(0\) as
    sum;
  \item if \(n>0\), a sequence \(\List a1n\) is summable if %
    \(\List a1{n-1}\) is summable and \(\sum_{i=1}^{n-1}a_i\) and
    \(a_n\) are summable, and then
    \(\sum_{i=1}^na_i=(\sum_{i=1}^{n-1}a_i)+a_n\).
  \end{itemize}
\end{definition}

\begin{lemma}
  \label{lemma:part-comm-mon-sum-permut}
  Let \((M,0,\mathord +)\) be a partial commutative monoid.
  Let \(\Vect a\in M^n\) and let \(f:\{1,\dots,n\}\to\{1,\dots,n\}\)
  be a bijection.
  The sequence \(\Vect a\) is summable iff
  \((a_{f(1)},\dots,a_{f(n)})\) is summable, and then %
  \(\sum_{i=1}^na_i=\sum_{i=1}^na_{f(i)}\).
\end{lemma}

Thanks to that lemma, the following definition makes sense.
\begin{definition}
  Let \((M,0,\mathord +)\) be a partial commutative monoid and \(I\)
  be a finite set.
  One says that \(\Vect a\in M^I\) is summable if there is an
  enumeration without repetitions \(\Vect i=(\List i1n)\) of the
  elements of \(I\) (so that \(n=\Card I\)) such that
  \((a_{i_1},\dots,a_{i_n})\) is summable, and if this is the case the
  sum \(\sum\Vect a=\sum_{i\in I}a_i\) of \(\Vect a\) is defined as
  \(\sum_{j=1}^na_{i_j}\).
  Indeed, by \Cref{lemma:part-comm-mon-sum-permut}, the summability
  and sum of \((a_{i_1},\dots,a_{i_n})\) does not depend on the
  enumeration \(\Vect i\) of \(I\).
\end{definition}

\begin{theorem} %
  \label{th:part-monoid-summable-fam}
  Let \((M,0,\mathord +)\) is a partial commutative monoid.
  Let \(I\) be a finite set and \((I_j)_{j\in J}\) be a finite family
  of pairwise disjoint sets such that \(\Union_{j\in J}I_j=I\).
  Let \(\Vect a\in M^I\).
  The following statements are equivalent
  \begin{itemize}
  \item \(\Vect a\) is summable;
  \item for all \(j\in J\) the family \((a_i)_{i\in I_j}\) is summable
    and the family \((\sum_{i\in I_j}a_i)_{j\in J}\) is summable.
  \end{itemize}
  When these two equivalent conditions hold, one has %
  \(\sum_{i\in I}a_i=\sum_{j\in J}\sum_{i\in I_j}a_i\).
\end{theorem}

The proofs of these facts are standard and can also be found
in~\cite{Ehrhard23a}.

\begin{lemma} %
  \label{th:part-monoid-summable-subfam}
  If \((a_i)_{i\in I}\) is a finite summable family in a partial
  commutative monoid \((M,0,\mathord +)\) and \(I'\subseteq I\), then
  \((a_i)_{i\in I'}\) is summable.
\end{lemma}
\begin{proof}
  Immediate consequence of \Cref{th:part-monoid-summable-fam} (take
  \(J=\Eset{1,2}\), \(I_1=I'\) and \(I_2=I\setminus I'\)).
\end{proof}

\subsection{Summability structures} %
\label{sec:partial-monoids}

Let \(\cL\) be a category with zero-morphisms, that is, \(\cL\) is
enriched over the category of pointed sets.
We use \(0_{X,Y}\) or simply \(0\) for the distinguished zero-element
of \(\cL(X,Y)\), so that \(f\Compl 0=0\Compl f=0\).
If \(\cL\) is an SMC, we also assume that \(\Tens 0f=0\) and
\(\Tens f0=0\).
If \(\cL\) has a terminal object \(\Top\), notice that
\(\cL(X,\Top)=\Eset0\) for any object \(X\).

The first structure we assume \(\cL\) to be equipped with is a
functor %
\(\Sfun:\cL\to\cL\) whose intuitive meaning is to map any object \(X\)
to the object \(\Sfun X\) of all pairs \((x_0,x_1)\) of elements of
\(X\) for which the sum \(x_0+x_1\) exists.
In accordance with this intuition, this functor is equipped with three
natural transformations %
\(\Sproj0,\Sproj1,\Ssum\in\cL(\Sfun X,X)\) intuitively mapping such a
pair to \(x_0\), \(x_1\) and \(x_0+x_1\) respectively.

\begin{definition}
  A \emph{pre-summability structure} on \(\cL\) is a triple %
  \((\Sfun,\Sproj0,\Sproj1,\Ssum)\) where \(\Sfun:\cL\to\cL\) is a
  functor and \(\Sproj0,\Sproj1,\Ssum\in\cL(\Sfun X,X)\) are natural
  transformations such that \(\Sproj0\) and \(\Sproj1\) are jointly
  monic.
\end{definition}
Remember that the latter condition means that if
\(f,g\in\cL(Y,\Sfun X)\) satisfy %
\((\Sproj i\Compl f=\Sproj i\Compl g)_{i=0,1}\) then \(f=g\); this is
the categorical way to say that \(\Sfun X\) is an object of pairs.
From now on we assume to be given a pre-summability structure on
\(\cL\).

\begin{definition}
  We say that two morphisms \((f_i\in\cL(Y,X))_{i=0,1}\) are
  \emph{summable} if there is \(h\in\cL(Y,\Sfun X)\) such that
  \((\Sproj i\Compl h=f_i)_{i=0,1}\).
  If such an \(h\) exists, it is unique by joint monicity of the
  \(\Sproj i\)'s and we set \(h=\Stuple{f_0,f_1}\) and %
  \(f_0+f_1=\Ssum\Compl\Stuple{f_0,f_1}\in\cL(Y,X)\).
  The morphism \(\Stuple{f_0,f_1}\) is the \emph{witness} of the
  summability of \(f_0\) and \(f_1\), and \(f_0+f_1\) is their
  \emph{sum}.
\end{definition}

\begin{lemma}
  The morphisms \(\Sproj0,\Sproj1\in\cL(\Sfun X,Y)\) are summable,
  with \(\Stuple{\Sproj0,\Sproj1}=\Id_{\Sfun X}\) and
  \(\Sproj0+\Sproj1=\Ssum\).
\end{lemma}
This is tautological.

\begin{definition}
  If \(\Vect i\in\Eset{0,1}^n\), we set %
  \(\Sproj{\Vect i}=\Sproj{i_1}\Compl\cdots\Sproj{i_n}
  \in\cL(\Sfun^n X,X)\).
\end{definition}

\begin{lemma} %
  \label{lemma:sproj-it-jointly-monic}
  The morphisms \(\Sproj{\Vect i}\in\cL(\Sfun^n X,X)\), for
  \(\Vect i\in\Eset{0,1}^n\), are jointly monic.
\end{lemma}
\begin{proof}
  Simple induction on \(n\).
\end{proof}

\begin{lemma} %
  \label{lemma:summability-compl}
  If \(f_0,f_1\in\cL(X,Y)\) are summable and %
  \(g\in\cL(U,X)\) and \(h\in\cL(Y,V)\) then %
  \(h\Compl f_0\Compl g\) and \(h\Compl f_1\Compl g\) are summable
  with %
  \(\Stuple{h\Compl f_0\Compl g,h\Compl f_1\Compl g}
  =\Sfun h\Compl\Stuple{f_0,f_1}\Compl g\) and %
  \(h\Compl f_0\Compl g+h\Compl f_1\Compl g
      =h\Compl (f_0+f_1)\Compl g\).
\end{lemma}
By naturality.

The additional assumptions on pre-summability structures that we will
introduce now for defining summability structures will turn each
hom-set of \(\cL\) into a partial commutative monoid in the sense of
\Cref{sec:partial-monoids}.

\begin{definition}\labeltext{(\(\Sfun\)-com)}{ax:scom}\ref{ax:scom}
  The pre-summability structure \((\Sfun,\Sproj0,\Sproj1,\Ssum)\) is
  \emph{commutative} if \(\Sproj1\) and \(\Sproj0\) are summable and
  \(\Sproj1+\Sproj0=\Ssum\).
\end{definition}

\begin{lemma}
  If the pre-summability structure \((\Sfun,\Sproj0,\Sproj1,\Ssum)\)
  is commutative and \(f_0,f_1\in\cL(Y,X)\) are summable, then
  \(f_1,f_0\) are summable and \(f_1+f_0=f_0+f_1\).
\end{lemma}

\begin{remark}
  The~\ref{ax:scom} axiom corresponds to the commutativity condition
  in \Cref{def:part-comm-monoids}.
  It should be noticed that even if \(f_0+f_1=f_1+f_0\), it is of
  course not the case that \(\Stuple{f_0,f_1}=\Stuple{f_1,f_0}\)
  (unless \(f_0=f_1\)) and the fact that these witnesses are distinct
  is an essential aspect of the theory.
\end{remark}

\begin{definition}
  \labeltext{(\(\Sfun\)-zero)}{ax:szero}\ref{ax:szero} %
  The pre-summability structure \((\Sfun,\Sproj0,\Sproj1,\Ssum)\)
  \emph{has zero} if \(0_{X,X}\) and \(\Id_X\) are summable and
  \(0+\Id_X=\Id_X\).
\end{definition}

\begin{lemma}
  If the pre-summability structure \((\Sfun,\Sproj0,\Sproj1,\Ssum)\)
  has zero then, for any \(f\in\cL(X,Y)\), \(0\) and \(f\) are summable and
  \(0+f=f\).
\end{lemma}
\begin{proof}
  Easy consequence of \Cref{lemma:summability-compl}.
\end{proof}

\begin{definition}
  \labeltext{(\(\Sfun\)-wit)}{ax:swit}\ref{ax:swit} %
  The pre-summability structure \((\Sfun,\Sproj0,\Sproj1,\Ssum)\)
  \emph{has witnesses} if, for any \(f_0,f_1\in\cL(X,\Sfun Y)\), if %
  \(\Ssum\Compl f_0\) and \(\Ssum\Compl f_1\) are summable, then %
  \(f_0\) and \(f_1\) are summable.
\end{definition}

\begin{lemma} %
  If the pre-summability structure \((\Sfun,\Sproj0,\Sproj1,\Ssum)\)
  satisfies~\ref{ax:swit}, then there is a unique natural
  \(\Sflip_X\in\cL(\Sfun^2X,\Sfun^2X)\) such that %
  \(\Sproj i\Compl\Sproj j\Compl\Sflip_X=\Sproj j\Compl\Sproj i\) %
  for all \(i,j\in\Eset{0,1}\).
  Moreover \(\Sflip_X^2=\Id_{\Sfun^2X}\) and %
  \(\Sflip_X=\Stuple{\Sfun\Sproj0,\Sfun\Sproj1}\).
\end{lemma}
\begin{proof}
  Since \(\Sproj0,\Sproj1\in\cL(\Sfun^2 X,\Sfun X)\) are summable, we
  know by \Cref{lemma:summability-compl} that
  \(\Sproj i\Compl\Sproj0,\Sproj i\Compl\Sproj1\in\cL(\Sfun^2X,X)\) %
  are summable for \(i=0,1\) which gives us witnesses %
  \[
    (f_i =\Stuple{\Sproj i\Compl\Sproj0,\Sproj i\Compl\Sproj1}
    =(\Sfun\Sproj i)\Stuple{\Sproj0,\Sproj1} =\Sfun\Sproj i\in
    \cL(\Sfun^2X,\Sfun X))_{i=0,1}\,.
  \]
  We have \((\Ssum\Compl f_i=\Sproj i\Compl\Ssum_{\Sfun X})_{i=0,1}\)
  by naturality of \(\Ssum\) so that \(\Ssum f_0\) and \(\Ssum f_1\)
  are summable by \Cref{lemma:summability-compl} again.
  So by~\ref{ax:swit} the morphisms \(f_0\) and \(f_1\) are summable.
  We set
  \begin{align*}
    \Sflip_X
    =\Stuple{f_0,f_1}
    =\Stuple{\Stuple{\Sproj0\Compl\Sproj0,\Sproj0\Compl\Sproj1},
    \Stuple{\Sproj1\Compl\Sproj0,\Sproj1\Compl\Sproj1}}
    =\Stuple{\Sfun\Sproj0,\Sfun\Sproj1}\,.
  \end{align*}
  This definition implies immediately that
  \(\Sproj i\Compl\Sproj j\Compl\Sflip_X=\Sproj j\Compl\Sproj i\) %
  for all \(i,j\in\Eset{0,1}\) and this characterizes \(\Sflip_X\) by
  \Cref{lemma:sproj-it-jointly-monic}.
  One proves in the same way naturality, as well as the equation %
  \(\Sflip_X^2=\Id_{\Sfun^2X}\).
\end{proof}

\begin{lemma}
  The following diagram commutes
  \begin{center}
    \begin{tikzcd}
      \Sfun^2 X
      \ar[rr,"\Sflip_X"]
      \ar[rd,swap,"\Ssum_{\Sfun X}"]
      &[-2em]
      &[-2em]
      \Sfun^2 X
      \ar[ld,"\Sfun\Ssum_X"]\\
      &
      \Sfun X
      &
    \end{tikzcd}
  \end{center}
\end{lemma}
\begin{proof}
  For \(i\in\Eset{0,1}\), we have %
  \begin{align*}
    \Sproj i\Compl(\Sfun\Ssum_X)\Compl\Sflip_X
    &=\Ssum_X\Compl\Sproj i\Compl\Sflip_X
    \text{\quad by naturality of }\Sproj i\\
    &=\Ssum_X\Compl\Sfun\Sproj i
    \text{\quad by definition of }\Sflip_X\\
    &=\Sproj i\Compl\Ssum_{\Sfun X}
      \text{\quad by naturality of }\Ssum
  \end{align*}
  and the triangle commutes by joint monicity of the \(\Sproj i\)'s.
\end{proof}

Now we use these properties of \(\Sflip\) to prove associativity of
our partially defined addition on \(\cL(X,Y)\).

\begin{lemma} %
  \label{lemma:flip-assoc}
  Let \((f_{ij}\in\cL(X,Y))_{(i,j)\in\Eset{0,1}}\) be such that %
  \(f_{i0}\) and \(f_{i1}\) are summable for \(i=0,1\), and the
  corresponding sums \((f_{i0}+f_{i1})_{i=0,1}\) are summable.
  Then \(f_{0j}\) and \(f_{1j}\) are summable for \(j=0,1\), the
  corresponding sums \((f_{0j}+f_{1j})_{j=0,1}\) are summable, and %
  \((f_{00}+f_{01})+(f_{10}+f_{11})=(f_{00}+f_{10})+(f_{10}+f_{11})\).
\end{lemma}
\begin{proof}
  By~\ref{ax:swit}, the morphisms %
  \(\Stuple{f_{00},f_{01}},\Stuple{f_{10},f_{11}}\in\cL(X,\Sfun Y)\)
  are summable; this summability has a witness
  \(\Stuple{\Stuple{f_{00},f_{01}},\Stuple{f_{10},f_{11}}}
  \in\cL(X,\Sfun^2 Y)\) so that we can define 
  \[
    h=\Sflip_Y\Compl\Stuple{\Stuple{f_{00},f_{01}},\Stuple{f_{10},f_{11}}}
    \in\cL(X,\Sfun^2Y)\,.
  \]
  For \(i\in\Eset{0,1}\), let \(h_i=\Sproj i\Compl h\in\cL(X,\Sfun Y)\).
  We have %
  \(\Sproj j\Compl h_i
  =\Sproj i\Compl\Sproj j\Compl
  \Stuple{\Stuple{f_{00},f_{01}},\Stuple{f_{10},f_{11}}}=f_{ji}\) %
  and hence \(f_{0i}\) and \(f_{1i}\) are summable with %
  \(\Stuple{f_{0i},f_{1i}}=h_i\) and sum %
  \begin{align*}
    f_{0i}+f_{1i}
    &=\Ssum_Y\Compl h_i\\
    &=\Ssum_Y
      \Compl\Sproj i\Sflip_Y
      \Compl\Stuple{\Stuple{f_{00},f_{01}},\Stuple{f_{10},f_{11}}}\\
    &=\Ssum_Y\Compl(\Sfun\Sproj i)
      \Compl\Stuple{\Stuple{f_{00},f_{01}},\Stuple{f_{10},f_{11}}}
      \text{\quad by definition of }\Sflip_Y\\
    &=\Sproj i\Compl\Ssum_{\Sfun Y}
      \Compl\Stuple{\Stuple{f_{00},f_{01}},\Stuple{f_{10},f_{11}}}
      \text{\quad by naturality of }\Ssum
  \end{align*}
  so that \(f_{00}+f_{10}\) and \(f_{01}+f_{11}\) are summable with
  \begin{align*}
    \Stuple{f_{00}+f_{10},f_{01}+f_{11}}
    =\Ssum_{\Sfun Y}
      \Compl\Stuple{\Stuple{f_{00},f_{01}},\Stuple{f_{10},f_{11}}}\,.
  \end{align*}
  We have
  \begin{align*}
    (f_{00}+f_{10})+(f_{01}+f_{11})
    &=\Ssum_Y\Compl\Ssum_{\Sfun Y}
    \Compl\Stuple{\Stuple{f_{00},f_{01}},\Stuple{f_{10},f_{11}}}\\
    &=\Ssum_Y\Compl(\Sfun\Ssum_Y)
      \Compl\Stuple{\Stuple{f_{00},f_{01}},\Stuple{f_{10},f_{11}}}
      \text{\quad by naturality of }\Ssum\\
    &=\Ssum_Y
      \Compl\Stuple{\Ssum_Y\Compl\Stuple{f_{00},f_{01}},
      \Ssum_Y\Compl\Stuple{f_{10},f_{11}}}
      \text{\quad by \Cref{lemma:summability-compl}}\\
    &=\Ssum_Y
      \Compl\Stuple{f_{00}+f_{01},f_{10}+f_{11}}\\
    &=(f_{00}+f_{01})+(f_{10}+f_{11})\,.
      \qedhere
  \end{align*}
\end{proof}

\begin{lemma}
  If \(f_0,f_1\in\cL(X,Y)\) are summable, then so are %
  \(\Sfun f_0,\Sfun f_1\in\cL(\Sfun X,\Sfun Y)\) with %
  \(\Stuple{\Sfun f_0,\Sfun f_1}=\Sflip_Y\Compl\Sfun\Stuple{f_0,f_1}\)
  and \(\Sfun f_0+\Sfun f_1=\Sfun(f_0+f_1)\).
\end{lemma}
\begin{proof}
  We have
  \begin{align*}
    \Sproj i\Compl\Sproj j\Compl\Sflip_Y\Compl\Sfun\Stuple{f_0,f_1}
    &=\Sproj j\Compl\Sproj i\Compl\Sfun\Stuple{f_0,f_1}\\
    &=\Sproj j\Compl\Stuple{f_0,f_1}\Compl\Sproj i
      \text{\quad by naturality of }\Sproj i\\
    &=f_j\Compl\Sproj i\\
    &=\Sproj i\Compl\Sfun f_j\\
    &=\Sproj i\Compl\Sproj j\Compl\Stuple{\Sfun f_0,\Sfun f_1}
  \end{align*}
  which proves the first statement.
  The second follows easily from \Cref{lemma:summability-compl} using
  as usual the joint monicity of the \(\Sproj i\)'s.
\end{proof}

\begin{definition}
  A pre-summability structure \((\Sfun,\Sproj0,\Sproj1,\Ssum)\) which
  satisfies~\ref{ax:scom}, \ref{ax:szero} and \ref{ax:swit} is called
  a \emph{summability structure} and a category \(\cL\) equipped with
  a summability structure is called a \emph{summable category}.
\end{definition}

\begin{theorem} %
  \label{lemma:sumcat-enriched-part-com-mon}
  Any summable category is enriched in partial commutative monoids by
  the partial addition operation induced by the summability structure
  (with neutral elements \(0_{X,Y}\)).
\end{theorem}
\begin{proof}
  It suffices to prove the associativity condition
  in~\Cref{def:part-comm-monoids}.
  It results from \Cref{lemma:flip-assoc} and~\ref{ax:szero}, upon
  taking \(f_{01}=0\).
\end{proof}

In particular we can speak of finite summable families of morphisms %
\(\Vect f=(f_i\in\cL(X,Y))_{i\in I}\) and of their sum
\(\sum\Vect f=\sum_{i\in I}f_i\) without particular cautions.

\begin{theorem} %
  \label{th:gen-sum-composition}
  If \((f_i\in\cL(X,Y))_{i\in I}\) and \((g_j\in\cL(Y,Z))_{j\in J}\)
  are finite summable families of morphisms, then the family %
  \((g_j\Compl f_i\in\cL(X,Z))_{(i,j)\in I\times J}\) is summable and %
  \begin{align*}
    \sum_{(i,j)\in I\times J}g_j\Compl f_i
    =(\sum_{j\in J}g_j)\Compl(\sum_{i\in I}f_i)\,.
  \end{align*}
\end{theorem}
\begin{proof}[Proof sketch]
  One takes repetition-free enumerations \((i_l)_{l=1}^m\) and
  \((j_k)_{k=1}^n\) of \(I\) and \(J\) and proves the result by
  induction on \(m+n\) coming back to
  \Cref{def:part-com-mon-general-sums} and using
  \Cref{lemma:summability-compl}.
\end{proof}

\subsection{The monad structure of \(\Sfun\)}
We assume that \(\cL\) is equipped with a summability structure, we
use the notations introduced above.
By~\ref{ax:szero} and~\ref{ax:scom}, there are natural morphisms %
\((\Sin i\in\cL(X,\Sfun X))_{i=0,1}\) given by %
\(\Sin0=\Stuple{\Id_X,0}\) and \(\Sin1=\Stuple{0,\Id_X}\), and we have
\(\Ssum\Compl\Sin i=\Id_X\).

\begin{lemma} %
  \label{lemma:ssum-monad-mult}
  There is a natural morphism
  \(\Sfunmonm_X\in\cL(\Sfun^2X,\Sfun X)\) such that %
  \(\Sproj0\Compl\Sfunmonm=\Sproj0\Compl\Sproj0\) and %
  \(\Sproj1\Compl\Sfunmonm=\Sproj0\Compl\Sproj1+\Sproj1\Compl\Sproj0\).
  In particular, we have \(\Sfunmonm_X\Compl\Sflip_X=\Sfunmonm_X\).
\end{lemma}
\begin{proof}
  By \Cref{th:gen-sum-composition}, the family %
  \((\Sproj i\Compl\Sproj j)_{(i,j)\in\Eset{0,1}^2}\) is summable and
  hence by \Cref{th:part-monoid-summable-subfam} the morphisms
  \(\Sproj0\Compl\Sproj0,\Sproj0\Compl\Sproj1,\Sproj1\Compl\Sproj0
  \in\cL(\Sfun^2X,X)\) are summable so that %
  \(\Sproj0\Compl\Sproj0,\Sproj0\Compl\Sproj1+\Sproj1\Compl\Sproj0\)
  are summable.
  We take
  \(\Sfunmonm_X
  =\Stuple{\Sproj0\Compl\Sproj0,\Sproj0\Compl\Sproj1+\Sproj1\Compl\Sproj0}\).
  Let us prove for instance naturality, so let \(f\in\cL(X,Y)\), we
  have
  \begin{align*}
    \Sfunmonm_Y\Compl(\Sfun^2 f)
    &=\Stuple{\Sproj0\Compl\Sproj0,\Sproj0\Compl\Sproj1+\Sproj1\Compl\Sproj0}
      \Compl(\Sfun^2 f)\\
    &=\Stuple{\Sproj0\Compl\Sproj0\Compl(\Sfun^2 f),
      (\Sproj0\Compl\Sproj1+\Sproj1\Compl\Sproj0)\Compl(\Sfun^2 f)}
    \text{\quad by \Cref{lemma:summability-compl}}\\
    &=\Stuple{\Sproj0\Compl\Sproj0\Compl(\Sfun^2 f),
      \Sproj0\Compl\Sproj1\Compl(\Sfun^2 f)
      +\Sproj1\Compl\Sproj0\Compl(\Sfun^2 f))}
    \text{\quad by \Cref{lemma:summability-compl}}\\
    &=\Stuple{f\Compl\Sproj0\Compl\Sproj0,
      f\Compl\Sproj0\Compl\Sproj1+f\Compl\Sproj1\Compl\Sproj0}
      \text{\quad by naturality of }\Sproj0\text{ and }\Sproj1\\
    &=(\Sfun f)\Compl\Stuple{\Sproj0\Compl\Sproj0,
      \Sproj0\Compl\Sproj1+\Sproj1\Compl\Sproj0}
    \text{\quad by \Cref{lemma:summability-compl}}\\
    &=(\Sfun f)\Compl\Sfunmonm_X\,.
  \end{align*}
  Last we have %
  \(\Sproj0\Compl\Sfunmonm_X\Compl\Sflip_X
  =\Sproj0\Compl\Sproj0\Compl\Sflip_X=\Sproj0\Compl\Sproj0\) %
  by definition of \(\Sflip\) and
  \begin{align*}
    \Sproj1\Compl\Sfunmonm_X\Compl\Sflip_X
    &=(\Sproj1\Compl\Sproj0+\Sproj0\Compl\Sproj1)\Compl\Sflip_X\\
    &=\Sproj1\Compl\Sproj0\Compl\Sflip_X+\Sproj0\Compl\Sproj1\Compl\Sflip_X
      \text{\quad by \Cref{lemma:summability-compl}}\\
    &=\Sproj0\Compl\Sproj1+\Sproj1\Compl\Sproj0
      \text{\quad by definition of }\Sflip\\
    &=\Sproj1\Compl\Sfunmonm_X
      \text{\quad by \Cref{lemma:sumcat-enriched-part-com-mon}}\,.
      \qedhere
  \end{align*}
\end{proof}

\begin{theorem} %
  \label{th:sfun-linear-monad}
  The triple \((\Sfun,\Sin0,\Sfunmonm)\) is a monad.
\end{theorem}
\begin{proof}
  The monad commutations are easily checked, using as usual the joint
  monicity of the \(\Sproj i\)'s.
  For instance, one proves easily that
  \begin{align*}
    (\Sfun\Sfunmonm_X)\Compl\Sfunmonm_X
    &=\Sfunmonm_{\Sfun X}\Compl\Sfunmonm_X\\
    &=\Stuple{\Sproj0\Compl\Sproj0\Compl\Sproj0,
      \Sproj1\Compl\Sproj0\Compl\Sproj0
      +\Sproj0\Compl\Sproj1\Compl\Sproj0
      +\Sproj0\Compl\Sproj0\Compl\Sproj1}\,.
      \qedhere
  \end{align*}
\end{proof}

Another important structure map of \(\Sfun\) was overlooked
in~\cite{Ehrhard22}.

\begin{lemma}
  There is a natural \emph{lift} morphism %
  \(\Sfunlift_X\in\cL(\Sfun X,\Sfun^2X)\) characterized by %
  \(\Sproj i\Compl\Sproj i\Compl\Sfunlift=\Sproj i\)
  and   
  \(\Sproj i\Compl\Sproj{1-i}\Compl\Sfunlift=0\).
\end{lemma}
\begin{proof}
  We have \(\Sin0=\Stuple{\Id_X,0}\) and hence %
  \(\Sin0\Compl\Sproj0=\Stuple{\Sproj0,0}\in\cL(\Sfun X,\Sfun X)\)
  and similarly
  \(\Sin1\Compl\Sproj1=\Stuple{0,\Sproj1}\in\cL(\Sfun X,\Sfun X)\).
  Therefore by~\ref{ax:szero} we have %
  \(\Ssum\Compl\Sin i\Compl\Sproj i=\Sproj i\) for \(i=0,1\) and
  since %
  \(\Sproj 0\) and \(\Sproj 1\) are summable, \ref{ax:swit} implies
  that \(\Sin0\Compl\Sproj0\) and \(\Sin1\Compl\Sproj1\) are summable
  and we set
  \begin{equation*}
    \Sfunlift=\Stuple{\Sin0\Compl\Sproj0,\Sin1\Compl\Sproj1}
    \in\cL(\Sfun X,\Sfun^2X)
  \end{equation*}
  which satisfies the announced condition.
\end{proof}

\begin{remark}
  The lift morphism \(\Sfunlift_X\) plays an important role in tangent
  categories as it allows to express that the differential is linear
  in a certain sense.
  It will play the very same role here.

  Equipped with \(\Ssum\) as counit and \(\Sfunlift\) as
  comultiplication, it is easy to check that \(\Sfun\) is also a
  comonad; it is actually a bimonad (with \(\Sflip\) as distributive
  law) but contrarily to the monad structure, the differential
  distributive law that we will introduce soon does not allow to
  extend the comonad structure of \(\Sfun\) to the Kleisli category of
  \(\oc\_\).
  Such an extension will be possible in the coherent theory of Taylor
  expansion developed in~\cite{EhrhardWalch23b}, and more
  precisely in the analytic situation where any morphism of the
  Kleisli category of \(\oc\_\) is the sum of its Taylor expansion.
  See that paper for more information about bimonads and about the
  extension of the bimonad structure of \(\Sfun\) when it represents a
  countable (instead of binary, as here) notion of summability.
\end{remark}

\subsection{Compatibility of the summability structure with the tensor
  product and the internal hom}
We assume that \(\cL\) has an SMC structure, that is, a distinguished
object \(\Sone\) (tensor unit) and a binary functor
\(\ITens:\cL^2\to\cL\) (tensor product) together with natural
isomorphisms \(\Leftu_X\in\cL(\Tens\Sone X,X)\),
\(\Rightu_X\in\cL(\Tens X\Sone,X)\),
\(\Assoc_{X_1,X_2,X_3}\in
\cL(\Tens{\Tensp{X_1}{X_2}}{X_3},\Tens{X_1}{\Tensp{X_2}{X_3}})\) and
\(\Sym_{X_1,X_2}\in\cL(\Tens{X_1}{X_2},\Tens{X_2}{X_1})\) subject to
the well-known McLane coherence conditions.

\begin{definition}\labeltext{(\(\Sfun\)-\(\otimes\))}{ax:stens}\ref{ax:stens}
  A pre-summability structure \((\cL,\Sproj0,\Sproj1,\Ssum)\) on an SMC
  \(\cL\) satisfies~\ref{ax:stens} if for any summable morphisms
  \(f_0,f_1\in\cL(X,Y)\) and \(g\in\cL(U,V)\), the morphisms %
  \(\Tens{f_0}g,\Tens{f_1}g\in\cL(\Tens XU,\Tens YV)\) are summable
  and satisfy \(\Tens{(f_0+f_1)}g=\Tens{f_0}g+\Tens{f_1}g\).
\end{definition}

\begin{lemma} %
  \label{lemma:sum-tens-dist-gen}
  Let \((f_i\in\cL(X,Y))_{i\in I}\) and \((g_j\in\cL(U,V))_{j\in J}\)
  be finite summable families of morphisms.
  Then the family
  \((\Tens{f_i}{g_j}\in\cL(\Tens XU,\Tens YV))_{(i,j)\in I\times J}\)
  is summable and we have %
  \(\Tens{(\sum_{i\in I}f_i)}{(\sum_{j\in J}g_j)}
  =\sum_{(i,j)\in I\times J\Tens{f_i}{g_j}}\).
\end{lemma}
\begin{proof}[Proof sketch]
  Without loss of generality we can assume that %
  \(I=\Eset{1,\dots,n}\) and \(J=\Eset{1,\dots,p}\) %
  for \(n,p\in\Nat\) and one proves the result by induction on
  \((n,p)\) ordered lexicographically (for instance).
  Notice that when \(I\) or \(J\) is empty, our assumption that the
  \(0\) morphisms satisfy \(\Tens 0g=0\) is essential.
\end{proof}

\begin{definition}
  An SMC \(\cL\) is a \emph{summable SMC} if it is equipped with a
  summability structure which satisfies~\ref{ax:stens}.
\end{definition}

\begin{lemma}
  Let \(\cL\) be a summable SMC.
  There is a natural morphism %
  \[
    \Sfunstr^1_{X_1,X_2}
    \in\cL(\Tens{\Sfun X_1}{X_2},\Sfun\Tensp{X_1}{X_2})
  \]
  such that %
  \((\Sproj i\Compl\Sfunstr^1=\Tens{\Sproj i}{X_2})_{i=0,1}\).
\end{lemma}
\begin{proof}
  We know that \(\Sproj0,\Sproj1\in\cL(\Sfun X_2,X_2)\) are summable, %
  so by~\ref{ax:stens} the morphisms %
  \(\Tens{\Sproj0}{X_2},\Tens{\Sproj1}{X_2}
  \in\cL(\Tens{\Sfun X_1}{X_2},\Tens{X_1}{X_2})\) are summable, we set %
  \(\Sfunstr^1=\Stuple{\Tens{\Sproj0}{X_2},\Tens{\Sproj1}{X_2}}\).
  Let \((f_i\in\cL(X_i,Y_i))_{i=1,2}\), we have
  \begin{align*}
    \Sfunstr^1\Compl\Tensp{\Sfun f_1}{f_2}
    &=\Stuple{\Tensp{\Sproj0}{X_2}\Compl\Tensp{\Sfun f_1}{f_2},
      \Tensp{\Sproj1}{X_2}\Compl\Tensp{\Sfun f_1}{f_2}} 
      \text{\quad by \Cref{lemma:summability-compl}}\\
    &=\Stuple{\Tens{f_1\Compl\Sproj 0}{f_2},\Tens{f_1\Compl\Sproj 1}{f_2}}\\
    &=\Sfun\Tensp{f_1}{f_2}\Compl\Sfunstr^1
      \text{\quad by \Cref{lemma:summability-compl}}
  \end{align*}
  which shows that \(\Sfunstr^1\) is natural.
\end{proof}
One defines %
\(\Sfunstr^2_{X_1,X_2}
\in\cL(\Tens{X_1}{(\Sfun X_2)},\Sfun\Tensp{X_1}{X_2})\) %
by %
\(\Sfunstr^2_{X_1,X_2}
=(\Sfun\Sym_{X_2,X_1})\Compl\Sfunstr^1_{X_2,X_1}\Sym_{X_1,\Sfun X_2}\).

\begin{theorem}
  The following diagrams commute
  \begin{equation*}
    \begin{tikzcd}
      \Tens{X_1}{X_2}
      \ar[r,"\Tens{\Sin0}{X_2}"]
      \ar[dr,swap,"\Sin0"]
      &
      \Tens{\Sfun X_1}{X_2}
      \ar[d,"\Sfunstr^1"]
      \\
      &
      \Sfun\Tensp{X_1}{X_2}
    \end{tikzcd}
    \Treesep
    \begin{tikzcd}
      \Tens{\Sfun^2X_1}{X_2}
      \ar[r,"\Sfunstr^1"]
      \ar[d,swap,"\Tens{\Sfunmonm}{X_2}"]
      &
      \Sfun\Tensp{\Sfun X_1}{X_2}
      \ar[r,"\Sfun\Sfunstr^1"]
      &
      \Sfun^2\Tensp{X_1}{X_2}
      \ar[d,"\Sfunmonm"]
      \\
      \Tens{\Sfun X_1}{X_2}
      \ar[rr,"\Sfunstr^1"]
      &&
      \Sfun\Tensp{X_1}{X_2}
    \end{tikzcd}
  \end{equation*}
  \begin{equation*}
    \begin{tikzcd}
      &[-2em]
      \Tens{\Sfun X_1}{\Sfun X_2}
      \ar[ld,swap,"\Sfunstr^1_{X_1,\Sfun X_2}"]
      \ar[rd,"\Sfunstr^2_{\Sfun X_1,X_2}"]
      &[-2em]\\
      \Sfun\Tensp{X_1}{\Sfun X_2}
      \ar[d,swap,"\Sfun\Sfunstr^2_{X_1,X_2}"]
      &
      &
      \Sfun\Tensp{\Sfun X_1}{X_2}
      \ar[d,"\Sfun\Sfunstr^1_{X_1,X_2}"]\\
      \Sfun^2\Tensp{X_1}{X_2}
      \ar[rr,"\Sflip_{\Tens{X_1}{X_2}}"]
      &&
      \Sfun^2\Tensp{X_1}{X_2}
    \end{tikzcd}
  \end{equation*}
  and hence \((\Sfun,\Sfunstr^1)\) is a commutative strong monad.
\end{theorem}
\begin{proof}
  The diagrams are proven commutative using the joint monicity of the
  \(\Sproj i\)'s.
  Commutativity of the monad means that the following diagram commutes
  \begin{equation*}
    \begin{tikzcd}
      &[-2em]
      \Tens{\Sfun X_1}{\Sfun X_2}
      \ar[ld,swap,"\Sfunstr^1_{X_1,\Sfun X_2}"]
      \ar[rd,"\Sfunstr^2_{\Sfun X_1,X_2}"]
      &[-2em]\\
      \Sfun\Tensp{X_1}{\Sfun X_2}
      \ar[d,swap,"\Sfun\Sfunstr^2_{X_1,X_2}"]
      &
      &
      \Sfun\Tensp{\Sfun X_1}{X_2}
      \ar[d,"\Sfun\Sfunstr^1_{X_1,X_2}"]\\
      \Sfun^2\Tensp{X_1}{X_2}
      \ar[rd,swap,"\Sfunmonm_{\Tens{X_1}{X_2}}"]
      &&
      \Sfun^2\Tensp{X_1}{X_2}
      \ar[ld,"\Sfunmonm_{\Tens{X_1}{X_2}}"]\\
      &
      \Sfun\Tensp{X_1}{X_2}
    \end{tikzcd}
  \end{equation*}
  which results from \Cref{lemma:ssum-monad-mult}.
\end{proof}

\begin{remark}
  It is a standard fact that the common value %
  \(\Sfunlaxt_{X_1,X_2}
  =\Sfunmonm_{\Tens{X_1}{X_2}}
  \Compl(\Sfun\Sfunstr^1_{X_1,X_2})
  \Compl\Sfunstr^2_{\Sfun X_1,X_2}
  =\Sfunmonm_{\Tens{X_1}{X_2}}
  \Compl(\Sfun\Sfunstr^1_{X_1,X_2})
  \Compl\Sfunstr^1_{X_1,\Sfun X_2}
  \in\cL(\Tens{\Sfun X_1}{\Sfun X_2},\Sfun\Tensp{X_1}{X_2})\) %
  (and associated unit \(\Sin0\in\cL(\Sone,\Sfun\Sone)\)) %
  turns \(\Sfun\) into a lax monoidal monad.
  Notice that this morphism is characterized by
  \begin{align*}
    \Sproj0\Compl\Sfunlaxt
    =\Tens{\Sproj0}{\Sproj0}\text{\quad and \quad}
    \Sproj1\Compl\Sfunlaxt
    =\Tens{\Sproj1}{\Sproj0}+\Tens{\Sproj0}{\Sproj1}\,.
  \end{align*}
\end{remark}

If the SMC \(\cL\) is closed, with internal hom of \(X\) and \(Y\)
denoted as \((\Limpl XY,\Evlin)\) where
\(\Evlin\in\cL(\Tens{\Limplp XY}{X},Y)\) is the evaluation morphism
(and, given \(f\in\cL(\Tens ZX,Y)\), we use
\(\Curlin f\in\cL(Z,\Limpl XY)\) for the currying of \(f\)), then we
need a further assumption on the summability structure expressing that
\(\Sfun{\Limplp XY}\) and \(\Limpl X{\Sfun Y}\) are isomorphic.
More precisely, notice that thanks to \ref{ax:stens} we have a
morphism \(\Sfunstr_{X,Y}\) defined as the following composition of
morphisms
\begin{equation*}
  \begin{tikzcd}
    \Tens{\Sfun{\Limplp XY}}{X}
    \ar[r,"\Sfunstr^1_{\Limpl XY,X}"]
    &[1.6em]
    \Sfun{\Tensp{\Limplp XY}{X}}
    \ar[r,"\Sfun\Evlin"]
    &
    \Sfun X
  \end{tikzcd}
\end{equation*}
so that %
\(\Sfunstri_{X,Y}
=\Curlin{\Sfunstr_{X,Y}}\in\cL(\Sfun{\Limplp XY},\Limpl X{\Sfun Y})\).

\begin{definition} %
  \labeltext{(\(\Sfun\)-\(\multimap\))}{ax:slimpl}\ref{ax:slimpl} %
  We say that \(\cL\) satisfies \ref{ax:slimpl} if the morphism %
  \(\Sfunstri_{X,Y}\) is an iso.
  If a summable SMC is closed, we always assume that it
  satisfies \ref{ax:slimpl}.
\end{definition}

\begin{lemma}
  \label{lemma:curlin-additive}
  Let \(\cL\) be a summable SMCC.
  If \((f_i\in\cL(\Tens ZX,Y))_{i=0,1}\) are summable, then so are %
  \((\Curlin{f_i}\in\cL(Z,\Limpl XY))_{i=0,1}\) and we have %
  \(\Curlin{(f_0+f_1)}=\Curlin{f_0}+\Curlin{f_1}\).
\end{lemma}
\begin{proof}
  We have \(\Curlin{\Stuple{f_0,f_1}}\in\cL(Z,\Limpl X{\Sfun Y})\) %
  and hence %
  \(\Invp{\Sfunstri}\Compl\Curlin{\Stuple{f_0,f_1}}
  \in\cL(Z,\Sfun{\Limplp XY})\).
  By naturality of \(\Sfunstri\) we have %
  \(\Sproj i\Compl\Invp{\Sfunstri}\Compl\Curlin{\Stuple{f_0,f_1}}
  =\Limplp X{\Sproj i}\Compl\Curlin{\Stuple{f_0,f_1}} =\Curlin{(\Sproj
    i\Compl\Stuple{f_0,f_1})}=\Curlin{f_i}\) for \(i=0,1\).
  Hence \((\Curlin{f_i})_{i=0,1}\) are summable with %
  \(\Curlin{f_0}+\Curlin{f_1}
  =\Ssum\Compl\Invp{\Sfunstri}\Compl\Curlin{\Stuple{f_0,f_1}}
  =\Curlin{(f_0+f_1)}\) by the same kind of computation.
\end{proof}

\subsection{Summability in a cartesian category}
We assume now that \(\cL\) is cartesian, that is, any finite family
\((X_i)_{i\in I}\) has a cartesian product %
\((\Bwith_{i\in I}X_i,(\Proj i)_{i\in I})\) where the
\(\Proj j\in\cL(\Bwith_{i\in I}X_i,X_j)\) are the projections; we use
\(\Stop\) for the terminal object.
When \((f_i\in\cL(Y,X_i))_{i\in I}\), we use \(\Tuple{f_i}_{i\in I}\)
for the unique morphism \(Y\to\Bwith_{i\in I}X_i\) which,
post-composed with \(\Proj j\), yields \(f_j\).

\begin{definition}\labeltext{(\(\Sfun\)-\(\Bwith\))}{ax:swith}\ref{ax:swith}
  A pre-summability structure \((\cL,\Sproj0,\Sproj1,\Ssum)\)
  satisfies~\ref{ax:swith} or is \emph{cartesian} if, for any finite
  family \(\Vect X=(X_i)_{i\in I}\) of objects, the morphism
  \(\Tuple{\Sfun\Proj i}_{i\in I}\in\cL(\Sfun(\Bwith_{i\in
    I}X_i),\Bwith_{i\in I}\Sfun X_i)\) is an iso.
  We use then %
  \(\Swithiso_{\Vect X}\in\cL(\Bwith_{i\in I}\Sfun
  X_i),\Sfun(\Bwith_{i\in I}X_i))\) for the inverse of %
  \(\Tuple{\Sfun\Proj i}_{i\in I}\), which is a strong monoidal
  structure for the functor \(\Sfun\) wrt.~to the monoidal structure
  induced on \(\cL\) by its cartesian product.
\end{definition}

In the sequel, when dealing with a (pre-)summability structure on a
cartesian category, we always assume that~\ref{ax:swith} holds.

\subsection{Summability in the elementary situation}

It turns out that most non-trivial summability structures result from
very simple properties of the category \(\cL\) that we describe now.
Due to the very simple nature of these properties, we call such
summability structures \emph{elementary}.
We assume to be given an SMC category \(\cL\) with zero-morphisms
which is cartesian%
\footnote{It is not really necessary that all products exist, we only
  need \(\With\Sone\Sone\) to exist, but this assumption is not very
  strong anyway.}.

\begin{remark}
  The internal hom \((\Limpl\Sone X,\Evlin)\) exists: we can take
  \(\Limplp\Sone X=X\) and %
  \(\Evlin=\Rightu_X\in\cL(\Tens X\Sone,X)\).
  We will always use this particular version of this internal hom.
\end{remark}

We set \(\Dbimon=\With\Sone\Sone\), which will play a role similar to
that of an object of infinitesimals in Synthetic Differential
Geometry~\cite{Kock09}.

We assume that, for all object \(X\) of \(\cL\), the internal hom %
\((\Limpl\Dbimon X,\Evlin)\) exists, where %
\(\Evlin\in\cL(\Tens{\Limplp\Dbimon X}{\Dbimon},X)\) is the evaluation
morphisms.
In that way we define a functor \(\Sfun=\Limplp\Dbimon\_:\cL\to\cL\).
Since \(\cL\) has zero-morphisms, we can define %
\(\Win0=\Tuple{\Id_\Sone,0}\), \(\Win1=\Tuple{0,\Id_\Sone}\) and %
\(\Wdiag=\Tuple{\Id_\Sone,\Id_\Sone}\) which all belong to
\(\cL(\Sone,\Dbimon)\).
Notice that if \(f\in\cL(\Sone,\Dbimon)\), then %
\(\Limpl fX\in\cL(\Sfun X,X)\) is a natural transformation. %

\begin{definition}\labeltext{(\(\Dbimon\)-epi)}{ax:depi}\ref{ax:depi}
  The category \(\cL\) is \emph{elementarily pre-summable} if, for any
  object \(X\) of \(\cL\), the morphisms \(\Tens X{\Win0}\) and
  \(\Tens X{\Win1}\) are jointly epic.
\end{definition}

\begin{remark}
  If \(\cL\) satisfies~\ref{ax:depi} then \(\Win0\) and \(\Win1\) are
  jointly epic, and if \(\cL\) is an SMCC then this latter property
  implies~\ref{ax:depi}.
\end{remark}

\begin{lemma}
  If \(\cL\) satisfies~\ref{ax:depi} then %
  \((\Sfun,\Sproj 0=\Limplp{\Win 0}{X}, \Sproj 1=\Limplp{\Win
    1}{X},\Ssum=\Limplp{\Wdiag}X)\) is a pre-summability structure on
  \(\cL\), that is \(\Sproj0\) and \(\Sproj1\) are jointly monic.
  Moreover, the conditions~\ref{ax:scom} and~\ref{ax:szero} hold.
\end{lemma}

In this pre-summability structure, saying that two morphisms
\(f_0,f_1\in\cL(X,Y)\) are summable means that there is
\(f\in\cL(\Tens X\Dbimon,Y)\) such that %
\(f_i=f\Compl\Tensp\Dbimon{\Win i}\Compl\Inv\Rightu\), and then we
have \(f_0+f_1=f\Compl\Tensp\Dbimon{\Wdiag}\Compl\Inv\Rightu\).
We set \(\Stuplet{f_0,f_1}=f\), so that %
\(\Stuple{f_0,f_1}=\A\Curlin{\Stuplet{f_0,f_1}}\).

\begin{lemma} %
  \label{lemma:tens-win-epic}
  For any object \(X\) and any \(n\in\Nat\), the morphisms %
  \((X\ITens\Win{i_1}\ITens\cdots\ITens\Win{i_n}
  )_{\Vect i\in\Eset{0,1}^n}\) %
  are jointly epic.
\end{lemma}
\begin{proof}
  Simple induction on \(n\).
\end{proof}

\begin{definition} %
  \label{def:elementarily-summable-cat}
  We say that \(\cL\) is \emph{elementarily summable} if it is
  elementarily pre-summable, and the induced summability structure
  satisfies~\ref{ax:swit}.
\end{definition}

It is not particularly enlightening to unfold this definition and
express directly the condition~\ref{ax:swit} in terms of \(\Dbimon\),
see~\cite{Ehrhard23a}.

\begin{remark}
  Being an elementary summability category is a \emph{property} of a
  category \(\cL\), and not an additional structure (contrarily to the
  general notion of summability structure).
\end{remark}

\begin{Example}
  Remember that the cartesian product in \(\COH\) of a family
  \((E_i)_{i\in I}\) is %
  \(E=\Bwith_{i\in I}E_i\) defined by
  \(\Web E=\Union_{i\in I}\Web{E_i}\) with %
  \((i,a)\Coh E(i',a')\) if \(i=i'\Implies a\Coh{E_i}a'\) and
  projections %
  \(\Proj i=\Set{((i,a),a)\St i\in I\text{ and }a\in\Web{E_i}}
  \in\COH(E,E_i)\) as easily checked.
  Given \((t_i\in\COH(F,E_i))_{i\in I}\), the unique %
  \(\Tuple{t_i}_{i\in I}\COH(F,E)\) such that %
  \((\Proj j\Compl\Tuple{t_i}_{i\in I}=t_j)_{j\in I}\) is %
  \(\Tuple{t_i}_{i\in I}
  =\Set{(b,(i,a))\St i\in I\text{ and }(b,(i,a))\in t_i}\).

  The category \(\COH\) is elementarily summable.
  It has zero-morphisms (with %
  \(0_{E,F}=\emptyset\in\Cl{\Limpl EF}=\COH(E,F)\)).
  The object \(\Dbimon=\With\Sone\Sone\) can be described by %
  \(\Web\Dbimon=\Eset{0,1}\) with \(0\Coh\Dbimon1\) and we have %
  \(\Win i=\Eset{(\Sonelem,i)}\) and %
  \(\Wdiag=\Eset{(\Sonelem,0),(\Sonelem,1)}\).
  It is easy to check that this category is elementarily summable.
  The induced functor \(\Sfun:\COH\to\COH\) can be described directly
  as follows: \(\Web{\Sfun E}=\Eset{0,1}\times\Web E\) and %
  \((i,a)\Coh{\Sfun E}(i',a')\) if \(a\Coh Ea'\) and
  \(i\not=i'\Implies a\not=a'\).
  Therefore, up to a trivial order isomorphism, %
  \(\Cl{\Sfun E}
  =\Eset{(x_0,x_1)\in\Cl E^2\St x_0\cup x_1\in\Cl E
    \text{ and }x_0\cap x_1=\emptyset}\) (with the product order).
  Therefore two morphisms \(t_0,t_1\in\COH(E,F)\) are summable iff %
  \(t_0\cup t_1\in\COH(E,F)\) (which is not surprising) and
  \(t_0\cap t_1=\emptyset\) (which is more), and then their sum is
  \(t_0\cup t_1\).
  The corresponding witness is %
  \(\Stuple{t_0,t_1}=\Eset{(a,(i,b))
    \St i\in\Eset{0,1}\text{ and } (a,b)\in t_i}
  \in\COH(E,\Sfun F)\).
  Notice that the second definition of \(E_x\) in
  \Cref{sec:coh-space-stable-diff} is directly related to this notion
  of summability: \(E_x=\Eset{x'\St (x,x')\in\Sfun E}\).
\end{Example}

\begin{Example}
  \label{ex:pcoh-cart-smc-elem-sum}
  Another crucial example is that of \emph{probabilistic coherence
    space} (PCS) introduced in~\cite{Girard04a,DanosEhrhard08}.
  A PCS is a pair \(X=(\Web X,\Pcoh X)\) where \(\Web X\) is a set and
  \(\Pcoh X\subseteq\Realpto{\Web X}\) that we consider as a poset
  (whose order relation is the product order) and that we assume to
  satisfy%
  \begin{itemize}
  \item \(\Pcoh X\) is non-empty;
  \item for all \(a\in\Web X\), the set
    \(\Set{x_a\St x\in\Pcoh X}\subseteq\Realp\) is bounded and is not
    reduced to \(\Set 0\);
  \item \(\Pcoh X\) is down-closed (that is if \(x\in\Pcoh X\) and
    \(y\in\Realpto{\Web X}\) satisfy \(y\leq x\) then \(y\in\Pcoh X\))
    and closed under the lubs of monotone chains %
    (that is if \((x(n)\in\Pcoh X)_{n\in\Nat}\) is monotone for the
    pointwise order, then %
    \((\sup_{n\in\Nat} x(n)_a)_{a\in\Web X}\in\Pcoh X\));
  \item \(\Pcoh X\) is closed under barycentric combinations, that
    is, %
    if \(x,y\in\Pcoh X\) and \(p\in\Intercc 01\), then %
    \((1-p)x+py\in\Pcoh X\) (the algebraic operations being defined
    pointwise).
  \end{itemize}
  \begin{remark}
    \label{rk:pcs-bipolar}
    In the literature, the definition of PCS is most often
    based on a duality typical of \LL{}; the present definition is
    equivalent to the duality-based definition as shown
    in~\cite{Girard04a,Ehrhard22}, and is perhaps more intuitive.
  \end{remark}
  A morphism from \(X\) to \(Y\) is a matrix
  \(t\in\Realpto{\Web X\times\Web Y}\) such that for all
  \(x\in\Pcoh X\), one has \(\Matappa tx\in\Pcoh Y\) where
  \(\Matappa tx=(\sum_{a\in\Web X}t_{a,b}x_a)_{b\in\Web Y}\in\Pcoh
  X\).
  These morphisms are in bijective correspondence with the functions
  \(f:\Pcoh X\to\Pcoh Y\) which are monotone, commute with lubs of
  monotone sequences and satisfy \(f((1-p)x+py)=(1-p)f(x)+pf(y)\): %
  given such an \(f\) we define its matrix
  \(t\in\Realpto{\Web X\times\Web Y}\) as follows.
  Let \((a,b)\in\Web X\times\Web Y\), then there is an
  \(\epsilon\in\Realpnz\) such that \(\epsilon\Base a\in\Pcoh X\)
  (where \(\Base a\in\Realpto{\Web X}\) is defined by
  \((\Base a)_{a'}=\Kronecker a{a'}\)) and we set
  \(t_{a,b}=\Inv\epsilon f(\epsilon\Base a)_b\) which does not depend
  on the choice of \(\epsilon\).
  In that way we have defined a category \(\PCOH\) where composition
  is defined by the usual product of matrices: if \(s\in\PCOH(X,Y)\)
  and \(t\in\PCOH(Y,Z)\) then \(\Matapp ts\in\PCOH(X,Z)\) is defined
  by \((\Matapp ts)_{a,c}=\sum_{b\in\Web Y}s_{a,b}t_{b,c}\).
  The identity morphism \(\Id\in\PCOH(X,X)\) is the diagonal matrix,
  \(\Id_{a,a'}=\Kronecker a{a'}\).

  The category \(\PCOH\) is an SMC.
  The tensor unit is \(\Sone=(\Set\Sonelem,\Intercc 01)\) (upon
  identifying \(\Realpto{\Set\Sonelem}\) with \(\Realp\)).
  Given \((x(i)\in\Pcoh{X_i})_{i=1,2}\), we set %
  \(\Tens{x(1)}{x(2)}
  =(x(1)_{a_1}x(2)_{a_2})_{(a_1,a_2)\in\Web{X_1}\times\Web{X_2}}
  \in\Realpto{\Web{X_1}\times\Web{X_2}}\).
  We can define \(\Tens{X_1}{X_2}\) as
  \((\Web{X_1}\times\Web{X_2},P)\) where \(P\) is the least subset of
  \(\Realpto{\Web{X_1}\times\Web{X_2}}\) which contains all the %
  \(\Tens{x(1)}{x(2)}\) for \((x(i)\in\Pcoh{X_i})_{i=1,2}\) and
  satisfies the two closure properties in the definition of PCS (the
  two last conditions).
  To describe more explicitly this operation, it is convenient to
  introduce the PCS \(\Limpl XY=(\Web X\times\Web Y,Q)\) where
  \(Q=\PCOH(X,Y)\).
  Then \(\Limpl X\Sone\) is (trivially) isomorphic to the PCS
  \(\Orth X\) where \(\Web{\Orth X}=\Web X\) and
  \(x'\in\Realpto{\Web X}\) belongs to \(\Pcoh{\Orth X}\) if, for all
  \(x\in\Pcoh X\), one has
  \(\Eval x{x'}=\sum_{a\in\Web X}x_ax'_a\leq 1\).
  It is then possible to prove that \(\Biorth X=X\) (a kind of
  ``bipolar theorem'', see \Cref{rk:pcs-bipolar}).
  Then one has \(\Tens{X_1}{X_2}=\Orth{\Limplp{X_1}{\Orth{X_2}}}\)
  and, based on this property, that equipped with \(\Sone\) and
  \(\ITens\), the category \(\PCOH\) is an SMC.
  This SMC is closed with \((\Limpl XY,\Evlin)\) as internal hom from
  \(X\) to \(Y\), with
  \(\Evlin_{((a,b),a'),b'}=\Kronecker a{a'}\Kronecker b{b'}
  \in\cL(\Tens{\Limplp XY}{X},Y)\) as evaluation morphism, which
  satisfies of course \(\Matappa\Evlin{\Tensp tx}=\Matappa tx\) for
  \(t\in\PCOH(X,Y)\) and \(x\in\Pcoh X\).
  This SMCC is even \Staraut{} with \(\Sbot=\Sone\) as dualizing
  object (this essentially boils down to the fact that
  \(\Biorth X=X\)).
  
  The category \(\PCOH\) has all products: given a family
  \((X_i)_{i\in I}\) of PCS, let \(X\) be defined by
  \(\Web X=\Union_{i\in I}\Set i\times\Web{X_i}\) and %
  \(x\in\Realpto{\Web X}\) belongs to \(\Pcoh X\) if %
  \(((x_{i,a})_{a\in\Web{X_i}}\in\Pcoh{X_i})_{i\in I}\).
  Equipped with %
  \((\Proj i\in\PCOH(X,X_i))_{i\in I}\) %
  defined by \((\Proj i)_{(j,a),a'}=\Kronecker ji\Kronecker a{a'}\), %
  \(X\) is easily seen to be the cartesian product of the \(X_i\)'s,
  and we set \(\Bwith_{i\in I}X_i=X\).
  Notice that \(\Pcohp{\Bwith_{i\in I}X_i}\) is isomorphic (for the
  order and for the barycentric structures) to
  \(\prod_{i\in I}\Pcoh{X_i}\).
  Given \((t_i\in\PCOH(Y,X_i))_{i\in I}\), the unique %
  \(\Tuple{t_i}_{i\in I}\in\PCOH(Y,\Bwith_{i\in I}X_i)\) such that %
  \((\Proj j\Compl\Tuple{t_i}_{i\in I}=t_j)_{j\in I}\) is given by
  \((\Tuple{t_i}_{i\in I})_{b,(j,a)}=(t_j)_{b,a}\) for all \(j\in I\),
  \(b\in\Web Y\) and \(a\in\Web{X_j}\).

  The category \(\PCOH\) has zero morphisms (namely the \(0\) matrix
  which obviously belongs to all homsets \(\PCOH(X,Y)\)).
  The object \(\Dbimon=\With\Sone\Sone\) is described by
  \(\Web\Dbimon=\Set{0,1}\) and \(\Pcoh\Dbimon=\Intercc01^2\) (upon
  identifying \(\Realpto{\Web\Dbimon}\) with \(\Realpto2\)).
  Then \((\Win 0,\Win 1,\Wdiag\in\PCOH(\Sone,\Dbimon))\) are given %
  by  \(\Matappa{\Win 0}u=(u,0)\), \(\Matappa{\Win 1}u=(0,u)\) and %
  \(\Matappa\Wdiag u=(u,u)\), for all \(u\in\Intercc01\).
  Given a PCS \(X\), the PCS \(\Sfun X=\Limplp\Dbimon X\) is given by
  \(\Web{\Sfun X}=\Set{0,1}\times\Web X\) and an element of
  \(\Pcohp{\Sfun X}\) is a \(t\in\Realpto{\Set{0,1}\times\Web X}\)
  such that \(\Matappa t{(\Base 0+\Base 1)}\in\Pcoh X\) where
  \(\Base0+\Base1\in\Pcoh\Dbimon\) corresponds to the pair
  \((1,1)\in\Intercc01^2\).
  In other words
  \begin{align}
    \label{eq:ssum-pcoh-isom}
    \Pcohp{\Sfun X}\Isom\Set{(x(0),x(1))\in\Pcoh X^2\St x(0)+x(1)\in\Pcoh X}
  \end{align}
  for the poset and barycentric structures.
  With this identification, the morphisms %
  \(\Sproj0,\Sproj1,\Ssum\in\PCOH(\Sfun X,X)\) are characterized by %
  \(\Matappa{\Sproj i}{(x(0),x(1))}=x(i)\) and %
  \(\Matappa{\Ssum}{(x(0),x(1)}=x(0)+x(1)\).
  It results easily from these observations that two morphisms
  \(t(0),t(1)\in\PCOH(Y,X)\) are summable iff %
  \(\forall y\in\Pcoh Y\ \Matappa{t(0)}y+\Matappa{t(1)}y\) and the
  corresponding witness \(\Stuple{t(0),t(1)}\in\PCOH(Y,\Sfun X)\) is
  given as a matrix by \(\Stuple{t(0),t(1)}_{b,(i,a)}=t(i)_{b,a}\) and
  characterized by %
  \(\Matappa{\Stuple{t(0),t(1)}}y=(\Matappa{t(0)}y,\Matappa{t(1)}y)\) %
  if we consider \Cref{eq:ssum-pcoh-isom} as an equality.
  From this characterization of summability and witnesses, if follows
  easily that \(\PCOH\) is an elementarily summable category.
\end{Example}

\begin{remark}
  If \(\cL\) is additive, that is, enriched in commutative monoids,
  then it is well-known that finite cartesian products are also
  coproducts.
  It follows that, in that case, \(\Dbimon=\Plus\Sone\Sone\) and hence
  \(\Sfun X=\Limplp\Dbimon X\) is canonically isomorphic to
  \(\With XX\) and the summability structure is trivial: all pairs of
  morphisms \(f_0,f_1\in\cL(X,Y)\) are summable with witness
  \(\Tuple{f_0,f_1}\in\cL(X,\With YY)\) and sum \(f_0+f_1\) (the
  addition provided by the enrichment).
\end{remark}

\begin{theorem}
  Any elementarily summable category satisfies~\ref{ax:swith}
  and~\ref{ax:stens}.
\end{theorem}
\begin{proof}[Proof sketch]
  The first property results from the fact that \(\Sfun\) is the right
  adjoint to the functor \(\Tens\_\Dbimon\) and as such preserves all
  existing limits.
  Next take \((f_i\in\cL(X,Y))_{i\in\Eset{0,1}}\) be summable so that %
  \(\Stuplet{f_0,f_1}\in\cL(\Tens{X}{\Dbimon},Y)\), and let %
  \(g\in\cL(U,V)\).
  We have %
  \((\Tens{\Stuplet{f_0,f_1}}{g})\Compl\Tensp X{\Sym_{U,\Dbimon}}
  \in\cL(X\ITens U\ITens\Dbimon,\Tens YV)\).
  It is easily checked that %
  \((\Tens{\Stuplet{f_0,f_1}}{g})\Compl\Tensp X{\Sym_{U,\Dbimon}}
  =\Stuplet{\Tens{f_0}g,\Tens{f_1}g}\).
\end{proof}

We know by \Cref{th:sfun-linear-monad} that \(\Sfun\) has a canonical
(bi)monad structure.
In the elementary situation we can describe more directly this
structure by means of a (bi)monoid structure on \(\Dbimon\).
We describe first the comonoid structure.

\begin{proposition}
  \label{prop:dbimon-comonoid}
  There is a unique \(\Dbimonm\in\cL(\Dbimon,\Tens\Dbimon\Dbimon)\)
  such that \(\Dbimonm\Compl\Win0=\Tens{\Win0}{\Win0}\) and
  \(\Dbimonm\Compl\Win1=\Tens{\Win1}{\Win0}+\Tens{\Win0}{\Win1}\) %
  (keeping the iso %
  \(\Leftu_1=\Rightu_1\in\cL(\Tens{\Sone}{\Sone},\Sone)\) implicit).
  The triple \((\Dbimon,\Proj0,\Dbimonm)\) is a commutative comonoid.
\end{proposition}
\begin{proof}
  Clearly \(\Win0,\Win1\in\cL(\Sone,\Dbimon)\) are summable %
  (with \(\Stuplet{\Win0,\Win1}=\Id_{\Dbimon}\) and %
  \(\Win0+\Win1=\Wdiag\)).
  So by \Cref{lemma:sum-tens-dist-gen} the morphisms %
  \((\Tens{\Win i}{\Win j}
  \in\cL(\Tens\Sone\Sone,\Tens\Dbimon\Dbimon))_{(i,j)\in\Eset{0,1}^2}\) %
  are summable.
  By \Cref{th:part-monoid-summable-subfam}, the morphisms %
  \(\Tens{\Win0}{\Win0}\), \(\Tens{\Win1}{\Win0}\) %
  and \(\Tens{\Win0}{\Win1}\) are summable, and hence the morphisms %
  \(\Tens{\Win0}{\Win0}\) and %
  \(\Tens{\Win1}{\Win0}+\Tens{\Win0}{\Win1}\) are summable.
  This gives us a witness %
  \(\Stuplet{\Tens{\Win0}{\Win0},
    \Tens{\Win1}{\Win0}+\Tens{\Win0}{\Win1}}
  \in\cL(\Tens{\Tens\Sone\Sone}\Dbimon,\Tens\Dbimon\Dbimon)\)
  and we set %
  \[
    \Dbimonm =\Stuplet{\Tens{\Win0}{\Win0},
      \Tens{\Win1}{\Win0}+\Tens{\Win0}{\Win1}}\Compl\Leftu\Compl\Leftu
    \in\cL(\Dbimon,\Tens\Dbimon\Dbimon)
  \]
  which obviously satisfies the announced equations.
  Uniqueness results from the joint epicity of \(\Win0\) and
  \(\Win1\).
  The fact that we define in that way a commutative comonoid is easily
  checked, again by joint epicity of \(\Win0\).
  For instance (keeping the associator \(\Assoc\) implicit), the
  morphism %
  \(\Dbimonm^{(3)} =\Tensp\Dbimonm\Dbimon\Compl\Dbimonm
  =\Tensp\Dbimon\Dbimonm\Compl\Dbimonm
  \in\cL(\Dbimon,\Tens\Dbimon{\Tens\Dbimon\Dbimon})\) %
  is characterized by %
  \(\Dbimonm^{(3)}\Compl\Win0=\Tens{\Win0}{\Tens{\Win0}{\Win0}}\) %
  and
  \(\Dbimonm^{(3)}\Compl\Win1
  =\Tens{\Win1}{\Tens{\Win0}{\Win0}}
  +\Tens{\Win0}{\Tens{\Win1}{\Win0}}
  +\Tens{\Win0}{\Tens{\Win0}{\Win1}}\).
  The commutations involving the counit \(\Proj0\) result from the
  fact that %
  \(\Proj i\Compl\Win j=\Kronecker ij\Id_\Sone\).
\end{proof}

\begin{theorem}
  There is a unique \(\Dbimonl\in\cL(\Tens\Dbimon\Dbimon,\Dbimon)\)
  such that %
  \(\Dbimonl\Compl\Tensp{\Win i}{\Win j}
  =\Kronecker ij\Win i\) %
  (keeping the iso %
  \(\Leftu_1=\Rightu_1\in\cL(\Tens{\Sone}{\Sone},\Sone)\) implicit).
  The triple \((\Dbimon,\Wdiag,\Dbimonl)\) is a commutative monoid,
  and \((\Dbimon,\Wdiag,\Dbimonl,\Proj0,\Dbimonm)\) is a bicommutative
  bimonoid.
\end{theorem}
\begin{proof}
  Remember that \(\Dbimon=\With\Sone\Sone\), so we can set %
  \(\Dbimonl
  =\Tuple{\Leftu_\Sone\Compl\Tensp{\Proj0}{\Proj0},
    \Leftu_\Sone\Compl\Tensp{\Proj1}{\Proj1}}\)
  which obviously satisfies the announced property, which characterizes
  \(\Dbimonl\) uniquely by \Cref{lemma:tens-win-epic}.
  The fact that we define in that way a commutative monoid is easy to
  check.
  For instance (keeping implicit the associator \(\Assoc\) and the
  \(\Leftu\) isos) we have
  \(\Dbimonl^{(3)} =\Dbimonl\Compl\Tensp\Dbimonl\Dbimon
  =\Dbimonl\Compl\Tensp\Dbimon\Dbimonl
  =\Tuple{\Proj0\ITens\Proj0\ITens\Proj0,\Proj1\ITens\Proj1\ITens\Proj1}
  \in\cL(\Dbimon\ITens\Dbimon\ITens\Dbimon,\Dbimon)\).
  The commutations of the diagram involving the unit \(\Wdiag\) result
  from the fact that \(\Proj i\Compl\Wdiag=\Id_\Sone\) for \(i=0,1\).
  The last statement means that the following diagrams commute
  \begin{equation*}
    \begin{tikzcd}
      \Sone\ar[rr,"\Id"]\ar[rd,swap,"\Wdiag"]
      &[-2em]&[-2em]
      \Sone\\
      &
      \Dbimon
      \ar[ur,swap,"\Proj0"]
      &
    \end{tikzcd}
    \Treesep\Treesep
    \begin{tikzcd}
      \Tens\Dbimon\Dbimon
      \ar[d,swap,"\Dbimonl"]
      \ar[r,"\Tens{\Dbimonm}{\Dbimonm}"]
      &
      \Tens{\Tens\Dbimon\Dbimon}{\Tens\Dbimon\Dbimon}
      \ar[dd,"\Sym_{2,3}"]
      \\
      \Dbimon
      \ar[d,swap,"\Dbimonm"]
      &\\
      \Tens\Dbimon\Dbimon
      &
      \Tens{\Tens\Dbimon\Dbimon}{\Tens\Dbimon\Dbimon}
      \ar[l,swap,"\Tens\Dbimonl\Dbimonl"]
    \end{tikzcd}
  \end{equation*}
  where \(\Sym_{2,3}\) is the McLane iso which exchanges the 2 central
  factors of the quaternary tensor product.
  The first diagram obviously commutes.
  For the second we use \Cref{lemma:tens-win-epic}.
  We have
  \begin{align*}
    \Dbimonm\Compl\Dbimonl\Compl\Tensp{\Win i}{\Win j}
    &=\Dbimonm
      \Compl\Kronecker ij\Win i\\
    &=\Kronecker ij\sum_{\Biind{(l,r)\in\Eset{0,1}^2}{l+r=i}}
      \Tens{\Win l}{\Win r}
  \end{align*}
  and
  \begin{align*}
    \Tensp\Dbimonl\Dbimonl
    \Compl\Sym_{2,3}
    \Compl\Tensp\Dbimonm\Dbimonm
    \Compl\Tensp{\Win i}{\Win j}
    &=\Tensp\Dbimonl\Dbimonl
      \Compl\Sym_{2,3}
      \Compl\Tensp
      {(\sum_{\Biind{(l,r)\in\Eset{0,1}^2}{l+r=i}}
      \Tens{\Win l}{\Win r})}
      {(\sum_{\Biind{(l',r')\in\Eset{0,1}^2}{l'+r'=j}}
      \Tens{\Win l}{\Win r})}\\
    &=\Tensp\Dbimonl\Dbimonl
      \Compl\Sym_{2,3}
      \Compl
      \sum_{\Biind{(l,r,l',r')\in\Eset{0,1}^4}{l+r=i,\ l'+r'=j}}
      \Win l\ITens\Win r\ITens\Win{l'}\ITens\Win{r'}\\
    &=\Tensp\Dbimonl\Dbimonl
      \Compl
      \sum_{\Biind{(l,r,l',r')\in\Eset{0,1}^4}{l+r=i,\ l'+r'=j}}
      \Win l\ITens\Win{l'}\ITens\Win{r}\ITens\Win{r'}\\
    &=\sum_{\Biind{(l,r,l',r')\in\Eset{0,1}^4}{l+r=i,\ l'+r'=j}}
      \Kronecker l{l'}\Kronecker r{r'}
      \Win l\ITens\Win{r}\\
    &=\Kronecker ij\sum_{\Biind{(l,r)\in\Eset{0,1}^2}{l+r=i}}
      \Tens{\Win l}{\Win r}
      \qedhere
  \end{align*}
\end{proof}
Then it can be checked that the bimonad structure of \(\Sfun\) is
induced by this bimonoid structure of \(\Dbimon\): %
the unit of the monad is %
\(\Sin0=\Limplp{\Proj0}X\in\cL(\Sfun X,X)\) (identifying \(X\) and
\(\Limpl\Sone X\)) and its multiplication is %
\(\Sfunmonm=\Limplp\Dbimonm X\in\cL(\Sfun X,\Sfun^2 X)\) %
(identifying \(\Limpl{\Tens\Dbimon\Dbimon}X\) and %
\(\Sfun^2X=\Limplp{\Dbimon}{\Limplp\Dbimon X}\) which are canonically
isomorphic).
The comonad structure of \(\Sfun\) can be described similarly using
the unit \(\Wdiag\) and the product \(\Dbimonl\) of the monoid
structure of \(\Dbimon\).
The distributive law of the bimonad, which is the flip isomorphism %
\(\Sflip_X\in\cL(\Sfun^2X,\Sfun^2X)\), is obtained similarly from the %
braiding of the SMC structure of \(\cL\): %
\(\Sflip_X=\Limplp{\Sym_{\Dbimon,\Dbimon}}{X}\) %
(leaving again implicit the canonical iso between
\(\Limpl{\Tens\Dbimon\Dbimon}X\) and \(\Sfun^2X\)).

\begin{remark}
  The general categorical concept of \emph{mate} provides a systematic
  understanding of this correspondence.
  For instance, the commutative comonoid structure of \(\Dbimon\)
  allows to define very easily a comonad \(\Sfun^{\mathord{\ITens}}\) on
  \(\cL\) such that \(\Sfun^{\mathord{\ITens}} X=\Tens X\Dbimon\).
  Then the monad \((\Sfun,\Sin0,\Sfunmonm)\) is the mate of the
  comonad \(\Sfun^{\mathord{\ITens}}\) through the adjunction %
  \(\Tens\_\Dbimon\Adj\Limpl\Dbimon\_\).
  This point of view is developed in~\cite{EhrhardWalch23}.
\end{remark}

\begin{Example}
  In \(\COH\), the bimonoid structure of \(\Dbimon\) can be described
  as follows.
  \begin{itemize}
  \item Counit %
    \(\Proj0=\{(0,\Sonelem)\}\in\COH(\Dbimon,\Sone)\);
  \item Comultiplication %
    \(\Dbimonm=\Set{(0,(0,0)),(1,(1,0)),(1,(0,1))}
    \in\COH(\Dbimon,\Tens\Dbimon\Dbimon)\);
  \item Unit %
    \(\Wdiag=\Set{(\Sonelem,0),(\Sonelem,1)}\in\COH(\Sone,\Dbimon)\);
  \item Multiplication %
    \(\Dbimonl=\Set{((0,0),0),((1,1),1))}
    \in\COH(\Tens\Dbimon\Dbimon,\Dbimon)\).
  \end{itemize}
\end{Example}

\begin{Example}
  Let us describe the bimonoid structure of \(\Dbimon\) in \(\PCOH\).
  An element \(u\) of \(\Pcoh{\Dbimon}\) can be written uniquely %
  \(u=u_0\Base0+u_1\Base1\) where \(u_0,u_1\in\Intercc01\), so that %
  \(\Pcoh\Dbimon\Isom\Intercc01^2\) as already mentioned.
  \begin{itemize}
  \item The unit \(\Proj0\in\PCOH(\Dbimon,\Sone)\) is characterized by %
    \(\Matappa{\Proj0}{u}=u_0\);
  \item the comultiplication
    \(\Dbimonm\in\PCOH(\Dbimon,\Tens\Dbimon\Dbimon)\) %
    is characterized by %
    \(\Matappa\Dbimonm
    u=u_0\Base{(0,0)}+u_1(\Base{(1,0)}+\Base{(0,1)})\);
  \item the unit \(\Wdiag\in\PCOH(\Sone,\Dbimon)\) is characterized by %
    \(\Matappa{\Wdiag}p=p(\Base0+\Base1)\);
  \item the multiplication
    \(\Dbimonl\in\PCOH(\Tens\Dbimon\Dbimon,\Dbimon)\) %
    is characterized by %
    \(\Matappa\Dbimonl{\Tensp uv}=u_0v_0\Base0+u_1v_1\Base1\).
  \end{itemize}

\end{Example}

\section{The differential structure}
\label{sec:diff-structure}

The following categorical concept is a basic infrastructure which is
pervasive in the abstract description of denotational models of \LL.
\begin{definition} %
  \label{def:resource-category}
  A \emph{resource category} is a category \(\cL\) such that
  \begin{itemize}
  \item \(\cL\) is an SMC with zero-morphisms;
  \item \(\cL\) is cartesian;
  \item \(\cL\) is equipped with a \emph{resource modality}, that is a
    tuple %
    \((\Oc,\Der{},\Digg{},\Seelyz,\Seelyt)\) where \(\Oc:\cL\to\cL\)
    is a functor, \(\Der{}\) (dereliction) and \(\Digg{}\) (digging)
    are respectively the counit and the comultiplication of a comonad
    structure on this functor and \((\Seelyz\in\cL(\Sone,\Oc\Stop))\),
    an iso, and \(\Seelyt_{X_1,X_2}\in\cL(\Tens{\Oc X_1}{\Oc X_2})\),
    a natural iso, turn \(\Oc{}\) into a symmetric monoidal comonad
    from the SMC \((\cL,\Sone,\ITens)\) to the SMC
    \((\cL,\Stop,\Bwith)\).
    These isos are called the \emph{Seely isomorphisms} of \(\cL\).
  \end{itemize}
\end{definition}

We assume that \(\cL\) is such a resource category.
The resource modality induces a Kleisli category %
\(\Kloc\cL\) whose objects are those of \(\cL\) and where %
\(\Kloc\cL(X,Y)=\cL(\Oc X,Y)\).
In this category the identity morphisms are \(\Klocid_X=\Der X\) %
and composition of \(f\in\Kloc\cL(X,Y)\) and \(g\in\Kloc\cL(Y,Z)\) %
is defined by \(g\Comp f=g\Compl(\Oc f)\Compl\Digg X\).

The basic intuition in this situation is that the morphisms of \(\cL\)
are linear whereas \(\Kloc\cL\) is a category of nonlinear morphisms.
Here the word ``linear'' can be used in its algebraic and its computer
science meaning.
This intuition is supported by the fact that there is a functor %
\(\Derfun:\cL\to\Kloc\cL\) which acts as the identity on objects and
maps \(f\in\cL(X,Y)\) to \(\Derfun(f)=f\Compl\Der X\in\Kloc\cL(X,Y)\).
This functor is not necessarily faithful (it is, in most known
categorical models of \LL), but it should nevertheless be considered
as a kind of ``inclusion'' of \(\cL\) into the larger \(\Kloc\cL\).

We assume from now on that \(\cL\) is equipped with a summability
structure (remember that this means in particular that~\ref{ax:stens}
and~\ref{ax:swith} hold).

The main idea of \CD{} is to associate with any nonlinear morphism
\(f\in\Kloc\cL(X,Y)\) a ``derivative'' %
\(\Dfun f\in\Kloc\cL(\Sfun X,\Sfun Y)\) which intuitively maps a
summable pair \((x_0,x_1)\) of elements of \(X\) to the summable pair
\((f(x_0),f'(x_0)\cdot x_1)\), and the chain rule of Calculus tells us
that this operation \(\Dfun\) should be functorial.
In other words \(\Dfun\) should be an extension of the functor \(\Sfun\)
to %
\(\Kloc\cL\) in the sense that if \(f\in\cL(X,Y)\), one has %
\(\Dfun(\Derfun(f))=\Derfun(\Sfun f)\).
Intuitively this condition means that the derivative of a linear map
is the map itself.
It is known that such extensions are in one-to-one correspondence
with \emph{distributive laws} between the functor \(\Sfun\) and the
comonad \(\Oc{\_}\).

\begin{definition} %
  \labeltext{(\(\Sdiff\)-chain)}{ax:sdchain}\ref{ax:sdchain}
  A pre-differential structure on \(\cL\) is a distributive law
  between the functor \(\Sfun\) and the comonad \(\Oc{\_}\), that is, a
  natural transformation \(\Sdiff_X\in\cL(\Oc\Sfun X,\Sfun\Oc X)\)
  such that the following diagrams commute
  \begin{equation*}
    \begin{tikzcd}
      \Oc\Sfun X
      \ar[rr,"\Sdiff_X"]
      \ar[dr,swap,"\Der{\Sfun X}"]
      &[-2em]&[-2em]
      \Sfun\Oc X
      \ar[dl,"\Sfun\Der X"]\\
      &
      \Sfun X
      &
    \end{tikzcd}
    \Treesep
    \begin{tikzcd}
      \Oc\Sfun X
      \ar[rr,"\Sdiff_X"]
      \ar[d,swap,"\Digg{\Sfun X}"]
      &&
      \Sfun\Oc X
      \ar[d,"\Sfun\Digg X"]\\
      \Occ{\Sfun X}
      \ar[r,"\Oc\Sdiff_X"]
      &
      \Oc\Sfun\Oc X
      \ar[r,"\Sdiff_{\Oc X}"]
      &
      \Sfun\Occ X
    \end{tikzcd}
  \end{equation*}
\end{definition}

Then the extended functor \(\Dfun:\Kloc\cL\to\Kloc\cL\) is defined by %
\(\Dfun X=\Sfun X\), and
\(\Dfun f=(\Sfun f)\Compl\Sdiff_X\in\cL(\Oc\Sfun X,\Sfun Y)\) for %
\(f\in\cL(\Oc X,Y)\).

This simple condition is not sufficient for specifying a differential
operation.
Here are the additional conditions.

\begin{definition} %
    \labeltext{(\(\Sdiff\)-local)}{ax:sdloc}\ref{ax:sdloc}
    \begin{equation*}
      \begin{tikzcd}
      \Oc\Sfun X
      \ar[rr,"\Sdiff_X"]
      \ar[dr,swap,"\Oc\Sproj0"]
      &[-2em]&[-2em]
      \Sfun\Oc X
      \ar[dl,"\Sproj0"]\\
      &
      \Oc X
      &
      \end{tikzcd}
    \end{equation*}
\end{definition}

\begin{definition} %
  \labeltext{(\(\Sdiff\)-add)}{ax:sdadd}\ref{ax:sdadd} %
  The natural transformation \(\Sdiff\) is also a distributive law
  between the functor \(\oc\) and the monad
  \((\Sfun,\Sin0,\Sfunmonm)\), that is
  \begin{equation*}
    \begin{tikzcd}
    &[-2em]
      \Oc X
      \ar[ld,swap,"\Oc\Sin0"]
      \ar[rd,"\Sin0"]
      &[-2em]
      \\
      \Oc\Sfun X
      \ar[rr,"\Sdiff_X"]
      &&
      \Sfun\Oc X
    \end{tikzcd}
    \Treesep
    \begin{tikzcd}
      \Oc\Sfun^2X
      \ar[r,"\Sdiff_{\Sfun X}"]
      \ar[d,swap,"\Oc\Sfunmonm_X"]
      &
      \Sfun\Oc\Sfun X
      \ar[r,"\Sfun\Sdiff_X"]
      &
      \Sfun^2\Oc X
      \ar[d,"\Sfunmonm_{\Oc X}"]
      \\
      \Oc\Sfun X
      \ar[rr,"\Sdiff_X"]
      &&
      \Sfun\Oc X
    \end{tikzcd}
  \end{equation*}
\end{definition}
This means that the comonad \(\Oc\) can be extended to the Kleisli
category of the monad \(\Sfun\).
Due to~\ref{ax:stens} and~\ref{ax:swith}, this latter Kleisli category
is monoidal and cartesian so that, when~\ref{ax:sdadd} holds, it
becomes a resource category which can be understood as a categorical
version of Clifford's ring of \emph{dual numbers}.

More concretely, the condition~\ref{ax:sdadd} means that derivatives
are additive morphisms, that is, preserve \(0\) and (the partially
defined) addition of morphisms.

\begin{definition} %
  \labeltext{(\(\Sdiff\)-\(\IWith\))}{ax:sdwith}\ref{ax:sdwith} %
  \begin{equation*}
    \begin{tikzcd}
      \Sone
      \ar[r,"\Sin0"]\ar[d,swap,"\Seelyz"]
      &
      \Sfun\Sone
      \ar[dd,"\Sfun\Seelyz"]
      \\
      \Oc\Stop
      \ar[d,swap,"\Oc0"]
      &
      \\
      \Oc\Sfun\Stop
      \ar[r,"\Sdiff_\Stop"]
      &
      \Sfun\Oc\Stop
    \end{tikzcd}
    \Treesep
    \begin{tikzcd}
      \Tens{\Oc\Sfun X_1}{\Oc\Sfun X_1}
      \ar[r,"\Tens{\Sdiff_{X_1}}{\Sdiff_{X_2}}"]
      \ar[d,swap,"\Seelyt_{\Sfun X_1,\Sfun X_2}"]
      &[1em]
      \Tens{\Sfun\Oc X_1}{\Sfun\Oc X_2}
      \ar[r,"\Sfunlaxt_{\Oc X_1,\Oc X_2}"]
      &[1em]
      \Sfun\Tensp{\Oc X_1}{\Oc X_2}
      \ar[dd,"\Sfun\Seelyt_{X_1,X_2}"]\\
      \Oc\Withp{\Sfun X_1}{\Sfun X_2}
      \ar[d,swap,"\Oc{\Swithiso_{X_1,X_2}}"]
      &&
      \\
      \Oc{\Sfun\Withp{X_1}{X_2}}
      \ar[rr,"\Sdiff_{\With{X_1}{X_2}}"]
      &&
      \Sfun{\Oc\Withp{X_1}{X_2}}
    \end{tikzcd}
  \end{equation*}
\end{definition}
%
This condition means that the differential structure is compatible
with the strong monoidal structure of the resource category \(\cL\).
It becomes quite important when the SMC \(\cL\) is assumed to be
closed since, in that situation, this strong monoidal structure turns
\(\Kl\cL\) into a cartesian closed category.
%

\begin{definition} %
  \labeltext{(\(\Sdiff\)-Schwarz)}{ax:sdschwarz}\ref{ax:sdschwarz} %
  \begin{equation*}
    \begin{tikzcd}
      \Oc\Sfun^2 X
      \ar[r,"\Sdiff_{\Sfun X}"]
      \ar[d,swap,"\Oc\Sflip_X"]
      &
      \Sfun\Oc\Sfun X
      \ar[r,"\Sfun\Sdiff_X"]
      &
      \Sfun^2\Oc X
      \ar[d,"\Sflip_{\Oc X}"]\\
      \Oc\Sfun^2 X
      \ar[r,"\Sdiff_{\Sfun X}"]
      &
      \Sfun\Oc\Sfun X
      \ar[r,"\Sfun\Sdiff_X"]
      &
      \Sfun^2\Oc X
    \end{tikzcd}
  \end{equation*}
\end{definition}
This condition means that the second derivative is a \emph{symmetric}
bilinear function: %
\(\frac{\partial^2f(x_1,x_2)}{\partial x_1\partial x_2}\cdot(u_1,u_2)
=\frac{\partial^2f(x_1,x_2)}{\partial x_2\partial x_1}\cdot(u_2,u_1)\).

The last condition was overlooked in~\cite{Ehrhard23a}, but the
corresponding condition was already recognized as important in the
theory of tangent categories~\cite{Rosicky84}.
\begin{definition} %
  \labeltext{(\(\Sdiff\)-lift)}{ax:sdlift}\ref{ax:sdlift} %
  \begin{equation*}
    \begin{tikzcd}
      \Oc\Sfun X
      \ar[rr,"\Sdiff_X"]
      \ar[d,swap,"\Oc\Sfunlift_X"]
      &&
      \Sfun\Oc X
      \ar[d,"\Sfunlift_{\Oc X}"]
      \\
      \Oc\Sfun^2X
      \ar[r,"\Sdiff_{\Sfun X}"]
      &
      \Sfun\Oc\Sfun X
      \ar[r,"\Sfun\Sdiff_X"]
      &
      \Sfun^2\Oc X
    \end{tikzcd}
  \end{equation*}
\end{definition}
Keeping in mind that \(\Sfunlift\) is the comultiplication of the
bimonad \(\Sfun\), one might expect the commutation corresponding to
the counit \(\Ssum\) of that comonad to commute, that is
\begin{equation*}
  \begin{tikzcd}
    \Oc\Sfun X\ar[rr,"\Sdiff_X"]\ar[rd,swap,"\Oc\Ssum_X"]
    &[-2em]&[-2em]
    \Sfun\Oc X
    \ar[dl,"\Ssum_{\Oc X}"]
    \\
    &
    X
    &
  \end{tikzcd}
\end{equation*}
but this would be too strong a requirement in the present setting as
it would require intuitively that all morphisms \(f\in\Kloc\cL(X,Y)\)
satisfy \(f(x_0+x_1)=f(x_0)+f'(x_0)\cdot x_1\) (for all summable pair
\((x_0,x_1)\) of elements of \(X\)), that is, are affine.
So, in \CD, this latter commutation is not required.
In the infinitary setting of \cite{EhrhardWalch23b}, it expresses that
morphisms are analytic in the sense that they coincide with their
Taylor expansion, so this commutation will be the an essential
ingredient in the definition of a \emph{coherent analytic category}.

\subsection{The induced differentiation monad}

The axiom \ref{ax:dchain} exactly means that \(\Sfun:\cL\to\cL\) can
be extended to a functor \(\Dfun:\Kl\cL\to\Kl\cL\).
One speaks here of extension because we consider \(\cL\) as a
``subcategory'' of \(\Kl\cL\) (the inclusion being the functor
\(\Derfun:\cL\to\Kl\cL\)).
Concretely, \(\Dfun\) is defined on objects by \(\A\Dfun X=\Sfun X\),
and given \(f\in\Kl\cL(X,Y)=\cL(\Oc X,Y)\), one sets
\begin{equation*}
  \A\Dfun f=(\A\Sfun f)\Compl\Sdiff_X\,.
\end{equation*}

\begin{proposition}
  The operation \(\Dfun\) is a functor \(\Kl\cL\to\Kl\cL\) which
  extends \(\Sfun\) in the sense that for any \(f\in\cL(X,Y)\), one
  has %
  \(\Ap\Dfun{\A\Derfun f}=\Ap\Derfun{\A\Sfun f}\).
\end{proposition}
This is completely standard in the theory of distributive laws.

We define %
\(\Dmonu_X=\A\Derfun{\Sin0}\in\Kl\cL(X,\Dfun X)\) and %
\(\Dmonm_X=\A\Derfun{\Sfunmonm_X}\in\Kl\cL(\A{\Dfun^2}X,\A\Derfun X)\).

\begin{proposition}
  The morphisms \(\Dmonu_X\in\Kl\cL(X,\Dfun X)\) and %
  \(\Dmonm_X\in\Kl\cL(\A{\Dfun^2}X,\A\Dfun X)\) are natural in \(X\)
  and turn the functor \(\Dfun\) into a monad on \(\Kl\cL\).
\end{proposition}
\begin{proof}[Proof sketch]
  The only non-trivial properties are the naturality of %
  \(\Dmonu\) and \(\Dmonm\).
  They result from \ref{ax:dadd}.
\end{proof}

\begin{remark}
  Intuitively,
  \begin{align*}
    \A\Dfun f(x_0,x_1) &= (f(x_0),f'(x_0)\cdot x_1)\\
    \Dmonu_X(x) &= (x,0)\\
    \Dmonm_X((x_{00},x_{01}),(x_{00},x_{01}))
    &= (x_{00},x_{01}+x_{10})
  \end{align*}
  and the naturality of \(\Dmonu\) and \(\Dmonm\) means that %
  \(f'(x)\cdot 0=0\) and %
  \(f'(x_{00})\cdot(x_{01}+x_{10})
  =f'(x_{00})\cdot x_{01}+f'(x_{00})\cdot x_{10}\).
\end{remark}

\subsection{Partial derivatives}
Given \(f\in\Kl\cL(X_1\IWith\cdots\IWith X_n,Y)\), we have seen how to
define the global differential %
\(\Dfun f\in\Kl\cL(\Sfun(X_1\IWith\cdots\IWith X_n),\Sfun Y)\) %
of \(f\), that is (up to the iso stipulated by \ref{ax:swith}),
\(\Dfun f\in\Kl\cL(\Sfun X_1\IWith\cdots\IWith \Sfun X_n,\Sfun
Y)\), %
which intuitively maps \((x_1,u_1),\dots,(x_n,u_n)\) to %
\((f(x_1,\dots,x_n),f'(x_1,\dots,x_n)\cdot(u_1,\dots,u_n))\).
For any \(i\in\Set{1,\dots,n}\) we also need to be able to define a
\(i\)th \emph{partial derivative} %
\(\Dfunpart if \in\Kl\cL(X_1\IWith\cdots\IWith\Sfun
X_i\IWith\cdots\IWith X_n,\Sfun Y)\) %
which intuitively maps %
\((x_1,\dots,(x_i,u),\dots,x_n)\) to %
\((f(x_1,\dots,x_n),f'_i(x_1,\dots,x_n)\cdot u)\).
We also expect these partial derivatives to satisfy
\begin{align*}
  f'(x_1,\dots,x_n)\cdot(u_1,\dots,u_n)
  =\sum_{i=1}^n f'_i(x_1,\dots,x_n)\cdot u_i
  =\sum_{i=1}^n\frac{\partial f(\List x1n)}{\partial x_i}\cdot u_i\,.
\end{align*}
The conditions introduced so far allow us to define such partial
derivatives and prove their expected properties without further
assumptions as we explain now.

We take \(n=2\) to simplify notations but the general case is not more
complicated conceptually.
Then it is possible to define %
\(\Wstren1_{X_1,X_2} %
\in\cL(\With{\Sfun X_1}{X_2},\Sfun{\Withp{X_1}{X_2}})\) as the
following composition of morphisms
\begin{equation*}
  \begin{tikzcd}
    \Sfun{X_1}\IWith{X_2}
    \ar[r,"\With{\Sfun{X_1}}{\Sin0}"]
    &[2em]
    \Sfun{X_1}\IWith\Sfun{X_2}
    \ar[r,"\Swithiso_{X_1,X_2}"]
    &[2em]
    \Sfun{\Withp{X_1}{X_2}}
  \end{tikzcd}
\end{equation*}
and we use \(\Dfunstr1_{X_1,X_2}\) for the associated morphism %
\(\A\Derfun{\Wstren1_{X_1,X_2}}\in\Kl\cL(\With{X_1}{\Dfun X_2},
\Dfun{\Withp{X_1}{X_2}})\) in the Kleisli category
\(\Kl\cL\).
We define similarly %
\(\Dfunstr2_{X_1,X_2}
\in\Kl\cL(\With{X_1}{\Dfun X_2},\Dfun{\Withp{X_1}{X_2}})\).
Intuitively \(\Dfunstr1((x_1,u_1),x_2)=((x_1,x_2),(u_1,0))\) and
\(\Dfunstr2(x_1,(x_2,u_2))=((x_1,x_2),(0,u_2))\).
It is also easily checked that \(\Dfunstr2\) can be obtained from %
\(\Dfunstr1\) using the symmetry isomorphism associated with
\(\IWith\): %
\(\Dfunstr2
=\Dfun{\Tuple{\Proj2,\Proj1}}
\Comp\Dfunstr1
\Comp\Tuple{\Proj2,\Proj1}\).

\begin{theorem} %
  \label{th:dfun-strength-with}
  The morphisms \(\Dfunstr1_{X_1,X_2}\) and \(\Dfunstr2_{X_1,X_2}\) of
  \(\Kl\cL\) are natural in \(X_1\) and \(X_2\) and define a
  commutative strength on the monad \(\Dfun\).
  More precisely, the following diagram commutes in \(\Kl\cL\)
  \begin{equation} %
    \label{eq:com-strength-diff-kleisli}
    \begin{tikzcd}
      &[-2em]
      \With{\Dfun X_1}{\Dfun X_2}
      \ar[ld,swap,"\Dfunstr1_{X_1,\Dfun X_2}"]
      \ar[rd,"\Dfunstr2_{\Dfun X_1,X_2}"]
      &[-2em]\\
      \Dfun\Withp{X_1}{\Dfun X_2}
      \ar[d,swap,"\Dfun\Dfunstr2_{X_1,X_2}"]
      &
      &
      \Dfun\Withp{\Dfun X_1}{X_2}
      \ar[d,"\Dfun\Dfunstr1_{X_1,X_2}"]\\
      \Dfun^2\Withp{X_1}{X_2}
      \ar[rr,"\Sflip_{\With{X_1}{X_2}}"]
      &&
      \Dfun^2\Withp{X_1}{X_2}
    \end{tikzcd}
  \end{equation}
  and the induced monoidality %
  \(\Dfunmon_{X_1,X_2}=\Dmonm_{\With{X_1}{X_2}}
  \Comp\Dfun\Dfunstr2_{X_1,X_2}
  \Comp\Dfunstr1_{\Dfun X_1,X_2}
  =\Dmonm_{\With{X_1}{X_2}}
  \Comp\Dfun\Dfunstr1_{X_1,X_2}
  \Comp\Dfunstr2_{X_1,\Dfun X_2}\)
  coincides with \(\A\Derfun{\Swithiso_{X_1,X_2}}\)
  which is an iso in \(\Kl\cL\).
\end{theorem}
\begin{proof}[Proof sketch]
  This is essentially trivial.
  For instance \Cref{eq:com-strength-diff-kleisli} is the image by
  \(\Derfun\) of the diagram
  \begin{equation*}
    \begin{tikzcd}
      &[-2em]
      \With{\Sfun X_1}{\Sfun X_2}
      \ar[ld,swap,"\Wstren1_{X_1,\Sfun X_2}"]
      \ar[rd,"\Wstren2_{\Sfun X_1,X_2}"]
      &[-2em]\\
      \Sfun\Withp{X_1}{\Sfun X_2}
      \ar[d,swap,"\Sfun\Wstren2_{X_1,X_2}"]
      &
      &
      \Sfun\Withp{\Sfun X_1}{X_2}
      \ar[d,"\Sfun\Wstren1_{X_1,X_2}"]\\
      \Sfun^2\Withp{X_1}{X_2}
      \ar[rr,"\Sflip_{\With{X_1}{X_2}}"]
      &&
      \Sfun^2\Withp{X_1}{X_2}
    \end{tikzcd}    
  \end{equation*}
  whose commutation is easily proven using the joint monicity of %
  \((\Sproj i\Compl\Sproj j)_{(i,j)\in\Set{0,1}}\).
  The naturality of \(\Dfunstr1\) and \(\Dfunstr2\) boils down to the
  commutativity in \(\cL\) of
  \begin{equation*}
    \begin{tikzcd}
      \Oc\Sfun{\Withp{X_1}{X_2}}
      \ar[r,"\Sdiff_{\With{X_1}{X_2}}"]
      \ar[d,swap,"\Tuple{\Oc\Sfun{\Proj1},\Oc\Sfun{\Proj2}}"]
      &[2em]
      \Sfun\Oc{\Withp{X_1}{X_2}}
      \ar[d,"\Tuple{\Sfun\Oc{\Proj1},\Sfun\Oc{\Proj2}}"]
      \\
      \With{\Oc\Sfun{X_1}}{\Oc\Sfun{X_2}}
      \ar[r,"\With{\Sdiff_{X_1}}{\Sdiff_{X_2}}"]
      &
      \With{\Sfun{\Oc X_1}}{\Sfun{\Oc X_2}}
    \end{tikzcd}
  \end{equation*}
  which results from the naturality of \(\Sdiff_X\) in \(X\) and the
  joint monicity of \(\Proj1\) and \(\Proj2\).
\end{proof}

\begin{definition}
  \label{def:kleisli-partial-derivatives}
  The two partial derivatives of \(f\in\Kl\cL(\With{X_1}{X_2},Y)\) %
  are %
  \begin{align*}
    \Dfun_1 f
    &=\Dfun f
      \Comp\Dfunstr1_{X_1,X_2}\in\Kl\cL(\Dfun{X_1}\IWith X_2,\Dfun Y)\\
    \Dfun_2 f
    &=\Dfun f
      \Comp\Dfunstr2_{X_1,X_2}\in\Kl\cL({X_1}\IWith{\Dfun X_2},\Dfun Y)\,.
  \end{align*}
\end{definition}

\begin{proposition}
  If \(f\in\Kl\cL(\With{X_1}{X_2},Y)\) %
  we have
  \begin{align*}
    \Dfun f=\Dmonm\Comp\Dfun_1{\Dfun_2f}=\Dmonm\Comp\Dfun_2{\Dfun_1f}\,.
  \end{align*}
\end{proposition}
\begin{proof}
  Apply \Cref{th:dfun-strength-with}.
\end{proof}
This means intuitively, as expected, that
\(f'(x_1,x_2)\cdot(u_1,u_2)
=f'_1(x_1,x_2)\cdot u_1+f'_2(x_1,x_2)\cdot u_2\).

Composing these morphisms \(\Dfunstr1\) and \(\Dfunstr2\), one can
define %
\(\Dfunstr i
\in\Kl\cL(X_1\IWith\cdots\IWith\Dfun{X_i}\IWith\cdots\IWith
X_n,\Ap\Dfun{X_1\IWith\cdots\IWith X_n})\) that one can also define
directly as \(\Dfunstr i=\A\Derfun{\Wstren i}\) where %
\(\Wstren i =\Swithiso_{\List X1n}
\Compl(\Sin0\IWith\cdots\IWith\Sfun{X_i}\IWith\cdots\IWith\Sin0)\),
that is, intuitively, %
\[
  \Dfunstr i(x_1,\dots,(x_i,u),\dots,x_n) =((\List
  x1n),(0,\dots,u,\dots,0))\,.
\]
Given \(f\in\Kl\cL(X_1\IWith\cdots\IWith X_n,Y)\), we can define the
\(i\)th partial derivative of \(f\) as
\begin{align*}
  \Dfun_i f
  =\Dfun f\Comp\Dfunstr i
  \in\Kl\cL(X_1\IWith\cdots\IWith\Dfun{X_i}\IWith\cdots\IWith X_n,\Dfun Y)\,.
\end{align*}
and given any repetition-free enumeration \((\List i1n)\) of
\(\Set{1,\dots,n}\) we have %
\begin{align*}
  \Dfun f=\Dmonm^n\Comp\Dfun_{i_1}\cdots\Dfun_{i_n}f\,.
\end{align*}

\subsection{The differential structure, in the elementary case}

Let \(\cL\) be an elementarily summable category (see
\Cref{def:elementarily-summable-cat}) and let
\((\Sfun,\Sproj0,\Sproj1,\Ssum)\) be the associated summability
structure.

Let \(\Ddiff\in\cL(\Dbimon,\Oc\Dbimon)\) be a morphism.

\begin{definition} %
  \labeltext{(\(\Ddiff\)-chain)}{ax:dchain}\ref{ax:dchain} %
  We say that \(\Ddiff\) satisfies \ref{ax:dchain} if it is a %
  \(\Oc\)-coalgebra structure on \(\Dbimon\), that is, the two
  following diagrams commute
  \begin{equation*}
    \begin{tikzcd}
      \Dbimon
      \ar[r,"\Ddiff"]
      \ar[dr,swap,"\Id_\Dbimon"]
      &
      \Oc\Dbimon
      \ar[d,"\Der\Dbimon"]
      \\
      &
      \Dbimon
    \end{tikzcd}
    \Treesep
    \begin{tikzcd}
      \Dbimon
      \ar[r,"\Ddiff"]
      \ar[d,swap,"\Ddiff"]
      &
      \Oc\Dbimon
      \ar[d,"\Digg\Dbimon"]
      \\
      \Oc\Dbimon
      \ar[r,"\Oc\Ddiff"]
      &
      \Occ\Dbimon
    \end{tikzcd}
  \end{equation*}
\end{definition}

Remember (from~\cite{Mellies09}, for instance) that the functor
\(\Oc\_\) inherits, from its strict symmetric monoidal structure
\((\Seelyz,\Seelyt)\) from the SMC \((\cL,\mathord{\IWith})\) to the
SMC \((\cL,\mathord{\ITens})\), a lax symmetric monoidal structure
\((\ExpMon0,\ExpMon2)\) from the SMC \((\cL,\mathord{\ITens})\) to
itself.
This means that \(\ExpMon0\in\cL(\Sone,\Oc\Sone)\) and %
\(\ExpMon2_{X,Y}\in\cL(\Tens{\Oc X}{\Oc Y},\Oc{\Tensp XY})\) satisfy
some coherence diagrams that we do not record here.
These morphisms are defined as follows:
\begin{equation*}
  \begin{tikzcd}
    \Sone
    \ar[r,"\Seelyz"]
    &
    \Oc\Stop
    \ar[r,"\Digg\Stop"]
    &[1.4em]
    \Occ\Stop
    \ar[r,"\Oc{\Invp{\Seelyz}}"]
    &[2em]
    \Oc\Sone
    &[2.4em]
    \\
    \Tens{\Oc X}{\Oc Y}
    \ar[r,"\Seelyt_{X,Y}"]
    &
    \Oc{\Withp XY}
    \ar[r,"\Digg{\With XY}"]
    &
    \Occ{\Withp XY}
    \ar[r,"\Oc{\Inv{(\Seelyt_{X,Y})}}"]
    &
    \Oc{\Tensp{\Oc X}{\Oc Y}}
    \ar[r,"\Oc{\Tensp{\Der X}{\Der Y}}"]
    &
    \Oc{\Tensp XY}
  \end{tikzcd}
\end{equation*}

This structure is quite important in particular when considering the
Eilenberg-Moore category \(\Em\cL\) of the comonad \(\Oc\_\). %
Remember that an object of that category is a pair %
\(P=(\Coalgca P,\Coalgst P)\) where \(\Coalgca P\) is an object of
\(\cL\) and \(\Coalgst P\in\cL(\Coalgca P,\Oc{\Coalgca P})\) satisfies
\begin{equation*}
  \begin{tikzcd}
    \Coalgca P
    \ar[r,"\Coalgst P"]
    \ar[dr,swap,"\Id_{\Coalgca P}"]
    &
    \Oc{\Coalgca P}
    \ar[d,"\Der{\Coalgca P}"]
    \\
    &
    \Coalgca P
  \end{tikzcd}
  \Treesep
  \begin{tikzcd}
    \Coalgca P
    \ar[r,"\Coalgst P"]
    \ar[d,swap,"\Coalgst P"]
    &
    \Oc{\Coalgca P}
    \ar[d,"\Digg{\Coalgca P}"]
    \\
    \Oc{\Coalgca P}
    \ar[r,"\Oc{\Coalgst P}"]
    &
    \Occ{\Coalgca P}
  \end{tikzcd}
\end{equation*}
and a morphism \(P\to Q\) in \(\Em\cL\) is an %
\(f\in\cL(\Coalgca P,\Coalgca Q)\) such that
\begin{equation*}
  \begin{tikzcd}
    \Coalgca P
    \ar[r,"f"]
    \ar[d,swap,"\Coalgst P"]
    &
    \Coalgca Q
    \ar[d,"\Coalgst Q"]
    \\
    \Oc{\Coalgca P}
    \ar[r,"\Oc f"]
    &
    \Oc{\Coalgca Q}
  \end{tikzcd}
\end{equation*}

Equipped with \(\ExpMon1\), the object \(\Sone\) is a
\(\Oc\)-coalgebra that we simply denote as \(\Sone\) and, given two
\(\Oc\)-coalgebras \(P\) and \(Q\), the pair %
\((\Tens{\Coalgca P}{\Coalgca Q},h)\) where \(h\) is defined as the
following composition of morphisms
\begin{equation*}
  \begin{tikzcd}
    \Tens{\Coalgca P}{\Coalgca Q}
    \ar[r,"\Tens{\Coalgst P}{\Coalgst Q}"]
    &[1.4em]
    \Tens{\Oc{\Coalgca P}}{\Oc{\Coalgca Q}}
    \ar[r,"\ExpMon2_{\Coalgca P,\Coalgca Q}"]
    &
    \Oc{\Tensp{\Coalgca P}{\Coalgca Q}}
  \end{tikzcd}
\end{equation*}
is a \(\Oc\)-coalgebra that we denote as \(\Tens PQ\).

\begin{definition}
  \labeltext{(\(\Ddiff\)-local)}{ax:dloc}\ref{ax:dloc} %
  We say that \(\Ddiff\) satisfies \ref{ax:dloc} if the following
  diagram commutes
  \begin{equation*}
    \begin{tikzcd}
      \Sone
      \ar[r,"\ExpMon0"]
      \ar[d,swap,"\Win0"]
      &
      \Oc\Sone
      \ar[d,"\Oc\Win0"]\\
      \Dbimon
      \ar[r,"\Ddiff"]
      &
      \Oc\Dbimon
    \end{tikzcd}
  \end{equation*}
  In other words \(\Win 0\) is a \(\Oc\)-coalgebra morphism from
  \(\Sone\) to \((\Dbimon,\Ddiff)\).
\end{definition}

\begin{definition}
  \labeltext{(\(\Ddiff\)-add)}{ax:dadd}\ref{ax:dadd} %
  We say that \(\Ddiff\) satisfies \ref{ax:dadd} if the following
  diagrams commute
  \begin{equation*}
    \begin{tikzcd}
      \Dbimon
      \ar[r,"\Ddiff"]
      \ar[d,swap,"\Proj0"]
      &
      \Oc\Dbimon
      \ar[d,"\Oc\Proj0"]
      \\
      \Sone
      \ar[r,"\ExpMon0"]
      &
      \Oc\Sone
    \end{tikzcd}
    \Treesep
    \begin{tikzcd}
      \Dbimon
      \ar[d,swap,"\Dbimonm"]
      \ar[rr,"\Ddiff"]
      &
      &
      \Oc\Dbimon
      \ar[d,"\Oc\Dbimonm"]
      \\
      \Tens\Dbimon\Dbimon
      \ar[r,"\Tens\Ddiff\Ddiff"]
      &
      \Tens{\Oc\Dbimon}{\Oc\Dbimon}
      \ar[r,"\ExpMon2_{\Dbimon,\Dbimon}"]
      &
      \Oc{\Tensp\Dbimon\Dbimon}
    \end{tikzcd}
  \end{equation*}
  In other words, the unit and the comultiplication of the bimonoid
  \(\Dbimon\) are \(\Oc\)-coalgebra morphisms from %
  \((\Dbimon,\Ddiff)\) to \(\Sone\) and to
  \(\Tens{(\Dbimon,\Ddiff)}{(\Dbimon,\Ddiff)}\) respectively.
\end{definition}

In~\cite{Ehrhard23a}, we proved the following result.

\begin{theorem}
  There is a bijective correspondence between the natural
  transformations \(\Sdiff_X\in\cL(\Oc\Sfun X,\Sfun\Oc X)\) which are
  differential structures on \(\cL\) and the morphisms
  \(\Ddiff\in\cL(\Dbimon,\Oc\Dbimon)\) which satisfy %
  \ref{ax:dchain}, \ref{ax:dloc} and \ref{ax:dadd}.
\end{theorem}

Let us just explain how the distributive law \(\Sdiff\) is defined when
\(\Ddiff\) is given:
first we define a morphism \(\Sdiffuc_X\) as the following composition
of morphisms
\begin{equation} %
  \label{eq:Sdiff-elem-def-uncurried}
  \begin{tikzcd}
    \Tens{\Oc{\Limplp{\Dbimon}{X}}}{\Dbimon}
    \ar[r,"\Tens{\Oc{\Limplp{\Dbimon}{X}}}{\Ddiff}"]
    &[2em]
    \Tens{\Oc{\Limplp{\Dbimon}{X}}}{\Oc\Dbimon}
    \ar[r,"\ExpMon2"]
    &
    \Oc{\Tensp{\Limplp{\Dbimon}{X}}{\Dbimon}}
    \ar[r,"\Oc{\Evlin}"]
    &
    \Oc X
  \end{tikzcd}
\end{equation}
and then we set %
\(\Sdiff_X=\Curlinp{\Sdiffuc_X}
\in\cL(\Oc{\Sfun X},\Sfun{\Oc X})\).

\begin{definition}
  A \emph{differential elementarily summable resource category} is an
  elementarily summable category resource \(\cL\) equipped with a
  morphism \(\Ddiff\in\cL(\Dbimon,\Oc\Dbimon)\) which satisfies
  \ref{ax:dchain}, \ref{ax:dloc} and \ref{ax:dadd}.
\end{definition}

The whole point of these definitions is the observation
in~\cite{Ehrhard23a} that such a differential structure on \(\Dbimon\) is
quite easy to obtain.
Remember in particular that a SMC \(\cL\) is \emph{Lafont} if \(\Oc\_\)
is the comonad associated with an adjunction between \(\cL\) and the
category \(\cL^{\mathord{\ITens}}\) of commutative comonoids on
\(\cL\).
More precisely, this means that the obvious forgetful functor
\(\cL^{\mathord{\ITens}}\to\cL\) has a right adjoint, and \(\Oc\_\) is
the comonad on \(\cL\) induced by this adjunction%
\footnote{More concretely, but more fuzzily also: in a Lafont resource
  category, any commutative comonoid is a \(\Oc\)-coalgebra.}.
It turns out that many interesting and non additive models of \LL{}
are Lafont resource categories, here are a few examples but there are
many others:
\begin{itemize}
\item the category of coherence spaces with the multiset based
  exponential;
\item the category of hypercoherence spaces with the multiset based
  exponential\footnote{This exponential has not been formally
    introduced as far as we know but is easy to describe.};
\item the category of nonuniform coherence spaces equipped with the
  Boudes exponential~\cite{BucciarelliEhrhard99,Boudes11};
\item the category of probabilistic coherence spaces with its unique
  known exponential~\cite{DanosEhrhard08,CrubilleEhrhardPaganiTasson16}.
\end{itemize}
\begin{theorem} %
  \label{th:lafont-elementary-differential-structure}
  If \(\cL\) is a Lafont resource category which is elementarily
  summable, then \(\Dbimon\) has exactly one differential coalgebra
  structure.
\end{theorem}
\begin{proof}[Sketch of the proof]
  We know that \((\Dbimon,\Proj0,\Dbimonm)\) is a commutative
  comonoid, see \Cref{prop:dbimon-comonoid}.
  This structure induces the announced \(\oc\)-coalgebra structure on
  \(\Dbimon\) through the Lafont property of \(\cL\).
\end{proof}

\begin{Example}
  The category \(\COH\) is a resource category.
  Its tensor product is defined by
  \(\Web{\Tens{E_1}{E_2}}=\Web{E_1}\times\Web{E_2}\) and coherence
  given by %
  \((a_1,a_2)\Coh{\Tens EF}(a'_1,a'_2)\) if
  \((a_i\Coh{E_i}a'_i)_{i=1,2}\) which is easily seen to be a functor:
  given %
  \((s_i\in\COH(E_i,F_i))_{i=1,2}\), the set %
  \(\Tens{s_1}{s_2}=\Set{((a_1,a_2),(b_1,b_2))
    \St ((a_i,b_i)\in s_i)_{i=1,2}}\) is an element of %
  \(\COH(\Tens{E_1}{E_2},\Tens{F_1}{F_2})\).
  The associated unit is \(\Sone=(\Set\Sonelem,\mathord=)\).

  The resource modality originally introduced by Girard, that we
  considered in \Cref{sec:coh-space-stable-diff} and that fails to
  provide an adequate differential setting, is defined by taking for %
  \(\Oc E\) the set of finite cliques of \(E\), with %
  \(x_1\Coh{\Oc E}x_2\) if \(x_1\cup x_2\in\Cl E\).
  As is quite well known \(\COH\) has also a free exponential whose
  definition is quite similar to the original one: %
  take for \(\Web{\Oc E}\) the set of all \(m\in\Mfin{\Web E}\) such
  that \(\Supp m\in\Cl E\), and coherence given by %
  \(m_1\Coh{\Oc E}m_2\) if \(m_1+m_2\in\Web{\Oc E}\).

  Notice that \(\Web{\Oc\Dbimon}\) is the set of all finite multisets
  of elements of \(\Set{0,1}\) since \(0\Scoh\Dbimon 1\), and
  \(m\Coh{\Oc\Dbimon}m'\) holds for all \(m,m'\in\Web{\Oc\Dbimon}\).
  The differential structure induced by the fact that this exponential
  is free is %
  \(\Ddiff\in\COH(\Dbimon,\Oc\Dbimon)\) is given by
  \begin{equation*}
    \Ddiff=\{(i,\Mset{\List i1k})\St k\in\Nat,\,i,\List i1k\in\Set{0,1}
    \text{ and }i=i_1+\cdots+i_k\}\,.
  \end{equation*}
  In other words \((i,m)\in\Ddiff\) if \(i=0\) and all the elements of
  \(m\) are \(0\), or \(i=1\) and all the elements of \(m\) are \(0\)
  but exactly one, which is \(1\), that is
  \begin{equation*}
    \Ddiff=\Set{(0,k\Mset 0)\St k\in\Nat}
    \cup\Set{(1,\Mset 1+k\Mset 0\St k\in\Nat)}\,.
  \end{equation*}
  Since \(0\Scoh\Dbimon 1\) we have \(m\Coh{\Oc\Dbimon}m'\) for all %
  \(m,m'\in\Web{\Oc\Dbimon}\).
  Moreover if \((0,m),(1,m')\in\Ddiff\) we have \(m\not=m'\) as
  required, since \(1\in\Supp{m'}\setminus\Supp m\).
  Let us check for instance that the second diagram of \ref{ax:dadd}
  commutes%
  \footnote{Remember that this verification is not really needed since
    we know that \ref{ax:dchain}, \ref{ax:dloc} and \ref{ax:dadd} hold
    by the simple fact that the exponential is free.
    We think it is nevertheless useful to have a better intuition on
    the morphism \(\Ddiff\).}.
  So let
  \((i,p) \in\Web\Dbimon\times\Web{\Oc{\Tensp\Dbimon\Dbimon}}\) %
  and let us write \(p=\Mset{(l_1,r_1),\dots,(l_k,r_k))}\).
  Saying that \((i,p)\in\Oc\Dbimonm\Compl\Ddiff\) means that there is %
  \(m\in\Web{\Oc\Dbimon}\) such that %
  \((i,m)\in\Ddiff\) and \((m,p)\in\Oc\Dbimonm\).
  That is,
  \(m=\Mset{\List i1k}\) with \(i=i_1+\cdots+i_k\), %
  \((i_j=l_j+l'_j)_{j=1}^k\)
  and \(p=\Mset{(l_1,l'_1),\dots,(l_k,l'_k)}\).
  To summarize, \((i,p)\in\Oc\Dbimonm\Compl\Ddiff\) holds iff %
  \(p=\Mset{(l_1,l'_1),\dots,(l_k,l'_k)}\) with
  \(i=l_1+\cdots+l_k+l'_1+\cdots+l'_k\) for some \(k\in\Nat\) and %
  \(\List l1k,\List{l'}1k\in\Set{0,1}\).
  Saying that
  \((i,p)\in\ExpMon2\Compl\Tensp\Ddiff\Ddiff\Compl\Dbimonm\) means
  that there are \(l,l'\in\Set{0,1}\) such that \(l+l'=i\), and %
  \(m,m'\in\Web{\Oc\Dbimon}\) such that %
  \((l,m),(l',m')\in\Ddiff\) and \(((m,m'),p)\in\ExpMon2\).
  Up to reindexing, this latter condition means that %
  \(m=\Mset{\List l1k}\), \(m'=\Mset{\List{l'}1k}\) and %
  \(p=\Mset{(l_1,l'_1),\dots,(l_k,l'_k)}\).
  This shows that \((i,p)\in\Oc\Dbimonm\Compl\Ddiff\) holds iff
  \((i,p)\in\ExpMon2\Compl\Tensp\Ddiff\Ddiff\Compl\Dbimonm\), that is,
  the second diagram of \ref{ax:dadd} commutes.

  It is also interesting to describe the associated distributive law
  \(\Sdiff_E\in\COH(\Oc{\Sfun E},\Sfun{\Oc E})\).
  The composition of morphisms described in
  \Cref{eq:Sdiff-elem-def-uncurried} gives us %
  \(\Sdiff'_E\in\COH(\Tens{\Oc{\Limplp\Dbimon E}}{\Dbimon},\Oc E)\):
  \begin{multline*}
    \Sdiff'_E
    =\{((\Mset{(i_1,a_1),\dots,(i_k,a_k)},i),\Mset{\List a1k})\\
    \St k\in\Nat,\,i,\List i1k\in\Set{0,1},\,i=i_1+\cdots+i_k\text{
      and } \Set{\List a1k}\in\Cl E\}
  \end{multline*}
  and hence
  \begin{multline*}
    \Sdiff_E
    =\{(\Mset{(i_1,a_1),\dots,(i_k,a_k)},(i,\Mset{\List a1k}))\\
    \St k\in\Nat,\,i,\List i1k\in\Set{0,1},\,i=i_1+\cdots+i_k\text{
      and } \Set{\List a1k}\in\Cl E\}\,.
  \end{multline*}
  This means that the differential
  \(\Dfun t=(\Sfun t)\Compl\Ddiff_E\in\Kl\COH(\Sfun E,\Sfun F)\) %
  of \(t\in\Kl\COH(E,F)\) is given by
  \begin{align*}
    \Dfun t
    &=\{(\Mset{(i_1,a_1),\dots,(i_k,a_k)},(i,b))
    \in\Web{\Oc{\Sfun E}}\times\Web{\Sfun F}\\
    &\Textsep\St (\Mset{\List a1k},b)\in t,\, i,\List i1k\in\Set{0,1}\text{ and
      } i=i_1+\cdots+i_k\}\\
    &=\{(\Mset{(0,a_1),\dots,(0,a_k)},(0,b))\St(\Mset{\List a1k},b)\in t\}\\
    &\Textsep\cup
      \{(\Mset{(0,a_1),\dots,(0,a_k),(1,a)},(1,b))\St(\Mset{\List a1k,a},b)
      \in t\text{ and }(a\Scoh Ea_j)_{j=1}^k\}
  \end{align*}
  where the condition \((a\Scoh Ea_j)_{j=1}^k\) comes from the fact
  that we must have %
  \(\Mset{(0,a_1),\dots,(0,a_k),(1,a)}\in\Web{\Oc{\Sfun E}}\).

  We can define a stable%
  \footnote{This means that \(\Fun t\) commutes with unions of
    directed families of cliques and with intersections of bounded
    non-empty finite families of cliques.} %
  function \(\Fun t:\Cl E\to\Cl F\) by %
  \(\Fun t(x)=\Set{b\in\Web Y\St\exists m\in\Web{\Oc E} \ \Supp
    m\subseteq x\text{ and }(m,b)\in t}\).
  Then, under the identification
  \begin{align*}
    \Cl{\Sfun E}=\Set{(x,u)\in\Cl E^2\St x\cup u\in\Cl E
    \text{ and }x\cap u=\emptyset}
  \end{align*}
  we have
  \begin{align*}
    \Fun{\Dfun t}(x,u)=
    (\Fun t(x),\Union_{a\in u}\Fun t(x\cup\Set a)\setminus\Fun t(x))
  \end{align*}
  and we recover the initial intuition of the derivative of a stable
  function.
  What makes this differentiation \(t\mapsto\Dfun t\) functorial (that
  is, the chain rule to hold) is the fact that we have moved to the
  free exponential, whose web uses finite multicliques instead of
  finite cliques.

  The great benefit of this systematic approach based on the
  elementary differential structure of \(\COH\) is that now
  \(\Dfun t\) is a morphism \(\Sfun E\to\Sfun F\) in \(\Kl\COH\) and
  hence induces a stable function \(\Cl{\Sfun E}\to\Cl{\Sfun F}\):
  this is a way of saying that the differential depends stably from
  the point where it is computed.
\end{Example}

\begin{Example}
  The category \(\PCOH\) is also a resource category.
  We have seen in Example~\ref{ex:pcoh-cart-smc-elem-sum} that this category
  is a cartesian SMCC which is actually \Staraut{} for the dualizing
  object \(\Sbot=\Sone\).
  Then we define \(\Oc X\) by \(\Web{\Oc X}=\Mfin{\Web X}\) and %
  \(\Pcohp{\Oc X}=\Biorth{\Set{\Prom x\St x\in\Pcoh X}}\) where %
  \(\Prom x\in\Realpto{\Web{\Oc X}}\) is defined by
  \(\Prom x_m=x^m=\prod_{a\in\Web X}x_a^{m(a)}\).
  Given \(t\in\PCOH(X,Y)\), it is easy to check that there is exactly one %
  \(\Oc t\in\Realpto{\Web{\Oc X}\times\Web{\Oc Y}}\) such that %
  \begin{align} %
    \label{eq:pcs-excl-prom-basic-req}
    \forall x\in\Pcoh X\quad\Matappa{\Oc t}{\Prom x}=\Promp{\Matappa tx}\,.
  \end{align}
  Explicitly, a simple computation using
  \Cref{eq:pcs-excl-prom-basic-req} shows that
  \begin{align*}
    \forall (m,p)\in\Web{\Oc X}\times\Web{\Oc Y}\quad 
    (\Oc t)_{m,p}=\sum_{r\in\Mexpset mp}\Multinomb pr t^r
  \end{align*}
  where \(\Mexpset mp\) is the set of all
  \(r\in\Mfin{\Web X\times\Web Y}\) such that %
  \(\forall a\in\Web X\ \sum_{b\in\Web Y}r(a,b)=m(a)\) and %
  \(\forall b\in\Web Y\ \sum_{a\in\Web X}r(a,b)=p(b)\), and %
  \(\Multinomb pr
  =\prod_{b\in\Web Y}\frac{\Factor{p(b)}}
  {\prod_{a\in\Web X}\Factor{r(a,b)}}\in\Nat\)
  is a multinomial coefficient.
  It can be proven that if \(s,t\in\Kl\PCOH(\Oc X,Y)\) satisfy %
  \(\forall x\in\Pcoh X\ \Matappa s{\Prom x}=\Matappa t{\Prom x}\), %
  then \(s=t\) (as matrices).
  So the function \(\Fun t:\Pcoh X\to\Pcoh Y\) defined by %
  \(\Fun t(x)=\Matappa t{\Prom x}\) fully determines \(t\); %
  such a function \(\Pcoh X\to\Pcoh Y\) will be called an
  \emph{analytic} function since indeed we have
  \begin{align*}
    \Fun t(x)=\Big(\sum_{m\in\Web{\Oc X}}t_{m,b}x^m\Big)_{b\in\Web Y}
  \end{align*}
  meaning that \(\Fun t\) is defined as a (generalized) power series
  with nonnegative coefficients.

  This functor is a comonad with counit %
  \(\Der X\in\PCOH(\Oc X,X)\) characterized by \(\Fun{\Der X}(x)=x\) %
  and \(\Digg X\in\PCOH(\Oc X,\Occ X)\) by %
  \(\Fun{\Digg X}(x)=\Promm x\), that is, as matrices, %
  \((\Der X)_{m,a}=\Kronecker m{\Mset a}\) %
  (for \((m,a)\in\Web{\Oc X}\times\Web X\)) and %
  \((\Digg X)_{m,M}=\Kronecker m{\sum M}\).
  It is easily checked to be strong monoidal from %
  \((\PCOH,\mathord{\IWith})\) to \((\PCOH,\mathord{\ITens})\).
  For instance the Seely isomorphism %
  \(\Seelyt_{X,Y}\in\PCOH(\Tens{\Oc X}{\Oc Y},\Oc{\Withp XY}\) is
  fully characterized by %
  \(\Matappa{\Seelyt_{X,Y}}{\Tensp{\Prom x}{\Prom y}}=\Prom{(x,y)}\) %
  (identifying \(\Pcohp{\With XY}\) with \({\Pcoh X}\times{\Pcoh Y}\)).

  It was proved in~\cite{CrubilleEhrhardPaganiTasson16} that this
  exponential is the free one, that is, the SMC \(\PCOH\) is a Lafont
  category.
  So by \Cref{th:lafont-elementary-differential-structure} the object
  \(\Dbimon=\With\Sone\Sone\) has a structure of \(\oc\)-coalgebra
  \(\Ddiff\in\PCOH(\Dbimon,\Oc\Dbimon)\), which is
  characterized by
  \begin{align*}
    \Ddiff_{i,\Mset{\List i1k}}=\Kronecker i{i_1+\cdots+i_k}
  \end{align*}
  for \(i,\List i1k\in\Set{0,1}\).
  This structure turns \(\PCOH\) into an elementary coherent
  differential category.
  So we have an induced distributive law %
  \(\Sdiff_X\in\PCOH(\Oc{\Sfun X},\Sfun{\Oc X})\) and an easy
  computation shows that
  \begin{align*}
    (\Sdiff_X)_{p,(i,m)}
    =
    \begin{cases}
      1 & \text{if }p=\Indactms 0m\\
      m(a) & \text{if }m=m_0+\Mset a
      \text{ and }p=\Indactms 0{m_0}+\Mset{(1,a)}\\
      0 & \text{otherwise.}
    \end{cases}
  \end{align*}
  so that the extension \(\Dfun\) of \(\Sfun\) to \(\Kl\PCOH\)
  acts as follows on morphisms.
  Let \(t\in\PCOH(\Oc X,Y)\) then
  \(\A\Dfun t\in\PCOH(\Oc{\Sfun X},\Sfun Y)\) is given by %
  \begin{align*}
    (\Dfun t)_{p,(i,b)}
    =
    \begin{cases}
      t_{m,b} & \text{if }i=0\text{ and }p=\Indactms 0m\\
      m(a)t_{m,b} & \text{if }i=1\text{, }p=\Indactms 0{m_0}+\Mset{(1,a)}
      \text{ and }m=m_0+\Mset a\\
      0 & \text{otherwise.}
    \end{cases}
  \end{align*}
  This means that, given an element of \(\Pcohp{\Sfun X}\) that we
  consider as a pair %
  \((x,u)\in\Pcoh X^2\) such that \(x+u\in\Pcoh X\), we have
  \begin{align*}
    \Fun{\Dfun t}(x,u)
    =(\Fun t(x),\Fun t'(x)\cdot u)
  \end{align*}
  where
  \[
    \Fun t'(x)\cdot u
    =\left(\sum_{\Biind{m\in\Mfin{\Web X}}{a\in\Web X}}
      (m(a)+1)t_{m+\Mset a,b}x^mu_a\right)_{b\in\Web Y}
    =\left(
      \lim_{\epsilon\to 0^+}\frac{\Fun{t}(x+\epsilon u)_b
        -\Fun t(x)_b}{\epsilon}
      \right)_{b\in\Web Y}
  \]
  is the differential of \(\Fun t\) computed at \(x\) in the direction
  \(u\).
\end{Example}

\section{The closed case}
\label{sec:closed-differential}
So far we have not considered function space constructions, but the
reader acquainted with the denotational semantics of \LL{} probably
knows that both SMC's \(\COH\) and \(\PCOH\) are symmetric monoidal
closed categories, and that the associated Kleisli categories
\(\Kl\COH\) and \(\Kl\PCOH\) are cartesian closed.

More generally, when \(\cL\) is a resource category which is closed
(as an SMC), we know that \(\Kl\cL\) is a CCC%
\footnote{This is the categorical counterpart of the Girard's
  translation of intuitionistic logic into linear logic.}.
We use the notation \(\Simpl XY\) for the object of morphisms from
\(X\) to \(Y\) in \(\Kl\cL\), which is \(\Limpl{\Oc X}{Y}\) and whose
associated evaluation morphism %
\(\Ev\in\Kl\cL(\With{(\Simpl XY)}{X},Y)\) is defined as the following
composition of morphisms in \(\cL\).
\begin{equation*}
  \begin{tikzcd}
    \Oc{\Withp{(\Simpl XY)}{X}}
    \ar[r,"\Invp{\Seelyt}"]
    &[1em]
    \Tens{\Oc\Limplp{\Oc X}{Y}}{\Oc X}
    \ar[r,"\Der{\Limpl{\Oc X}Y}"]
    &[1.8em]
    \Tens{\Limplp{\Oc X}{Y}}{\Oc X}
    \ar[r,"\Evlin"]
    &[-1em]
    Y
  \end{tikzcd}
\end{equation*}
Given \(f\in\Kl\cL(\With ZX,Y)=\cL(\Oc{\Withp ZX},Y)\), we have %
\(f\Compl\Seelyt\in\cL(\Tens{\Oc Z}{\Oc X},Y)\) and hence %
\(\Ap\Curlin{f\Compl\Seelyt}\in\Kl\cL(Z,\Simpl XY)\) and this
morphism \(\A\Cur f=\Ap\Curlin{f\Compl\Seelyt}\) is uniquely
characterized by the equation
\begin{align*}
  \Ev\Comp\Withp{\A\Cur f}{X}=f\,.
\end{align*}

The main ingredient in the interpretation of a coherent differential
\(\lambda\)-calculus in the CCC associated with a coherent
differential resource category which is closed (as an SMC) will be a
morphism allowing to internalize the action of the \(\Dfun\) functor
as a morphism %
\(\Dfunint\in\Kl\cL(\Simpl XY,\Simpl{\Dfun X}{\Dfun Y})\).
This is not a surprise since one the main features of a strong monad is
precisely to allow such internalizations.
This morphism is defined by %
\(\Dfunint=\Ap\Cur{\A\Dfun\Ev\Comp\Dfunstr2_{\Simpl XY,X}}\) %
where \(\A\Dfun\Ev\Comp\Dfunstr1_{\Simpl XY,X}\) %
is typed as follows in \(\Kl\cL\):
\begin{equation*}
  \begin{tikzcd}
    \With{(\Timpl XY)}{\Dfun X}
    \ar[r,"\Dfunstr2_{\Simpl XY,X}"]
    &[2em]
    \Ap\Dfun{\With{(\Timpl XY)}{X}}
    \ar[r,"\A\Dfun\Ev"]
    &
    \Dfun Y\,.
  \end{tikzcd}
\end{equation*}

Let \(f\in\Kl\cL(\With ZX,Y)\).
There are two morphisms in %
\(\Kl\cL(Z,\Timpl{\Dfun X}{\Dfun Y})\) that we can naturally define
using \(f\), namely
\begin{align*}
  \Dfunint\Comp\A\Cur f
  \text{\quad and\quad}
  \Ap\Cur{\Dfun_2 f}=\Ap\Cur{\Dfun f\Comp\Dfunstr1_{Z,X}}
\end{align*}

\begin{proposition}
  For any \(f\in\Kl\cL(\With ZX,Y)\) we have %
  \(\Dfunint\Comp\A\Cur f=\Ap\Cur{\Dfun f\Comp\Dfunstr1_{Z,X}}\).
\end{proposition}
The proof is easy, and the meaning of this statement is that \(\cL\)
validates a form of ``differential \(\beta\)-reduction''.
This proposition shows how to compute the derivative of an abstraction
wrt.~to one of its free parameters.

Consider now \(f\in\Kl\cL(\With ZX,\Simpl YU)\) and %
\(g\in\Kl\cL(\With ZX,Y)\); \(f\) should be seen as a function %
\(Y\to U\) depending on two parameters in \(Z\) and \(X\) and \(g\)
as an argument for that function, similarly parameterized.
We can apply \(f\) to \(g\), defining %
\(\App fg=\Ev\Comp\Tuple{f,g}\in\Kl\cL(\With ZX,U)\) %
and then we can take the derivative of \(\App fg\) wrt.~the second
parameter which is
\begin{align*}
  \A{\Dfun_2}{\App fg}=\A\Dfun{\App fg}\Comp\Dfunstr2
  \in\Kl\cL(\With Z{\A\Dfun X},\A\Dfun{U})\,.
\end{align*}

On the other hand we have %
\(\A{\Dfun_2}g\in\Kl\cL(\With Z{\A\Dfun X},\A\Dfun Y)\) and %
\(\A{\Dfun_2}f\in\Kl\cL(\With Z{\A\Dfun X}),\Ap\Dfun{\Simpl YU})\) %
so that %
\(\Dfunstrf_{Y,U}\Comp\A{\Dfun_2}f
\in\Kl\cL(\With{Z}{\A\Dfun X},\Simpl Y{\A\Dfun U})\) %
where %
\(\Dfunstrf_{U,V}=\A\Derfun{\Sfunstri_{\Oc Y,U}}
\in\Kl\cL(\Ap\Dfun{\Simpl YU},\Simpl Y{\A\Dfun U})\) is an iso by %
\ref{ax:slimpl}.
Therefore %
\(\Dfunint\Comp\Dfunstrf_{Y,U}\Comp\A{\Dfun_2}f
\in\Kl\cL(\With{Z}{\A\Dfun X},\Simpl{\A\Dfun Y}{\A{\Dfun^2}U})\) %
and hence %
\(\App{\Dfunint\Comp\Dfunstrf_{Y,U}\Comp\A{\Dfun_2}f}{\A{\Dfun_2}g}
\in\Kl\cL(\With Z{\A\Dfun X},\A{\Dfun^2}U)\).
Remember that \(\Dmonm_U\in\Kl\cL(\A{\Dfun^2}{U},\A\Dfun U)\) is the
multiplication of the monad \(\Dfun\) on \(\Kl\cL\).

\begin{theorem} %
  \label{th:dfun-part-app}
  If \(f\in\Kl\cL(\With ZX,\Simpl YU)\) and %
  \(g\in\Kl\cL(\With ZX,Y)\) then we have
  \begin{equation*}
    \A{\Dfun_2}{\App fg}
    =\Dmonm_U
    \Comp\App{\Dfunint\Comp\Dfunstrf_{Y,U}\Comp\A{\Dfun_2}f}{\A{\Dfun_2}g}
    =\App{(\Simpl{\Dfun Y}{\Dmonm_U})
      \Comp\Dfunint
      \Comp\Dfunstrf_{Y,U}
      \Comp\A{\Dfun_2}f}
    {\A{\Dfun_2}g}\,.
  \end{equation*}
\end{theorem}
\begin{proof}[Proof sketch]
  It suffices to establish the following commutation
  \begin{equation*}
    \begin{tikzcd}
      \With{\Ap\Dfun{\Simpl YU}}{\A\Dfun Y}
      \ar[r,"\With{\Dfunstrf_{Y,U}}{\A\Dfun Y}"]
      \ar[dd,swap,"\Dfunmon_{\Simpl YU,Y}"]
      &[2em]
      \With{\Simplp Y{\A\Dfun U}}{\A\Dfun Y}
      \ar[r,"\Dfunstr1_{\Simpl Y{\A\Dfun U},Y}"]
      &[2em]
      \Ap\Dfun{\With{\Simplp Y{\A\Dfun U}}{Y}}
      \ar[d,"\A\Dfun\Ev"]
      \\
      &
      &
      \A{\Dfun^2}U
      \ar[d,"\Dmonm_U"]
      \\
      \Ap\Dfun{\With{\Simplp YU}{Y}}
      \ar[rr,"\A\Dfun\Ev"]
      &&
      \A\Dfun{U}
    \end{tikzcd}
  \end{equation*}
  The proof is not straightforward and uses crucially \ref{ax:sdwith}.
  It can be found in~\cite{Ehrhard23b}.
\end{proof}
This theorem shows how to compute the derivative of an application
wrt.~one of its parameters.
`
\section{Fixpoints}
\label{sec:fixpoints}

We start with stating a few standard results about fixpoints in a CCC
which is enriched in \(\omega\)-cpos.

\subsection{Reminder about fixpoint operators in a CCC}
\begin{definition}
  An \(\omega\)-cpo is a poset which has a least element \(0\) and
  where any monotone sequence has a lub.
  If \(D\) and \(D'\) are \(\omega\)-cpos, a function \(f:D\to D'\) is
  Scott-continuous
  (or simply continuous) if \(f\) is monotone and commutes with the
  lubs of monotone sequences.
\end{definition}

\begin{remark}
  In the literature, Scott continuity is usually defined as
  preservation of the lubs of arbitrary directed sets.
  However such lubs do not always exist in the situations we are
  interested in. This is specifically the case in categories arising
  in continuous probabilistic settings, but in such situations lubs of
  monotone countable sequences can be assumed to exist thanks to the
  monotone convergence theorem.
  We nevertheless use the term ``Scott continuity'' since the
  fundamental ideas of Dana~Scott are also central in such models.
\end{remark}

\begin{proposition}
  If \(D\) is an \(\omega\)-cpo and \(f:D\to D\) is continuous then
  \(f\) has a least fixpoint, which is \(\sup_{n\in\Nat}f^n(0)\).
\end{proposition}

\begin{definition} %
  \label{def:lambda-category}
  A \emph{\(\lambda\)-category} is a CCC which is enriched in
  \(\omega\)-cpos and in which the pairing operation %
  \(\cC(X,Y_1)\times\cC(X,Y_2)\to\cC(X,\With{Y_1}{Y_2})\) and the
  currying operation %
  \(\cC(\With ZX,Y)\to\cC(Z,\Impl XY)\) %
  are Scott continuous.
\end{definition}

\begin{proposition} %
  \label{prop:lambda-cat-basic-fixpoint}
  Let \(\cC\) be a \(\lambda\)-category and \(Y\) be an object of
  \(\cC\).
  Then any morphism \(h\in\cC(Y,Y)\) has a least fixpoint, that is,
  there is a morphism \(y\in\cC(\Stop,Y)\) such that \(h\Comp y=y\)
  and \(y\) is minimal with this property in the poset
  \(\cC(\Stop,X)\).
\end{proposition}
\begin{proof}
  One defines a sequence \((y_n\in\cC(\Stop,Y))_{n\in\Nat}\) by
  setting %
  \(y_0=0\) and \(y_{n+1}=f\Comp y_n\) and using the fact that
  composition is monotone, one checks that this sequence is monotone.
  Its lub \(y\) satisfies the required property by Scott continuity of
  composition.
\end{proof}

\begin{theorem} %
  \label{th:lambda-cat-fixpoint-operator}
  Let \(\cC\) be a \(\lambda\)-category.
  For any objects \(X\) of $\cC$ there is a morphism %
  \(\cY\in\cC(\Impl XX,X)\) such that, for any morphism %
  \(f\in\cC(Z,\Simpl XX)\) %
  the morphism \(\cY\Comp f\in\cC(Z,X)\) %
  is the least morphism \(g\in\cC(Z,X)\) such that %
  \(\Ev\Comp\Tuple{f,g}=g\).
\end{theorem}
This is a standard result in semantics, we give the proof because we
think it it helps understanding \Cref{sec:diff-fix}.
\begin{proof}
  Apply \Cref{prop:lambda-cat-basic-fixpoint} with
  \(Y=\Implp{\Implp XX}X\) and \(h=\Cur(H)\in\cC(Y,Y)\) where \(H\) is
  the following composition of morphisms in \(\cC\):
  \begin{equation*}
    \begin{tikzcd}
      \With Y{\Implp XX}
      \ar[d,"\With Y{\Tuple{\Id,\Id}}"]
      \\
      \With Y{\With{\Implp XX}{\Implp XX}}
      \ar[d,"\Tuple{\Proj3,\Proj1,\Proj2}"]
      \\
      \With{\With{\Implp XX}{Y}}{\Implp XX}
      \ar[d,"\With{\Implp XX}\Ev"]
      \\
      \With{\Implp XX}{X}
      \ar[d,"\Ev"]
      \\
      X
    \end{tikzcd}
  \end{equation*}
  which gives us \(\cY'\in\cC(\Stop,\Simpl{(\Simpl XX)}{X})\) which is
  the least morphism such that \(H\Comp\cY'=\cY'\).
  Then \(\cY\) is the following composition of morphisms in \(\cC\):
  \begin{equation*}
    \begin{tikzcd}
      \Simpl XX
      \ar[r,"\Tuple{0,\Simpl XX}"]
      &[2em]
      \With\Stop{(\Simpl XX)}
      \ar[r,"\With{\cY'}{(\Simpl XX)}"]
      &[3em]
      \With{(\Simpl{(\Simpl XX)}{X})}{(\Simpl XX)}
      \ar[r,"\Ev"]
      &
      X
    \end{tikzcd}
  \end{equation*}
  By monotonicity and Scott continuity of all the CCC operations in
  \(\cC\), it follows that \(\cY=\sup_{n\in\Nat}\cY_n\) where
  \((\cY_n\in\cC(\Simpl XX,X))_{n\in\Nat}\) is the (obviously
  monotone) sequence of morphisms inductively defined by \(\cY_0=0\)
  and \(\cY_{n+1}=\Ev\Comp\Tuple{\Simpl XX,\cY_n}\).

  Therefore
  \begin{align*}
    \cY\Comp f
    &=\sup_{n\in\Nat}(\cY_n\Comp f)\\
    &=\sup_{n\in\Nat}(\cY_{n+1}\Comp f)
      \text{\quad since }\cY_0=0\\
    &=\sup_{n\in\Nat}(\Ev\Comp\Tuple{f,\cY_n\Comp f})\\
    &=\sup_{n\in\Nat}g_n
  \end{align*}
  where \((g_n\in\cC(Z,X))_{n\in\Nat}\) is the (obviously monotone)
  sequence of morphisms inductively defined by \(g_0=0\) and %
  \(g_{n+1}=\Ev\Comp\Tuple{f,g_n}\).
\end{proof}

\subsection{The differential of fixpoints}
\label{sec:diff-fix}
Let \(\cL\) be a summable category.
\begin{definition}
  Let \(f,g\in\cL(X,Y)\), we write \(f\leq g\) if there is
  \(h\in\cL(X,Y)\) such that \(f\) and \(h\) are summable and
  \(g=f+h\).
\end{definition}

\begin{lemma}
  The relation \(\leq\) is a preorder relation on \(\cL(X,Y)\).
  The composition operation %
  \(\cL(X,Y)\times\cL(Y,Z)\to\cL(X,Z)\) is monotone wrt.~this preorder
  relation and if \(\cL\) is a symmetric monoidal summable category
  (that is~\ref{ax:stens} holds), then the tensor product of \(\cL\)
  is monotone wrt.~this preorder relation. If \(\cL\) is a cartesian
  summable category (that is~\ref{ax:swith} holds), then the pairing
  operation \(\cL(X,Y_1)\times\cL(X,Y_2)\to\cL(X,\With{Y_1}{Y_2})\) is
  monotone.
\end{lemma}
\begin{proof}
  Monotonicity of composition results from \Cref{lemma:summability-compl}.
  The two next statements are obvious consequences of~\ref{ax:stens}
  and~\ref{ax:swith} respectively.
\end{proof}

\begin{definition}
  We say that \(\cL\) \emph{is Scott} if the following conditions are
  satisfied:
  \begin{itemize}
  \item in any homset, the relation \(\leq\) is an order relation;
  \item for any objects \(X\) and \(Y\), any monotone sequence
    \((f_n)_{n\in\Nat}\) of elements of \(\cL(X,Y)\) has a least upper
    bound \(\sup_{n\in\Nat}f_n\in\cL(X,Y)\);
  \item the composition operation %
    \(\cL(X,Y)\times\cL(Y,Z)\to\cL(X,Z)\) is Scott-continuous in the
    sense that it commutes with the lubs of monotone sequences (its
    domain being equipped with the product order relation);
  \item if \(\cL\) is a symmetric monoidal summable category, we also
    require \(\ITens\) to commute with the lubs of monotone sequences;
  \item and if \(\cL\) is a summable resource category, then the
    functor \(\Oc\_\) is required to be monotone and to commute with the
    lubs of monotone sequences.
  \end{itemize}
\end{definition}

\begin{lemma}
  \label{lemma:curlin-continuous}
  Let \(\cL\) be a summable symmetric monoidal closed category which is Scott.
  Then \(\Curlin:\cL(\Tens ZX,Y)\to\cL(Z,\Limpl XY)\) is continuous.
\end{lemma}
\begin{proof}
  Monotonicity results from \Cref{lemma:curlin-additive}.
  Continuity results from the fact that the inverse of the map %
  \(\Curlin:\cL(\Tens ZX,Y)\to\cL(Z,\Limpl XY)\) is the function %
  \(f\mapsto\Evlin\Compl\Tensp fX\) which is continuous by our
  assumptions about \(\cL\).
\end{proof}

\begin{proposition}
  Let \(\cL\) be a cartesian summable category.
  If \(\cL\) is Scott then the pairing operation %
  \(\cL(X,Y_1)\times\cL(X,Y_2)\to\cL(X,\With{Y_1}{Y_2})\) is
  continuous.
\end{proposition}
\begin{proof}
  By the universal property of the cartesian product.
\end{proof}



\begin{theorem}
  Let \(\cL\) be a Scott summable resource category which is
  closed (as an SM category).
  Then the cartesian closed category \(\Kl\cL\) is a
  \(\lambda\)-category (in the sense of \Cref{def:lambda-category}).
%
\end{theorem}
The proof is straightforward, using \Cref{lemma:curlin-continuous}.

\begin{Example}
  Saying that \(s,t\in\Kl\COH(E,F)\) satisfy \(s\leq t\) simply means
  \(s\subseteq t\) and the the fact that \(\COH\) is Scott comes from
  the fact that the set of cliques of a coherence space is closed
  under directed unions.
  
  The situation is completely similar in \(\PCOH\).
  Given \(x,y\in\Pcoh X\), we have \(x\leq y\) iff \(x_a\leq y_a\) for
  all \(a\in\Web X\) as easily checked.
  Therefore a monotone sequence \((x(n)\in\Pcoh X)_{n\in\Nat}\) has a
  lub in \(\Pcoh X\), namely
  \(x=(\sup_{n\in\Nat}x(n)_a)_{a\in\Web X}\in\Pcoh X\).
  It follows easily that \(\PCOH\) is Scott.
  As a consequence, for any probabilistic coherence space \(X\), we
  have %
  \(\cY\in\Kl\PCOH(\Simpl XX,X)\) which maps \(t\in\Kl\PCOH(X,X)\) to
  its least fixpoint \(\sup_{n\in\Nat}\Fun t^n(0)\).
  The fact that this fixpoint operator is itself an analytic morphism
  is a remarkable property of this semantics, and is deeply related to
  the fact that, in this semantics, the morphisms are matrices with
  nonnegative coefficients.

  If this were not the case, we could accept as a morphism the
  following \(w\in\Kl\PCOH(\With\Sone\Sone,\Sone)\) such that %
  \(\Fun w(u,v)=u+v-uv\).
  In other words, for \(m\in\Web{\With\Sone\Sone}\), we have
  \(w_{m,\Sonelem}=1\) if
  \(m\in\Set{\Mset{(1,\Sonelem)},\Mset{(2,\Sonelem)}}\),
  \(w_{m,\Sonelem}=-1\) if \(m=\Mset{(1,\Sonelem),(2,\Sonelem)}\) and
  \(w_{m,\Sonelem}=0\) otherwise.
  If the fixpoint operator were accepted by this semantics, we would
  be able to define \(t\in\Kl\PCOH(\Sone,\Sone)\) such that
  \(\Fun t(u)=\Fun w(u,\Fun t(u))=u+\Fun t(u)+u\Fun t(u)\), that is %
  \(u(1-\Fun t(u))=0\).
  So we must have \(\Fun t(u)=1\) if \(u\in\Interoc 01\), and %
  \(\Fun t(0)=0\) because \(\Fun t(0)\) should be the least fixpoint
  of the function \(\Fun w(0,\_)\).
  So the function \(\Fun t:\Intercc01\to\Intercc01\) is not
  continuous, and \emph{a fortiori} cannot be described as a
  powerseries (even with possibly negative coefficients).
\end{Example}

Let \(\cL\) be a coherent differential resource category which is
closed (as an SMC) and Scott.
Given \(f\in\Kl\cL(\With ZX,\Simpl YY)\), we can define %
\(\Ap\Sfix f\in\Kl\cL(\With ZX,Y)\) as %
\(\Ap\Sfix f=\cY\Comp f\) where %
\(\cY\in\Kl\cL(\Simpl YY,Y)\) comes from
\Cref{th:lambda-cat-fixpoint-operator}.
Notice that 
\begin{equation*}
  \Ap\Sfix f=\App f{\Ap\Sfix f}\,.
\end{equation*}
Remember also that the family %
\((\Ap{\Sfix_n}f\in\Kl\cL(\With ZX,Y))_{n\in\Nat}\) %
defined inductively by \(\Ap{\Sfix_0}f=0\) and
\(\Ap{\Sfix_{n+1}}f=\App f{\Ap{\Sfix_n}f}\) %
is monotone and that
\begin{equation*}
  \Ap\Sfix f=\sup_{n\in\Nat}\Ap{\Sfix_n}f\,.
\end{equation*}

\begin{theorem}
  Let \(\cL\) be a coherent differential resource category which is
  closed (as an SMC) and Scott.
  Let \(f\in\Kl\cL(\With ZX,\Simpl YY)\), we have
  \begin{equation*}
    \Ap{\Dfun_2}{\Ap\Sfix f}
    =\Ap\Sfix{(\Simpl{\Dfun Y}{\Dmonm_Y})
      \Comp\Dfunint
      \Comp\Dfunstrf_{Y,Y}
      \Comp\A{\Dfun_2}f}
  \end{equation*}
\end{theorem}
\begin{proof}
  It suffices to prove by induction on \(n\in\Nat\) that %
  \[
    \Ap{\Dfun_2}{\Ap{\Sfix_n}f}
    =\Ap{\Sfix_n}{(\Simpl{\Dfun Y}{\Dmonm_Y})
      \Comp\Dfunint
      \Comp\Dfunstrf_{Y,Y}
      \Comp\A{\Dfun_2}f}\,.
  \]
  The base case is obvious.
  Next we have
  \begin{align*}
    &\Ap{\Dfun_2}{\Ap{\Sfix_{n+1}}f}
    =\Ap{\Dfun_2}{\App{f}{\Ap{\Sfix_n}{f}}}\\
    &\hspace{3em}=\App{(\Simpl{\Dfun Y}{\Dmonm_Y})
      \Comp\Dfunint
      \Comp\Dfunstrf_{Y,Y}
      \Comp\A{\Dfun_2}f}{\Ap{\Dfun_2}{\Ap{\Sfix_n}f}}
      \text{\quad by \Cref{th:dfun-part-app}}\\
    &\hspace{3em}=\App{(\Simpl{\Dfun Y}{\Dmonm_Y})
      \Comp\Dfunint
      \Comp\Dfunstrf_{Y,Y}
      \Comp\A{\Dfun_2}f}
      {\Ap{\Sfix_n}{(\Simpl{\Dfun Y}{\Dmonm_Y})
      \Comp\Dfunint
      \Comp\Dfunstrf_{Y,Y}
      \Comp\A{\Dfun_2}f}}\\
    &\hspace{28em}\text{by inductive hypothesis}\\
    &\hspace{3em}=\Ap{\Sfix_{n+1}}{(\Simpl{\Dfun Y}{\Dmonm_Y})
      \Comp\Dfunint
      \Comp\Dfunstrf_{Y,Y}
      \Comp\A{\Dfun_2}f}\text{\quad by definition.}
      \qedhere
  \end{align*}
\end{proof}

\begin{remark}
  So the differential of a fixpoint can itself be written as a
  fixpoint, meaning that we can combine our coherent differential
  calculus with general fixpoints, which is another major difference
  with the differential \(\lambda\)-calculus and \(\LL\).
  These results justify the way we deal with fixpoints in \(\PCFD\) in
  \Cref{sec:syntax}.
\end{remark}

\subsection{Linear and multilinear morphisms}
\label{sec:multilinearity}

\begin{definition}
  A morphism \(f\in\Kl\cL({X_1\IWith\cdots\IWith X_n},Y)\) is
  \(n\)-linear in \(\List X1n\) if there is
  \(f_0\in\cL({X_1\ITens\cdots\ITens X_n},Y)\) such that
  \(f\) coincides with the following composition of morphisms
  \begin{equation*}
    \begin{tikzcd}
      \Ap\Oc{X_1\IWith\cdots\IWith X_n}
      \ar[r,"\Inv{(\Seely^n)}"]
      &[1em]
      {\A\Oc{X_1}}\ITens\cdots\ITens{\A\Oc{X_n}}
      \ar[r,"\Der{X_1}\ITens\cdots\ITens\Der{X_n}"]
      &[5em]
      X_1\ITens\cdots\ITens X_n
      \ar[r,"f_0"]
      &[-0.6em]
      Y
    \end{tikzcd}
  \end{equation*}
  and in that case we write \(f=\Derfunm n{f_0}\).
\end{definition}

\begin{theorem}
  If \(f\in\Kl\cL(X_1\IWith\cdots\IWith X_n,Y)\) is \(n\)-linear and
  \(i\in\Set{1,\dots,n}\), then the \(i\)th partial differential %
  \(\Dfun_i f\in\Kl\cL(X_1\IWith\cdots\IWith\Dfun
  X_i\IWith\cdots\IWith X_n,\Dfun Y)\) of \(f\) satisfies the following
  commutation in \(\Kl\cL\)
  \begin{center}
    \begin{tikzcd}
      X_1\IWith\cdots\IWith\Dfun X_i\IWith\cdots\IWith X_n
      \ar[d,swap,"X_1\IWith\cdots\IWith\Sproj j\IWith\cdots\IWith X_n"]
      \ar[r,"\Dfun_i f"]
      &
      \Dfun Y
      \ar[d,"\Sproj j"]
      \\
      X_1\IWith\cdots\IWith X_i\IWith\cdots\IWith X_n
      \ar[r,"f"]
      &
      Y
    \end{tikzcd}
  \end{center}
  for \(j=0,1\).
\end{theorem}
Notice that this diagram commutes for \(j=0\) for any
\(f\in\Kl\cL(X_1\IWith\cdots\IWith X_n,Y)\), so this result concerns
only the case \(j=1\).
This result is essential in the semantics of the constructions
\(\Ifd dMPQ\) and \(\Letd dxMP\) of the language \(\PCFD\) we describe
now.

\section{A syntax for coherent differentiation}
\label{sec:syntax}

To conclude the paper, we describe briefly a syntax \(\PCFD\) which
extends Scott-Plotkin's \(\PCF\) with differentiation.
The syntax is directly derived from the semantical framework described
above.
Its theory is developed in full detail in~\cite{Ehrhard23b} to which
we refer, so that most results in this section are provided without
proofs.
We start with some simple considerations about rewriting systems.

\subsection{Rewriting systems}
\label{sec:rewriting}

Usually, a rewriting system is a set \(T\) of terms together with a
rewriting relation \(\rho\subseteq T\times T\).

In the present setting as well as in the original differential
\(\lambda\)-calculus of \Cref{sec:difflam}, a term \(t\in T\) (or a
state of the Krivine machine that we will introduce) can reduce to
several different terms, not because several redexes are available in
\(t\) (as in the usual \(\lambda\)-calculus), but because \(t\)
reduces to a ``sum'' of terms, since the
rewriting system must somehow implement the Leibniz rule of Calculus.

So the rewriting relations that we consider have type %
\(\rho\subseteq T\times\Mfin T\).
Such a relation can be lifted into a relation %
\(\Mslift\rho\subseteq\Mfin T\times\Mfin T\) defined by %
\((m,m')\in\Mslift\rho\) if \(m=\Mset t+m_1\) and \(m'=m_0+m_1\) with
\((t,m_0)\in\rho\).

\begin{lemma}
  The reflexive and transitive closure \(\Mslifttr\rho\) of
  \(\Mslift\rho\) is the least reflexive and transitive relation on
  \(\Mfin T\) such that
  \begin{itemize}
  \item if \((t,m)\in\rho\) then \((\Mset t,m)\in\Mslifttr\rho\)
  \item and if \(((m_i,m'_i)\in\Mslifttr\rho)_{i=1,2}\) then %
    \((m_1+m_2,m'_1+m'_2)\in\Mslifttr\rho\).
  \end{itemize}
\end{lemma}

To enforce the algebraic flavor of this kind of rewriting, we adopt
the following conventions:
\begin{itemize}
\item the singleton multiset \(\Mset t\in\Mfin T\) is simply written \(t\);
\item the empty multiset \(\Msetempty\in\Mfin T\) is simply written \(0\)
\item and we use \(t_1+\cdots+t_k\) for \(\Mset{\List t1k}\in\Mfin T\).
\end{itemize}
In other words we identify \(\Mfin T\) with the free
\(\Nat\)-semimodule generated by \(T\).

\subsection{Syntax}

The grammar of types is inductively defined by
\begin{align*}
  A,B,\cdots %
  \Bnfeq \Tdiffm d\Tnat \Bnfor \Timpl AB
\end{align*}
and then, given a type \(A\), we define the type \(\Tdiff A\)
inductively by \(\Tdiff{(\Tdiffm d\Tnat)}=\Tdiffm{d+1}\Tnat\) %
and \(\Tdiff{(\Timpl AB)}=\Timpl A{\Tdiff B}\).

The syntax of terms is inductively defined as follows; we split it
into three kinds of constructions:
\begin{align*}
  M,N,P,Q,\dots
  \Bnfeq
  x
  &\Bnfor
    \App PN
    \Bnfor
    \Abst xAM
    \Bnfor
    \Fix M
  &\lambda\text{-calculus}
  \\
  &\Bnfor
    \Num\nu
    \Bnfor
    \Succd dM
    \Bnfor
    \Predd dM
    \Bnfor
    \Ifd dMPQ
    \Bnfor
    \Letd dxMN
  &\text{arithmetics}
  \\
  &\Bnfor
    \Pd M
    \Bnfor
    \Pdinj diM
    \Bnfor
    \Pdsum dM
    \Bnfor
    \Pdflip dlM
    \Bnfor
    \Pdproj diM
  &\text{differentiation}
\end{align*}
where \(\nu\in\Nat\), \(d,l\in\Nat\) and \(i\in\Set{0,1}\).

\begin{definition}
  We say that a type \(E\) is \emph{sharp} if it is cannot be
  written %
  \(E=\Tdiff A\) for some other type \(A\), which simply means that %
  \(E=(\Timpl{A_1}{\Timpl{\cdots}{\Timpl{A_n}{\Tnat}}})\) %
  (with the usual convention that \(\Timpl{}{}\) associates on the
  right).
  We use letters \(E\), \(F\) to denote types when we want to stress
  that they are sharp.
\end{definition}

\begin{remark}
  For any type \(A\) there is exactly one \(d\in\Nat\) and one sharp
  type such that \(A=\Tdiffm dE\).
\end{remark}

The typing rules are given in \Cref{fig:cdpcf-typing}.
\begin{figure}
  \centering
  \begin{center}
  \begin{prooftree}
    \infer0{\Tseq{\Gamma,x:A}{x}{A}}
  \end{prooftree}
  \Treesep
  \begin{prooftree}
    \hypo{\Tseq\Gamma P{\Timpl AB}}
    \hypo{\Tseq\Gamma NA}
    \infer2{\Tseq\Gamma{\App PN}B}
  \end{prooftree}
  \Treesep
  \begin{prooftree}
    \hypo{\Tseq{\Gamma,x:A}MB}
    \infer1{\Tseq\Gamma{\Abst xAM}{\Timpl AB}}
  \end{prooftree}
  \Treesep
  \begin{prooftree}
    \hypo{\Tseq\Gamma M{\Timpl AA}}
    \infer1{\Tseq\Gamma{\Fix M}A}
  \end{prooftree}
  \end{center}
  \begin{center}
    \begin{prooftree}
      \hypo{\nu\in\Nat}
      \infer1{\Tseq\Gamma{\Num\nu}\Tnat}
    \end{prooftree}
    \Treesep
    \begin{prooftree}
      \hypo{\Tseq\Gamma M{\Tdiffm d\Tnat}}
      \infer1{\Tseq\Gamma{\Succd dM}{\Tdiffm d\Tnat}}
    \end{prooftree}
  \end{center}
  \begin{center}
    \begin{prooftree}
      \hypo{\Tseq\Gamma M{\Tdiffm d\Tnat}}
      \hypo{\Tseq\Gamma PA}
      \hypo{\Tseq\Gamma QA}
      \infer3{\Tseq\Gamma{\Ifd dMPQ}{\Tdiffm dA}}
    \end{prooftree}
    \Treesep
    \begin{prooftree}
      \hypo{\Tseq\Gamma M{\Tdiffm d\Tnat}}
      \hypo{\Tseq{\Gamma,x:\Tnat}PA}
      \infer2{\Tseq\Gamma{\Letd dxMP}{\Tdiffm dA}}
    \end{prooftree}
  \end{center}
  \begin{center}
    \begin{prooftree}
      \hypo{\Tseq\Gamma M{\Timpl AB}}
      \infer1{\Tseq\Gamma{\Pd M}{\Timpl{\Tdiff A}{\Tdiff B}}}
    \end{prooftree}
    \Treesep
    \begin{prooftree}
      \hypo{\Tseq{\Gamma}{M}{\Tdiffm dA}}
      \infer1{\Tseq{\Gamma}{\Pdinj diM}{\Tdiffm{d+1}A}}
    \end{prooftree}
    \Treesep
    \begin{prooftree}
      \hypo{\Tseq{\Gamma}{M}{\Tdiffm{d+2}A}}
      \infer1{\Tseq{\Gamma}{\Pdsum dM}{\Tdiffm{d+1}A}}
    \end{prooftree}
  \end{center}
  \begin{center}
    \begin{prooftree}
      \hypo{\Tseq{\Gamma}{M}{\Tdiffm{d+l+2}A}}
      \infer1{\Tseq{\Gamma}{\Pdflip dlM}{\Tdiffm{d+l+2}A}}
    \end{prooftree}
    \Treesep
    \begin{prooftree}
      \hypo{\Tseq{\Gamma}{M}{\Tdiffm{d+1}A}}
      \infer1{\Tseq{\Gamma}{\Pdproj diM}{\Tdiffm{d}A}}
    \end{prooftree}
  \end{center}
  \caption{Typing rules for \(\PCFD\)}
  \label{fig:cdpcf-typing}
\end{figure}

As explained in~\cite{Ehrhard23b}, this syntax can be equipped with a
rewriting system \(\Pcfred\) which is inspired by the categorical
setting described previously.
Beyond ordinary substitution, the definition of this rewriting system
requires a ``differential modification'' operator \(\Diffm xM\) whose
definition, by induction on \(M\), is given if \Cref{fig:diffm-def}.

\begin{remark}
  It is important to notice that, in sharp contrast with the linear
  substitution %
  \(\frac{\partial M}{\partial x}\cdot N\) of the differential
  \(\lambda\)-calculus, the construction \(\Diffm xM\) \emph{does not
    introduce actual sums of terms}, but only potential ones by
  inserting \(\Pdsum d\_\) syntactic constructs at various places.
\end{remark}

\begin{figure}
  \centering
  \begin{align*}
    \Diffm xx&=x
    &
      \Diffm xy&=\Pdinj00 y \text{\quad if }y\not=x
    \\
    \Diffm x{\App PN}&=\App{\Pdsum 0{\Tdiff{\Diffm xP}}}{\Diffm xN}
    &
      \Diffm x{\Abst yBP}&=\Abst yB{\Diffm xP}
    \\
    \Diffm x{\Fix P}&=\Fix{\Pdsum 0{\Tdiff{\Diffm xP}}}
    &&
    \\
    \Diffm x{\Num\nu}&=\Pdinj00{\Num\nu}
    &
      \Diffm x{\Succd dN}&=\Succd{d+1}{\Diffm xN}
    \\
    \Diffm x{\Ifd dNPQ}
             &=\Pdsum0{\Pdflip 0d{\Ifd{d+1}{\Diffm xN}{\Diffm xP}{\Diffm xQ}}}
    &&
    \\
    \Diffm x{\Letd dyNP}
             &=\Pdsum0{\Pdflip 0d{\Letd{d+1}y{\Diffm xN}{\Diffm xP}}}
    &&
    \\
             &\text{\quad\quad\quad with }y\not=x
               &&
  \end{align*}
  \caption{Definition of \(\Diffm xM\)}
  \label{fig:diffm-def}
\end{figure}

\begin{lemma}
  \label{lemma:diffm-typing}
  If \(\Tseq{\Gamma,x:A}{M}{B}\) then
  \(\Tseq{\Gamma,x:\Tdiff A}{\Diffm xM}{\Tdiff B}\).
\end{lemma}
\begin{proof}[Proof hint]
  Straightforward induction on \(M\), or rather on the derivation of
  the typing judgment \(\Tseq{\Gamma,x:A}{M}{B}\).
  As a first example, assume that \(M=\Fix N\) with %
  \(\Tseq{\Gamma,x:A}{N}{\Timpl BB}\) so that
  \(\Tseq{\Gamma,x:A}{\Fix N}{B}\).
  By inductive hypothesis, we have %
  \(\Tseq{\Gamma,x:\Tdiff A}{\Diffm xN}{\Timpl B{\Tdiff B}}\) and
  hence %
  \(\Tseq{\Gamma,x:\Tdiff A} {\Pd{\Diffm xN}}{\Timpl{\Tdiff B}{\Tdiffm
      2B}}\) so that %
  \(\Tseq{\Gamma,x:\Tdiff A}{\Pdsum0{\Pd{\Diffm xN}}} {\Timpl{\Tdiff
      B}{\Tdiff B}}\) since %
  \((\Timpl{\Tdiff B}{\Tdiffm 2B}) = \Tdiffm2{(\Timpl{\Tdiff B}B)}\).
  Therefore
  \(\Tseq{\Gamma,x:\Tdiff A}{\Fix{\Pdsum0{\Pd{\Diffm xN}}}}{\Tdiff
    B}\) as required.

  As a second example, take \(M=\Ifd dNPQ\) with %
  \(\Tseq{\Gamma,x:A}{N}{\Tdiffm d\Tnat}\),
  \(\Tseq{\Gamma,x:A}{P}{C}\) and
  \(\Tseq{\Gamma,x:A}{Q}{C}\) so that \(B=\Tdiffm dC\).
  By inductive hypothesis, we have %
  \(\Tseq{\Gamma,x:\Tdiff A}{\Diffm xN}{\Tdiffm{d+1}\Tnat}\),
  \(\Tseq{\Gamma,x:\Tdiff A}{\Diffm xP}{\Tdiff C}\) and
  \(\Tseq{\Gamma,x:\Tdiff A}{\Diffm xQ}{\Tdiff C}\).
  It follows that %
  \[
    \Tseq{\Gamma,x:\Tdiff A}{\Ifd{d+1}{\Diffm xN}{\Diffm xP}{\Diffm xQ}}
    {\Tdiffm{d+1+1}C}
  \]
  and hence 
  \[
    \Tseq{\Gamma,x:\Tdiff A}
    {\Pdflip 0d{\Ifd{d+1}{\Diffm xN}{\Diffm xP}{\Diffm xQ}}}
    {\Tdiffm{d+2}C}
  \]
  and finally
  \[
    \Tseq{\Gamma,x:\Tdiff A}
    {\Pdsum0{\Pdflip 0d{\Ifd{d+1}{\Diffm xN}{\Diffm xP}{\Diffm xQ}}}}
    {\Tdiff B}
  \]
  as required.
\end{proof}

\begin{remark}
  The main purpose of this differential modification is to allow the
  following rewriting
  \begin{align*}
    \Ap\Pd{\Abst xAM}\Rel\Pcfred
    \Abst x{\Tdiff A}{\Diffm xM}\,,
  \end{align*}
  which gives its operational meaning to the \(\Pd\) construction of
  the syntax, exactly as the ordinary \(\beta\)-rewriting %
  \(\App{\Abst xAM}N\Rel\Pcfred\Subst MNx\) gives its operational
  meaning to the application construct of the \(\lambda\)-calculus.

  In contrast with the differential substitution of the differential
  \(\lambda\)-calculus, differentiation in \(\PCFD\) requires a
  combination of differential modification and ordinary substitution.
  Let \(N\) be such that \(\Tseq\Gamma N{\Tdiff A}\), so that \(N\)
  should be intuitively considered as the pair made of
  \((N_i=\Pdproj 0iN)_{i=0,1}\) which satisfy
  \((\Tseq\Gamma{N_i}A)_{i=0,1}\) and are summable in the type \(A\).
  Then the term %
  \[
    \Pdproj 01{\Subst{\Diffm xM}{N}x}
  \]
  of \(\PCFD\) has the same meaning as the term
  \[
    \Subst{\left(\Diffp{M}x{\Pdproj01N}\right)}{\Pdproj00N}x
  \]
  of the differential \(\lambda\)-calculus.
  For that reason we can understand \(\PCFD\) as a sublanguage of the
  differential \(\lambda\)-calculus (extended with integers and
  fixpoint operators).
\end{remark}

\begin{remark}
  It is important to notice that the rewriting \(\Pcfred\) we equip
  \(\PCFD\) with is not an ordinary rewriting relation from terms to
  terms, but from terms to finite multisets of terms as explained in
  \Cref{sec:rewriting}.
  More specifically, there are exactly three rewriting rules which
  produce non-singleton multisets, namely %
  \(\Pdproj di{\Pdinj d{1-i}M}\Rel\Pcfred 0\) (for \(i=0,1\)) and %
  \(\Pdproj d1{\Pdsum dM} \Rel\Pcfred\Pdproj d0{\Pdproj d1 M}+\Pdproj
  d1{\Pdproj d0 M}\).
  For this rewriting system, one can prove a form of subject reduction
  which expresses that if %
  \(\Tseq\Gamma MA\) and \(M\Rel\Pcfred\sum_{i=1}^kM_i\) %
  then we have \(\Tseq\Gamma{M_i}A\) for \(i=1,\dots,k\).
  Notice that actually \(k\in\Set{0,1,2}\).
\end{remark}

\Cref{lemma:diffm-typing}, together with an ordinary substitution
lemma (if \(\Tseq{\Gamma,x:A}MB\) and \(\Tseq\Gamma NA\) then
\(\Tseq{\Gamma}{\Subst MNx}{B}\)), allows to prove subject reduction.

\begin{theorem}
  If \(\Tseq\Gamma MA\) and \(M\Rel\Pcfred\sum_{i=1}^k M_i\) then
  \((\Tseq\Gamma{M_i}A)_{i=1}^k\).
\end{theorem}

\begin{corollary}
  Assume that \(\List M1p\) are terms such that %
  \((\Tseq\Gamma{M_i}A)_{j=1}^p\).
  If \(\sum_{j=1}^pM_j\Rel{\Mslifttr{\Pcfred}}\sum_{j=1}^{p'}M'_{j}\)
  then we have
  \((\Tseq\Gamma{M'_j}A)_{j=1}^{p'}\).
\end{corollary}

\subsection{Operational semantics}
Rather than providing a complete definition of the \(\PCFD\) reduction
system, which is lengthy and has already been given in~\cite{Ehrhard23b}%
\footnote{We hope to be able to improve and somehow simplify this
  system in a near future.} %
we describe a seemingly more canonical ``Krivine machine'' which
allows to evaluate terms \(M\) of \(\PCFD\) such that
\(\Tseq{}M\Tnat\).

A state of the machine is a triple %
\(\State \delta Ms\) where \(M\) a closed term, \(s\) is a stack and
\(\delta\in\Finseq{\Set{0,1}}\).
Our stacks are defined by the following grammar:
\begin{align*}
  s,t,\dots
  \Bnfeq \Stempty
  \Bnfor \Starg Ms
  \Bnfor \Stsucc s
  \Bnfor \Stpred s
  \Bnfor \Stif \delta PQs
  \Bnfor \Stlet \delta xPs
  \Bnfor \Stdiff is\,.
\end{align*}
Stacks are typed by judgments of shape \(\Stseq sE\) where \(E\) is a
sharp type.
The typing rules for stacks are given in \Cref{fig:typing-stacks}.

\begin{figure}
  \begin{center}
  \begin{prooftree}
    \infer0{\Stseq\Stempty\Tnat}
  \end{prooftree}
  \Treesep
  \begin{prooftree}
    \hypo{\Stseq s\Tnat}
    \infer1{\Stseq{\Stsucc s}\Tnat}
  \end{prooftree}
  \Treesep
  \begin{prooftree}
    \hypo{\Stseq s\Tnat}
    \infer1{\Stseq{\Stpred s}\Tnat}
  \end{prooftree}
  \end{center}
  \begin{center}
    \begin{prooftree}
      \hypo{\Tseq{}{P}{\Tdiffm dE}}
      \hypo{\Tseq{}{Q}{\Tdiffm dE}}
      \hypo{\Stseq sE}
      \hypo{\delta\in\Set{0,1}^d}
      \infer4{\Stseq{\Stif\delta PQs}\Tnat}
    \end{prooftree}
    \Treesep
    \begin{prooftree}
      \hypo{\Tseq{}{P}{\Tdiffm dE}}
      \hypo{\Stseq sE}
      \hypo{\delta\in\Set{0,1}^d}
      \infer3{\Stseq{\Stlet\delta xPs}\Tnat}
    \end{prooftree}
  \end{center}
  \begin{center}
    \begin{prooftree}
      \hypo{\Tseq{}PA}
      \hypo{\Stseq sE}
      \infer2{\Stseq{\Starg Ps}{\Timpl AE}}
    \end{prooftree}
    \Treesep
    \begin{prooftree}
      \hypo{\Stseq{s}{\Timpl{\Tdiff A}{E}}}
      \hypo{i\in\Set{0,1}}
      \infer2{\Stseq{\Stdiff is}{\Timpl AE}}
    \end{prooftree}
  \end{center}
  \caption{Typing rules for stacks}
  \label{fig:typing-stacks}
\end{figure}

\begin{definition}
  A state \(e=\State\delta Ms\) is \emph{well typed} if %
  \(\Tseq{}M{\Tdiffm dE}\), \(\delta\in\Set{0,1}^{\Len\delta}\) and %
  \(\Stseq sE\) for some sharp type \(E\).
\end{definition}

The transition rules for states are given in
\Cref{fig:state-reductions-push,fig:state-reductions-pop,fig:state-reductions-access}
where we classified them in three categories.
Notice that this is a rewriting system in the sense of
\Cref{sec:rewriting}, that is, from states to finite multisets (or
finite formal sums) of states.
The only rules yielding actual sums are the 3rd and 5th transition
rules in \Cref{fig:state-reductions-access}.

With a stack \(s\) such that \(\Tseqst sE\) we can associate a context
\(\Stctx s\Holee\), that is, as a closed term of type \(\Tnat\) with
one hole \(\Holee\) of type \(E\) in linear position; the definition
of this context is given in %
\Cref{fig:stack-to-context}.
\begin{figure}
  \centering
  \begin{align*}
    \Stctx\Stempty&=\Holee
    &&\\
    \Stctx{\Stsucc s}&=\Stctx s\Holef{\Succ{\Holee}}
    &
      \Stctx{\Stpred s}&=\Stctx s\Holef{\Pred{\Holee}}
    \\
    \Stctx{\Stif\delta PQs}
                  &=\Stctx s\Holef{\Pdproj 0\delta{\Ifd0{\Holee}{P}{Q}}}
    &
      \Stctx{\Stlet\delta xPs}
                     &=\Stctx s\Holef{\Pdproj 0\delta{\Letd0 x{\Holee}{P}}}
    \\
    \Stctx{\Starg Ps}&=\Stctx s\Holef{\App{\Holee}P}
    &
      \Stctx{\Stdiff is}&=\Stctx s\Holef{\Pdproj 0i{\Pd\Holee}}
  \end{align*}
  \caption{Context associated with a stack}
  \label{fig:stack-to-context}
\end{figure}

\begin{lemma}
  \label{lemma:typing-context-of-stack}
  If \(\Tseqst sE\) and \(\Tseq{}ME\), then
  \(\Tseq{}{\Stctx s\Holef M}\Tnat\).
\end{lemma}
\begin{proof}
  Straightforward induction on \(s\).
\end{proof}

\begin{definition}
  Given a state \(e\) we define a term %
  \(\Termofst e\) by %
  \(\Termofst{\State\delta Ms}=\Stctx s\Holef{\Pdproj0\delta{M}}\).
\end{definition}

\begin{lemma}
  If \(e\) is a well typed state then \(\Tseq{}{\Termofst e}{\Tnat}\).
\end{lemma}
\begin{proof}
  Direct consequence of \Cref{lemma:typing-context-of-stack}.
\end{proof}

\begin{theorem}
  \label{th:state-trans-to-term-red}
  If \(e\Rel\Stred\sum_{i=1}^ke_i\) then %
  \(\Termofst e\Rel\Pcfred\sum_{i=1}^k\Termofst{e_i}\) (notice that
  \(k\in\Set{0,1,2}\)).
\end{theorem}

\begin{remark}
  Performing transitions from the state \(e=\State\delta Ms\)
  amounts actually to evaluating the term %
  \(\Pdproj 0\delta M\) in the environment \(s\) in a ``weak head''
  restriction of the \(\PCFD\) reduction system.
  As a whole, we could consider \((\delta,s)\) as the context %
  \(\Stctx s[\Pdproj{0}\delta{[\ ]}]\) which suggests to integrate the
  access path \(\delta\) in the stack since the purpose of the stack
  is to store the current context of evaluation.
  
  We did not do so because many rules of the rewriting system
  \(\Pcfred\) express some commutations between the \(\Pdproj di\_\)
  constructs (stored in the \(\delta\) component of the state) and the
  other constructs of the language (stored in the \(s\) component).
  These commutations express that the \(s\)-context and the
  \(\delta\)-context act \emph{in parallel} on the term component of
  the machine, strongly suggesting to keep them separate.
  The benefit of this choice is that, in the transition rules of %
  \Cref{fig:state-reductions-push,fig:state-reductions-pop,fig:state-reductions-access},
  we do not mention these commutations explicitly: they are
  implemented in a purely implicit way, which is a major improvement
  of this machine wrt.~the rewriting system \(\Pcfred\).

  For instance, the rewriting system features the reduction %
  \(\Pdproj{0}i{\App NP}\Rel\Pcfred\App{\Pdproj{0}iN}P\) where we see
  that the action of the projection is transferred from \(\App NP\) to
  \(N\).
  This transfer of action of the projection is implemented implicitly
  in the transitions
  \(\State{\delta}{\Pdproj{0}i{\App NP}}{s}
  \Rel\Stred\State{\delta i}{\App NP}s
  \Rel\Stred\State{\delta i}N{\Starg Ps}\)
  of \Cref{fig:state-reductions-access,fig:state-reductions-push}.
  Using only a stack for storing the context, we would have
  obtained a sequence of reductions like %
  \((\Pdproj{0}i{\App NP},s)
  \Rel\Stred(\App NP,\pi_i\cdot s)
  \Rel\Stred(N,\Starg P{\pi_i\cdot s})\).
  But we might have \(N=\Pdsum{0}{\Abst xBQ}\) and then the only
  natural option ---~keeping in mind the fundamental principle that
  the stack should be accessed only from the top~--- would be to push
  again the \(\theta^0\) onto the stack, leading to something like %
  \((\Abst xBQ,\theta^0\cdot\Starg P{\pi_i\cdot s})\) but then the
  argument that the abstraction \(\Abst xBQ\) is waiting for is not
  available on the top of the stack.
  To solve this issue we would need an equivalence relation on stacks
  accounting for the above mentioned commutation reduction rules of
  \(\Pcfred\).
  In other words, the stack should not be indexed by a finite totally
  ordered set (that is, should not be a list), but rather by a tree or
  perhaps a more general directed acyclic graph.
  Our dichotomy between the stack and the access word avoids these
  technicalities in a very simple and, we think, natural way.
  Notice by the way that the access word is not dealt with as a stack
  since we insert and remove elements anywhere in the word, and even
  perform cyclic permutations of factors, see
  \Cref{fig:state-reductions-access}; a simple implementation of such
  a data structure could use a linked list.
\end{remark}

\begin{definition}
  \label{def:final-state}
  A \emph{final state} is a state of shape
  \(\State{\Seqempty}{\Num\nu}{\Stempty}\).
\end{definition}

\begin{lemma}
  If \(e\) is a well typed state and there is no transition from
  \(e\), then \(e\) is a final state.
\end{lemma}
\begin{proof}
  Simple case analysis on the typing rules of terms and stacks.
\end{proof}

\begin{figure}
  \centering
  \begin{align*}
    \State\delta{\Succd{\Len\delta}M}s
    &\Rel\Stred\State\delta M{\Stsucc s}
    &\State\delta{\Predd{\Len\delta}M}s
    &\Rel\Stred\State\delta M{\Stpred s}
    \\
    \State{\epsilon\delta}{\Ifd{\Len\delta}MPQ}s
    &\Rel\Stred\State\delta M{\Stif{\epsilon}{P}{Q}{s}}
    &\State{\epsilon\delta}{\Letd{\Len\delta}xMP}s
    &\Rel\Stred\State\delta M{\Stlet{\epsilon}{x}{P}{s}}
    \\
    \State\delta{\App NP}s
    &\Rel\Stred\State\delta N{\Starg Ps}
    &\State\delta{\Fix N}s
    &\Rel\Stred\State\delta N{\Starg{\Fix N}s}
    \\
    \State{\delta i}{\Pd N}{s}
    &\Rel\Stred\State\delta N{\Stdiff is}
    &&
  \end{align*}
  \caption{Transition rules for states --- pushing onto the stack}
  \label{fig:state-reductions-push}
\end{figure}
\begin{figure}
  \centering
  \begin{align*}
    \State{\Seqempty}{\Num\nu}{\Stsucc s}
    &\Rel\Stred\State\Seqempty{\Num{\nu+1}}s
    &\State{\Seqempty}{\Num\nu}{\Stpred s}
    &\Rel\Stred\State\Seqempty{\Num{\nu\Kminus1}}s
    \\
    \State{\Seqempty}{\Num0}{\Stif\delta PQs}
    &\Rel\Stred\State\delta Ps
    &\State{\Seqempty}{\Num{\nu+1}}{\Stif\delta PQs}
    &\Rel\Stred\State\delta Qs
    \\
    \State{\Seqempty}{\Num\nu}{\Stlet\delta xPs}
    &\Rel\Stred\State\delta {\Subst P{\Num\nu}x}s
    &&
    \\
    \State{\delta}{\Abst xBN}{\Starg Ps}
    &\Rel\Stred\State\delta{\Subst NPx}{s}
    &\State{\delta}{\Abst xBN}{\Stdiff is}
    &\Rel\Stred\State{\delta i}{\Abst{x}{\Tdiff B}{\Diffm xN}}{s}
  \end{align*}
  \caption{Transition rules for states --- popping from the stack}
  \label{fig:state-reductions-pop}
\end{figure}
\begin{figure}
  \centering
  \begin{align*}
    \State{\epsilon\delta}{\Pdproj{\Len\delta}{i}N}s
    &\Rel\Stred\State{\epsilon i\delta}Ns
    &&\\
    \State{\epsilon i\delta}{\Pdinj{\Len\delta}{i}N}s
    &\Rel\Stred\State{\epsilon\delta}Ns
    &\State{\epsilon i\delta}{\Pdinj{\Len\delta}{1-i}N}s
    &\Rel\Stred0
    \\
    \State{\epsilon 0\delta}{\Pdsum{\Len\delta}N}{s}
    &\Rel\Stred\State{\epsilon 00\delta}{N}{s}
    &\State{\epsilon 1\delta}{\Pdsum{\Len\delta}N}{s}
    &\Rel\Stred\State{\epsilon 01\delta}{N}{s}+\State{\epsilon 10\delta}{N}{s}
    \\
    \State{\epsilon i_1\cdots i_{l+2}\delta}{\Pdflip{\Len\delta}l N}{s}
    &\Rel\Stred
      \State{\epsilon i_{l+2}i_1\cdots i_{l+1}\delta}{N}{s}
      &&
  \end{align*}
  \caption{Transition rules for states --- handling the access word}
  \label{fig:state-reductions-access}
\end{figure}

Let \(\States\) be the set of all well typed states.

\begin{lemma}
  If \(e\in\States\), \(e\Rel\Stred u\) and \(e'\) is a state such that
  \(u_{e'}\not=0\), then \(e'\in\States\).
\end{lemma}
\begin{proof}
  Simple inspection of the transition rules.
  As an example taken from \Cref{fig:state-reductions-access}, assume
  that \(e=\State{\epsilon 1\delta}{\Pdsum{\Len\delta}N}{s}\) and
  \(u=\State{\epsilon 01\delta}{N}{s}+\State{\epsilon
    10\delta}{N}{s}\).
  So
  \(e'\in\Set{\State{\epsilon 01\delta}{N}{s},\State{\epsilon
      10\delta}{N}{s}}\), say \(e'=\State{\epsilon 10\delta}{N}{s}\).
  Let \(d=\Len\delta\).
  There must be a type \(A\) such that \(\Tseq{}N{\Tdiffm{d+2}A}\) so that
  \(\Tseq{}{\Pdsum dN}{\Tdiffm{d+1}A}\).
  There are uniquely determined sharp type \(E\) and \(h\in\Nat\) such
  that %
  \(A=\Tdiffm hE\) and hence \(\Tseq{}{\Pdsum dN}{\Tdiffm{h+1+d}E}\)
  and since \(e\) is well typed we must have \(h=\Len\epsilon\) and
  \(\Tseqst sE\).
  So \(\Len{\epsilon10\delta}=h+2+d\) and since %
  \(\Tseq{}N{\Tdiffm{h+2+d}E}\), the state \(e'\) is well typed.

  Let us also deal with the case %
  \(e=\State{\delta i}{\Pd N}{s}\) and %
  \(e'=\State\delta N{\Stdiff is}\) of \Cref{fig:state-reductions-push}.
  For \(\Pd N\) to be typed we need \(\Tseq{}N{\Timpl BA}\) and then
  we have \(\Tseq{}{\Pd N}{\Timpl{\Tdiff B}{\Tdiff A}}\).
  There are uniquely determined sharp type \(E\) and \(d\in\Nat\) such
  that %
  \(A=\Tdiffm dE\), so that %
  \(\Tseq{}{\Pd N}{\Tdiffm{d+1}{(\Timpl{\Tdiff B}E)}}\).
  Since \(e\) is well typed, we must have %
  \(d=\Len\delta\) and \(\Tseqst s{\Timpl{\Tdiff B}E}\) so that %
  \(\Tseqst{\Stdiff is}{\Timpl{B}E}\) (see
  \Cref{fig:typing-stacks}) and hence \(e'\) is well typed since we
  have %
  \(\Tseq{}{N}{\Tdiffm d{(\Timpl BE)}}\).

  As a last example consider the case %
  \(e=\State{\delta}{\Abst xBN}{\Stdiff is}\) and
  \(e'=\State{\delta i}{\Abst{x}{\Tdiff B}{\Diffm xN}}{s}\) from %
  \Cref{fig:state-reductions-pop}.
  We must have \(\Tseq{x:B}{N}{A}\) for some type \(A=\Tdiffm dE\)
  (with \(E\) sharp and \(d\in\Nat\) uniquely defined).
  Accordingly \(\Tseq{}{\Abst xBN}{\Tdiffm d{(\Timpl BE)}}\) and since
  \(e\) is well typed, we must have \(d=\Len\delta\) and
  \(\Tseqst{\Stdiff is}{\Timpl BE}\) which, by the typing rules of %
  \Cref{fig:typing-stacks}, entails \(\Tseqst s{\Timpl{\Tdiff B}{E}}\).
  By \Cref{lemma:diffm-typing}, we have %
  \(\Tseq{x:\Tdiff B}{\Diffm xN}{\Tdiffm{d+1}E}\) so that \(e'\) is
  well typed since \(\Len{\delta i}=d+1\).
\end{proof}

\begin{remark}
  Notice that, on one side, the transition rules of
  \Cref{fig:state-reductions-push,fig:state-reductions-pop,fig:state-reductions-access}
  are deterministic in the sense that, for any \(e\in\States\), there
  is at most one \(u\in\Mfin\States\) such that \(e\Rel\Pcfred u\),
  and that, when there is no such transition from \(e\), then \(e\) is
  final in the sense of \Cref{def:final-state}.
  So we can define a function \(\Pcfredf:\States\to\Mfin\States\) such
  that \(\Pcfredf(e)=u\) if \(e\Rel\Pcfred u\) and \(\Pcfredf(e)=e\)
  if \(e\) is final.

  On the other side, these transition rules contain some
  nondeterminism precisely in the fact that transitions are from a
  state \(e\) to a finite multisets of states \(e_1+\cdots+e_k\)
  (which can be understood as the various possible results of a
  transition from \(e\)) and not from states to states.
  One of the purposes of the next section is to show that this
  nondeterminism is an illusion.
\end{remark}

\subsection{Denotational semantics}

Let \(\cL\) be a coherent differential resource category which is
closed (as an SMC) and Scott and where the coproduct
\(\Snat=\Bplus_{i\in\Nat}\Sone\) exists.

First, we define by induction on the type \(A\) an object \(\Tsem A\).
We take \(\Tsem\Tnat=\Snat\) and more generally
\(\Tsem{\Tdiffm d\Tnat}=\Dfun^d\Snat\).
And then \(\Tsem{\Timpl AB}=\Limplp{\Oc{\Tsem A}}{\Tsem B}\).

Then, given a term \(M\) of \(\PCFD\), a context %
\(\Gamma=(x_1:A_1,\dots,x_n:A_n)\) and a type \(B\) such that %
\(\Tseq\Gamma MB\), one defines by induction on the typing
derivation of \(\Tseq\Gamma MB\) (that is, by induction on \(M\)) an
element \(\Psem M\Gamma\in\Kl\cL(\Tsem\Gamma,\Tsem B)\).

We refer to~\cite{Ehrhard23b} for the precise definition of this
interpretation of terms, the syntax of \(\PCFD\) has been chosen in
order to make it fairly straightforward.
Concerning the ``object of integers'' \(\Snat\), we use
\begin{itemize}
\item the existence of a canonical isomorphism
  \(\Snatiso\in\cL(\Plus\Sone\Snat,\Snat)\), which is the key
  ingredient for interpreting \(\Succd dM\), \(\Predd dM\) and
  \(\Ifd dMPQ\);
\item the existence of a canonical \(\oc\)-coalgebra structure on
  \(\Snat\), which is the key ingredient for interpreting
  \(\Letd dxMP\).
  This is due to the fact that \(\Sone\) is a \(\oc\)-coalgebra
  (thanks to the Seely isomorphisms) and the fact that
  \(\oc\)-coalgebras are closed under arbitrary colimits which exist
  in \(\cL\).
\end{itemize}
As an example, using the definition of \(\Snat\) as a coproduct, we
can define, for any object \(X\) of \(\cL\), a morphism
\(f\in\cL(\Snat\ITens\Withp XX,X)\) uniquely characterized by
\begin{center}
  \begin{tikzcd}
    \Tens\Sone{\Withp XX}
    \ar[d,swap,"\Leftu"]
    \ar[r,"\Tens{\Snum 0}{\Withp XX}"]
    &[2em]
    \Tens\Snat{\Withp XX}
    \ar[d,"f"]
    \\
    \With XX
    \ar[r,"\Proj1"]
    &
    X
  \end{tikzcd}
  \Treesep
  \begin{tikzcd}
    \Tens\Sone{\Withp XX}
    \ar[d,swap,"\Leftu"]
    \ar[r,"\Tens{\Snum{\nu+1}}{\Withp XX}"]
    &[2em]
    \Tens\Snat{\Withp XX}
    \ar[d,"f"]
    \\
    \With XX
    \ar[r,"\Proj2"]
    &
    X
  \end{tikzcd}
\end{center}
where, for \(\nu\in\Nat\), the morphism
\(\Snum\nu\in\cL(\Sone,\Snat)\) is the \(\nu\)th injection of
\(\Sone\) into the coproduct \(\Snat\).
We set \(\Sif=\Derfun_2 f\in\Kl\cL(\Snat\IWith (X\IWith X),Y)\) which
is bilinear (see \Cref{sec:multilinearity}).
Then using \Cref{def:kleisli-partial-derivatives} we define %
\(\Sif^d =\A{\Dfun_1^d}\Sif\in\Kl\cL(\Dfun^d\Snat\IWith(\With
XX),\Dfun^d X)\) %
that we use straightforwardly to interpret the \(\Ifd dMPQ\) construct
of \(\PCFD\).

Then one can prove a standard substitution lemma.
\begin{lemma}
  \label{lemma:substit-sem}
  If \(\Tseq{\Gamma,x:A}MB\) and \(\Tseq\Gamma NA\), one has %
  \(\Psem{\Subst MNx}\Gamma=\Psem
  M{\Gamma,x:A}\Comp\Tuple{\Tsem\Gamma,\Psem N\Gamma}\) %
  in \(\Kl\cL\).
\end{lemma}
Notice indeed that %
\(\Psem M{\Gamma,x:A}\in\Kl\cL(\Tsem\Gamma\IWith\Tsem A,\Tsem B)\) %
and
\(\Tuple{\Tsem\Gamma,\Psem N\Gamma}
\in\Kl\cL(\Tsem\Gamma,\Tsem\Gamma\IWith\Tsem A)\).

We have an analogous lemma for the differential modification.
\begin{lemma}
  \label{lemma:modif-sem}
  If \(\Tseq{\Gamma,x:A}MB\) then %
  \(\Psem{\Diffm xM}{\Gamma,x:\Tdiff A}
  =\Ap{\Dfun_2}{\Psem M{\Gamma,x:A}}\).
\end{lemma}
Notice that %
\(\Psem M{\Gamma,x:A}\in\Kl\cL(\Tsem\Gamma\IWith\Tsem A,\Tsem B)\) %
and hence
\(\Ap{\Dfun_2}{\Psem
  M{\Gamma,x:A}}\in\Kl\cL(\Tsem\Gamma\IWith\A\Dfun\Tsem A,\A\Dfun\Tsem
B)\) so that the equation above is well typed.

\begin{theorem}
  If \(\Tseq\Gamma MA\) and \(M\Rel\Pcfred\sum_{i=1}^k M_i\) then the
  morphisms %
  \(\Psem{M_i}\Gamma\in\Kl\cL(\Tsem\Gamma,\Tsem A)\) are summable in
  \(\cL(\A\Oc{\Tsem\Gamma},\Tsem A)\) and we have
  \(\Psem M\Gamma=\sum_{i=1}^k\Psem{M_i}\Gamma\).
\end{theorem}
This result expresses the soundness of this denotational semantics.
The proof uses \Cref{lemma:substit-sem,lemma:modif-sem}.
Using the notions introduced in \Cref{sec:rewriting}, this generalizes
easily as follows.

\begin{corollary}
  \label{cor:soundness-den-sem}
  Assume that \(\List M1p\) are terms such that %
  \((\Tseq\Gamma{M_i}A)_{j=1}^p\) and the morphisms %
  \((\Psem{M_i}\Gamma\in\cL(\A\Oc{\Tsem\Gamma},\Tsem A))_{j=1}^p\) are
  summable.
  If \(\sum_{j=1}^pM_j\Rel{\Mslifttr{\Pcfred}}\sum_{j=1}^{p'}M'_{j}\)
  then the morphisms
  \((\Psem{M'_j}\Gamma\in\cL(\A\Oc{\Tsem\Gamma},\Tsem A))_{j=1}^{p'}\)
  are summable and we have %
  \(\sum_{j=1}^p\Psem{M_j}\Gamma=\sum_{j=1}^{p'}\Psem{M'_j}\Gamma\).
\end{corollary}

\subsubsection{Adequacy and determinism}
Now we specialize to the case where \(\cL=\PCOH\).

If \(\Tseq{}MA\) then we know that %
\(\Psem M{}\in\Pcoh{\Tsem A}\subseteq\Realpto{\Web{\Tsem A}}\).
Moreover, a simple inspection of the definition of the semantics shows
that actually \(\Psem M{}\in\Nat^{\Web{\Tsem A}}\).
Of course the situation would be different if the language \(\PCFD\)
were extended with a probabilistic choice operator (or more simply,
\Eg{}, with a ``constant'' \(\mathsf{rand}\) of type \(\Tnat\) which
has probability \(1/2\) to reduce to \(\Num 0\) and \(1/2\) to reduce
to \(\Num1\)), but this is not the case in the present paper and
in~\cite{Ehrhard23b}.

If \(A=\Tnat\), this means that \(\Psem M{}\in\Nat^\Nat\) and that we have
\begin{align*}
  \sum_{\nu\in\Nat}{\Psem M{}}_\nu\in\Intercc 01
\end{align*}
so that \(\forall\nu\in\Nat\ {\Psem M{}}_\nu\in\Set{0,1}\) and there
is at most one \(\nu\in\Nat\) such \({\Psem M{}}_\nu=1\).
In other words, either \(\Psem M{}=0\) or \(\Psem M{}=\Base\nu\) (for
a uniquely determined \(\nu\in\Nat\)).

\begin{theorem}
  Let \(M\) be a term such that \(\Tseq{}M\Tnat\) and let \(\nu\in\Nat\).
  The two following conditions are equivalent.
  \begin{itemize}
  \item \(\Psem M{}=\Base\nu\)
  \item \(\State{\Seqempty}{M}{\Stempty}\Rel{\Mslifttr{\Pcfred}}\Num\nu\).
  \end{itemize}
\end{theorem}
The implication \(\Leftarrow\) boils down to
\Cref{cor:soundness-den-sem} through the translation \(\Termofst e\)
from states to terms and \Cref{th:state-trans-to-term-red}.
The implication \(\Rightarrow\) is proven using an adaptation of the
reducibility method applied to an intersection typing system
associated with a relational semantics of \(\PCFD\) which underlies
the \(\PCOH\) semantics.

So the calculus \(\PCFD\), and its operational semantics formalized by
our Krivine machine, is essentially deterministic in the sense that,
starting from a well typed state \(\State{\Seqempty}{M}{\Stempty}\),
there is at most one reduction path which leads to a final state
\(\State{\Seqempty}{M}{\Stempty}\) where \(\nu\in\Nat\) is uniquely
determined by \(M\) (interpreting the reduction \(e\Rel\Stred e_1+e_2\)
as a nondeterministic choice), the other ones leading to \(0\).
The situation is not completely satisfactory yet since we do not
know \emph{a priori} which transition path is ``the good one''.

Another important contribution of~\cite{Ehrhard23b} is a solution of
this issue based on a simple and natural idea suggested to us by
Guillaume~Geoffroy: make the access word \(\delta\) of a state
\(\State\delta Ns\) writable.

\section*{Conclusion}
We have presented coherent differentiation from a semantical and
syntactical point of view, explaining how this new setting allows to
combine the ideas of differential \LL{} with determinism and with
probabilistic computations.

Even if we consider this as a major improvement wrt.~the earlier
approaches to differential \LL{}, the precise meaning of the resulting
functional calculus is still mysterious.
More recently, in a joint work with Aymeric~Walch, we have extended
this approach to iterated derivatives and to Taylor expansions of
terms, still in a deterministic setting~\cite{EhrhardWalch23b}.
These new results might provide the sought programming interpretation
of CD as it allows to enforce within the language strong restrictions
on the resource consumption of programs.

\section*{Acknowledgment}
I want mainly to thank Aymeric Walch who made many important
observations about the first presentation of CD in~\cite{Ehrhard23a}
of which the present paper has benefited crucially, mainly in
\Cref{sec:partial-monoids,sec:diff-structure}.


\bibliography{newbiblio.bib}


\end{document}